\documentclass[nonacm,sigconf]{acmart}

\usepackage{listings}
\usepackage{algorithm}
\usepackage[noend]{algorithmic}

\usepackage[T1]{fontenc}
\usepackage{url}
\usepackage[hyphenbreaks]{breakurl}
\usepackage[utf8]{inputenc}
\usepackage{enumitem}
\usepackage{xspace}

\usepackage{amssymb}
\usepackage{amsthm}
\usepackage{mathtools}
\usepackage{stmaryrd}
\usepackage[nameinlink]{cleveref}
\usepackage{subcaption}

\newcommand{\eg}{e.\,g., }

\newtoggle{REPORT}
\newtoggle{DEBUG}
\toggletrue{REPORT}
\toggletrue{DEBUG}

\definecolor{reportcolor}{HTML}{188781}
\definecolor{papercolor}{HTML}{536872}

\newcommand{\reportmode}{\toggletrue{REPORT}\togglefalse{DEBUG}}

\newcommand{\reportonly}[1]{\iftoggle{DEBUG}{{\color{reportcolor} (Report) #1}}{\iftoggle{REPORT}{#1}{}}}

\newcommand{\reportorpaper}[2]{\iftoggle{DEBUG}{{\color{reportcolor} (Report) #1} {\color{papercolor} (Paper) #2}}{\iftoggle{REPORT}{#1}{#2}}}

\newtoggle{ANONYMOUS}
\togglefalse{ANONYMOUS}

\newcommand{\authormode}{\togglefalse{ANONYMOUS}}

\newcommand{\authoronly}[1]{\iftoggle{ANONYMOUS}{}{#1}}

\newtoggle{COMMENTS}
\toggletrue{COMMENTS}

\newcommand{\nocommentmode}{\togglefalse{COMMENTS}}

\newcommand{\toggleablecomment}[3]{%
  \iftoggle{COMMENTS}{{\color{#1}{\textbf{#2}: #3}}}{}}

\newcommand{\steffen}[1]{\toggleablecomment{Red}{Steffen}{#1}}

\newcommand{\mref}[2]{\hyperref[#2]{#1~\ref*{#2}}}
\definecolor{ultralightgray}{HTML}{eeeeee}

\usepackage{tikz}
\usetikzlibrary{calc,intersections,through,backgrounds,matrix,patterns}
\usetikzlibrary{arrows, positioning, fit, shapes, backgrounds, patterns,
shapes.misc,decorations.markings}

\definecolor{darkgray}{gray}{0.20}
\definecolor{shaclextendcolor}{HTML}{004bb3}
\definecolor{shacltranslatecolor}{gray}{0.0}
\definecolor{rdfstarextendcolor}{HTML}{00b368}
\definecolor{lightgray}{gray}{0.40}
\definecolor{meta1color}{HTML}{0066cc}
\definecolor{meta2color}{HTML}{cc0066}
\definecolor{typecolor}{gray}{0.55}

\newlength\vgap
\newlength\hgap
\newlength\kvgap

\newcommand{\wheretolook}{extended version\xspace}

\newcommand{\sccq}{q}
\newcommand{\sccqname}{SCCQ\xspace}
\newcommand{\sccquery}[2]{\{ #1 \} \gets \{ #2 \}}
\newcommand{\sccqpattern}{P}
\newcommand{\sccqtemplate}{H}
\newcommand{\sccqformal}{\sccqtemplate \gets \sccqpattern}

\newcommand{\vconcept}[1]{V_{#1}}

\newcommand{\targetquery}{\psi}
\newcommand{\rdftype}{\text{rdf:type}}
\newcommand{\construct}{\texttt{CONSTRUCT}\xspace}
\newcommand{\where}{\texttt{WHERE}\xspace}
\newcommand{\qvar}[1]{\textit{#1}}
\newcommand{\var}{\operatorname{var}}

\newcommand{\allshapes}{\mathcal{S}}
\newcommand{\shapesin}{\mathcal{S}_{\mathrm{in}}}
\newcommand{\shapesout}{\mathcal{S}_{\text{out}}}
\newcommand{\shapesoutopt}{\mathcal{S}_{\text{out-opt}}}
\newcommand{\shapescandidates}{\mathcal{S}_{\text{can}}}

\newcommand{\shaclvalid}[2]{\operatorname{valid}(#1,#2)}
\newcommand{\sparqleval}[2]{\llbracket #1 \rrbracket_{#2}}

\newcommand{\DLogics}{\mathcal{ALCHOI}}

\newcommand{\atomC}[2]{#1{:}#2}
\newcommand{\atomP}[3]{(#1,#2){:}#3}
\newcommand{\ABox}{\mathcal{A}}
\newcommand{\TBox}{\mathcal{T}}
\newcommand{\Int}{\mathcal{I}}
\newcommand{\KB}{\mathcal{K}}

\newcommand{\obj}{\operatorname{ind}}
\newcommand{\nest}{\operatorname{ndep}}

\newcommand{\ConceptNames}{\mathbf{C}}
\newcommand{\IndividualNames}{\mathbf{I}}
\newcommand{\RoleNames}{\mathbf{R}}
\newcommand{\Variables}{\mathbf{V}}

\newcommand{\AllConceptNames}{\bar{\mathbf{C}}}
\newcommand{\AllIndividualNames}{\bar{\mathbf{I}}}
\newcommand{\AllRoleNames}{\bar{\mathbf{R}}}
\newcommand{\AllVariables}{\bar{\mathbf{V}}}

\newcommand{\voc}{\operatorname{voc}}
\newcommand{\graphout}{G_{\mathrm{out}}}
\newcommand{\graphin}{G_{\mathrm{in}}}

\newcommand{\problemLUB}{\textsc{OutputShapes}\xspace}
\newcommand{\problemLUBType}{$(\shapesin, q) \mapsto \shapesout$\xspace}
\newcommand{\problemLUBTypeOpt}{$(\shapesin, q) \mapsto \shapesoutopt$\xspace}
\newcommand{\problemLUBIn}{%
  A finite set of shapes $\shapesin$ and a \sccqname $q$.
}

\newcommand{\problemSE}{\textsc{IsOutputShape}\xspace}
\newcommand{\decision}{\{\textsc{yes}, \textsc{na}\}\xspace}
\newcommand{\problemSEType}{$(\shapesin, q, s) \mapsto \decision$\xspace}

\newcommand{\vkb}{\Sigma_\mathrm{vkb}}
\newcommand{\map}{\Sigma_\mathrm{map}}
\newcommand{\propsub}{\Sigma_\mathrm{prop}}
\newcommand{\ssin}{\Sigma_\mathrm{in}}

\newcommand{\UNA}{\operatorname{UNA}}
\newcommand{\CWA}{\operatorname{CWA}}

\newcommand{\Concept}[1]{\operatorname{C}_{#1}}

\newcommand{\MapAxioms}{\operatorname{MA}}
\newcommand{\MapAxiomsSin}{\operatorname{MA}_{\shapesin}}
\newcommand{\PropertyAxioms}{\operatorname{RS}}

\newcommand{\dom}{\operatorname{dom}}
\newcommand{\ins}{\operatorname{ins}}

\newcommand{\vcg}{\operatorname{vcg}}

\newcommand{\graphmed}{G_{\mathrm{med}}}

\newcommand{\graphext}{G_{\mathrm{ext}}}
\newcommand{\graphvar}{G_{\Variables}}

\newcommand{\himage}[1]{h(#1)}

\theoremstyle{plain}
\newtheorem{proposition}{Proposition}
\newtheorem{lemma}{Lemma}
\newtheorem{theorem}{Theorem}
\newtheorem{corollary}{Corollary}
\newtheorem*{proposition*}{Proposition}
\newtheorem*{lemma*}{Lemma}
\newtheorem*{theorem*}{Theorem}
\newtheorem*{corollary*}{Corollary}

\theoremstyle{definition}
\newtheorem{definition}{Definition}
\newtheorem{example}{Example}
\newtheorem*{definition*}{Definition}
\newtheorem*{example*}{Example}

\newcommand{\anotherexample}[5]{%
    \begin{example}
        With input $q_{#1} = #2$ and $S_{#1} = \{#3\}$, we obtain the set $\{#4\}$ as output.
        
        #5
    \end{example}%
}

\usetikzlibrary{positioning, quotes}
\tikzset{new node/.style={circle,minimum size=1.0em,inner sep=0.0em}}

\SetSymbolFont{stmry}{bold}{U}{stmry}{m}{n}

\definecolor{colorgin}{HTML}{f58231}
\definecolor{colorgmed}{HTML}{800000}
\definecolor{colorgout}{HTML}{e6194B}
\definecolor{colorgvar}{HTML}{4363d8}

\newcommand{\namein}[1]{\textcolor{colorgin}{#1}}
\newcommand{\namemed}[1]{\textcolor{colorgmed}{\dot{#1}}}
\newcommand{\nameout}[1]{\textcolor{colorgout}{\ddot{#1}}}
\newcommand{\nameoutnodots}[1]{\textcolor{colorgout}{#1}}
\newcommand{\namevar}[1]{\textcolor{colorgvar}{#1}}

\newcommand{\eA}{A} %
\newcommand{\eB}{B} %
\newcommand{\eE}{E} %

\newcommand{\ep}{p} %
\newcommand{\er}{r} %

\newcommand{\ea}{a} %
\newcommand{\eb}{b} %
\newcommand{\ee}{e} %

\newcommand{\ex}{\qvar{x}} %
\newcommand{\ew}{\qvar{w}} %
\newcommand{\ey}{\qvar{y}} %
\newcommand{\ez}{\qvar{z}} %

\newcommand{\eS}{S_1}
\newcommand{\eSout}{S_{1\textrm{-out}}}
\newcommand{\esonen}{s_1}
\newcommand{\estwon}{s_2}
\newcommand{\esthreen}{s_3}
\newcommand{\esone}{\eA \sqsubseteq \exists \ep.\eB}
\newcommand{\estwo}{\exists \er.\top \sqsubseteq \eB}
\newcommand{\esthree}{\eB \sqsubseteq \eE}

\newcommand{\eGone}{G_1}
\newcommand{\eGtwo}{G_2}

\newcommand{\eqonen}{q_1}
\newcommand{\eqone}{\sccquery{%
\atomC{\ey}{\eEon},%
\atomC{\ez}{\eBon},%
\atomP{\ey}{\ez}{\epon}%
}{%
\atomP{\ew}{\ey}{\epi},%
\atomC{\ey}{\eBi},%
\atomP{\ex}{\ez}{\epi},%
\atomC{\ez}{\eEi}}%
}
\newcommand{\eqtwon}{q_2}
\newcommand{\eqtwo}{\sccquery{%
\atomP{\ex}{\ey}{\epon},%
\atomC{\ez}{\eAon}%
}{%
\atomP{\ex}{\ey}{\epi},%
\atomC{\ez}{\eAi}}
}

\newcommand{\eAi}{\namein{\eA}}
\newcommand{\eBi}{\namein{\eB}}
\newcommand{\eEi}{\namein{\eE}}

\newcommand{\epi}{\namein{\ep}}

\newcommand{\eri}{\namein{\er}}

\newcommand{\esonei}{\eAi \sqsubseteq \exists \epi.\eBi}
\newcommand{\estwoi}{\exists \eri.\top \sqsubseteq \eBi}
\newcommand{\esthreei}{\eBi \sqsubseteq \eEi}

\newcommand{\eAm}{\namemed{\eA}}
\newcommand{\eBm}{\namemed{\eB}}
\newcommand{\eEm}{\namemed{\eE}}

\newcommand{\epm}{\namemed{\ep}}

\newcommand{\epmi}{\namemed{\ep^-}}

\newcommand{\eAon}{\nameoutnodots{\eA}}
\newcommand{\eBon}{\nameoutnodots{\eB}}
\newcommand{\eEon}{\nameoutnodots{\eE}}

\newcommand{\epon}{\nameoutnodots{\ep}}

\newcommand{\epoin}{\nameoutnodots{\ep^-}}

\newcommand{\eAo}{\nameout{\eA}}
\newcommand{\eBo}{\nameout{\eB}}
\newcommand{\eEo}{\nameout{\eE}}

\newcommand{\epo}{\nameout{\ep}}

\newcommand{\epoi}{\nameout{\ep^-}}

\newcommand{\ewv}{\namevar{\vconcept{w}}}
\newcommand{\exv}{\namevar{\vconcept{x}}}
\newcommand{\eyv}{\namevar{\vconcept{y}}}
\newcommand{\ezv}{\namevar{\vconcept{z}}}

\lstdefinelanguage{SHACL}{
    keywords = {NodeShape,targetClass,property,path,class,minCount,targetSubjectsOf,a}
}

\lstdefinelanguage{SPARQL}{
    keywords = {CONSTRUCT,WHERE,a}
}

\reportmode{} %

\authormode{}  %

\nocommentmode{} %

\begin{document}

\reportorpaper{
  \title[From Shapes to Shapes (Extended Version)]{%
  \texorpdfstring{From Shapes to Shapes: Inferring SHACL Shapes for Results of SPARQL \texttt{CONSTRUCT} Queries (Extended Version)}%
                 {From Shapes to Shapes: Inferring SHACL Shapes for Results of SPARQL CONSTRUCT Queries (Extended Version)}}
}{
  \title[From Shapes to Shapes]{%
  \texorpdfstring{From Shapes to Shapes: Inferring SHACL Shapes for Results of SPARQL \texttt{CONSTRUCT} Queries}%
                 {From Shapes to Shapes: Inferring SHACL Shapes for Results of SPARQL CONSTRUCT Queries}}
}

\authoronly{
  \author{Philipp Seifer}
  \orcid{0000-0002-7421-2060}
  \affiliation{%
    \institution{University of Koblenz}
    \city{Koblenz}
    \country{Germany}
  }
  \email{pseifer@uni-koblenz.de}

  \author{Daniel Hernández}
  \orcid{0000-0002-7896-0875}
  \affiliation{%
    \institution{University of Stuttgart}
    \city{Stuttgart}
    \country{Germany}
  }
  \email{daniel.hernandez@ki.uni-stuttgart.de}

  \author{Ralf Lämmel}
  \orcid{0000-0001-9946-4363}
  \affiliation{%
    \institution{University of Koblenz}
    \city{Koblenz}
    \country{Germany}
  }
  \email{laemmel@uni-koblenz.de}

  \author{Steffen Staab}
  \orcid{0000-0002-0780-4154}
  \affiliation{%
    \institution{University of Stuttgart}
    \city{Stuttgart}
    \country{Germany}
  }
  \affiliation{%
    \institution{University of Southampton}
    \city{Southampton}
    \country{UK}
  }
  \email{steffen.staab@ki.uni-stuttgart.de}
}

\begin{abstract}
    SPARQL \construct{} queries allow for the specification of data processing pipelines that transform given input graphs into new output graphs.
    It is now common to constrain graphs through SHACL shapes allowing users to understand which data they can expect and which not.
    However, it becomes challenging to understand what graph data can be
    expected at the end of a data processing pipeline without knowing the 
    particular input data:
    Shape constraints on the input graph may affect the output graph, but may no longer apply literally, and new shapes may be imposed by the query template.
    In this paper, we study the derivation of shape constraints that hold on all possible output graphs of a given SPARQL \construct{} query.
    We assume that the SPARQL \construct{} query is fixed, e.g., being part of a program, whereas the input graphs adhere to input shape constraints but may otherwise vary over time and, thus, are mostly unknown.
    We study a fragment of SPARQL \construct{} queries (\sccqname) and a fragment of SHACL (Simple SHACL).
    We formally define the problem of deriving the most restrictive set of Simple SHACL shapes that constrain the results from evaluating a \sccqname{} over any input graph restricted by a given set of Simple SHACL shapes.
    We propose and implement an algorithm that statically analyses input SHACL shapes and \construct{} queries and prove its soundness and complexity.
\end{abstract}



\reportonly{%
  \settopmatter{printfolios=true}
}

\maketitle

\section{Introduction}
\label{sec:intro}

Shape description languages like SHACL~\cite{bibshacl} or ProGS~\cite{DBLP:conf/semweb/SeiferLS21} can play two different, but equally important roles. 
\emph{Normatively} they impose schematic constraints on the evoluion of a graph, such that a triple store may automatically reject illegitimate configurations.
Used \emph{informatively}, they aid software developers in understanding graphs, or inform downstream applications, e.g.,~\cite{DBLP:conf/semweb/LeinbergerSSLS19}.
\steffen{Daniel mentioned another downstream application of SHACL having to do with usage, e.g., of functional dependencies. It would be good to cite a second downstream application.}

Graph query languages like SPARQL \construct or G-CORE~\cite{DBLP:conf/sigmod/AnglesABBFGLPPS18} allow for the composition of queries into data processing pipelines.
To execute a given pipeline, the developer must understand what it may output, regardless of its inputs.
Even if the possible inputs are well-described using a shape language like SHACL, it becomes very challenging to understand which shapes apply after one or several querying steps:
SHACL constraints that apply on the input graph may or may no longer apply, \eg existential quantification may become inapplicable because the corresponding relationship might not be part of the \where clause, and new constraints may or may not be imposed by the \construct template.
The developer may hold misconceptions about the structure of the result graph, which might even seem to be endorsed by one particular graph instance, but can lead to errors (\eg when processing query results within a program) for other valid instances of input graphs.

In this paper, we define the problem of computing a set of SHACL shapes characterizing the possible output graphs of a SPARQL \construct query based on (1) the set of shapes applicable to input graphs, and (2) the graph patterns and the template of the query.
We present an algorithm for constructing a sound upper approximation by statically analyzing shapes and query, relying on an encoding in description logics, and \emph{without} referring to any specific input graph.
Thus, our approach allows for investigating data processing pipelines regardless of what (valid) data will be encountered.

\paragraph*{Outline} 
The paper is structured as follows.
\Cref{sec:foundations} introduces foundations, including the fragment of SPARQL queries, SHACL shapes, and description logics we rely upon.
In \Cref{sec:approach} we formalize our validation problem. 
Throughout \Cref{sec:cand}, \Cref{sec:extended}, and \Cref{sec:algorithm} we break the validation problem down into subproblems, and present algorithms for solving them.
We discuss related work in \Cref{sec:related} and conclude in \Cref{sec:conclusion}.
\reportorpaper{%
	For an overview of the appendix, which includes full proofs, extended examples, 
	an overview of our implementation\footnote{\url{https://github.com/softlang/s2s}} including a feasibility experiment thereof, 
  as well as details on how the approach can be generalized to a larger fragment of SHACL, see \Cref{a:structure}.
}{%
	The extended version\footnote{\url{https://arxiv.org/abs/XXXX.XXXXX} (placeholder)} contains full proofs, additional examples,
	an overview of our implementation~\cite{darus-3977_2024} (also available on GitHub\footnote{\url{https://github.com/softlang/s2s}}),
	and discusses how our approach can be generalized to a much larger fragment of SHACL.
}

\section{Foundations}
\label{sec:foundations}

We interpret all RDF classes, instances, and properties as description logic concepts, individuals, and roles, respectively.
For clarity, we always use description logic terminology, e.g., we will refer to ``concept'' rather than ``RDF class''.
We assume that $\ConceptNames\subset\AllConceptNames$, $\IndividualNames\subset\AllIndividualNames$, $\RoleNames\subset\AllRoleNames$,
and $\Variables\subset\AllVariables$ are \emph{finite} subsets of the four infinite, 
pairwise disjoint sets $\AllConceptNames$, $\AllRoleNames$, $\AllIndividualNames$, and $\AllVariables$.
We assume these finite sets to be given as (sufficiently large) inputs to our problem in order to simplify definitions. 
This is not a restriction as their size is arbitrary.
We use $A,B,E\in \ConceptNames$ for description logic \emph{concept names},
$a,b,e\in \IndividualNames$ for description logic \emph{individual names},
$p,r\in \RoleNames$ for description logic \emph{role names},
and $w,x,y,z\in \Variables$ as SPARQL variables.

\subsection[The Description Logic ALCHOI]{The Description Logic $\DLogics$}

We use the description logic $\DLogics$ to define the semantics of RDF graphs
and SHACL shapes following the formalism by \citet{DBLP:conf/lpnmr/BogaertsJB22}.
We next present the standard $\DLogics$ syntax \reportorpaper{and semantics}{and assume standard semantics as} defined in~\citet{DBLP:conf/dlog/2003handbook}.

\begin{definition}[$\DLogics$ concept descriptions]
  $\DLogics$ \emph{concept descriptions} are defined by the following grammar
  \begin{align*}
    C &\Coloneqq
     \top \mid \bot \mid A \mid \neg C \mid \{a\} \mid C \sqcap C \mid C \sqcup C
     \mid \exists \rho.C 
     \mid \forall \rho.C \\
    \rho &\Coloneqq p \mid p^-
  \end{align*}
  where the symbols $\top$ and $\bot$ are two special concept names, and $A$, $a$, and $p$ stand for concept names, individual names, and role names, respectively.
  Given two concept descriptions $C$ and $D$, two individual names $a,b \in \IndividualNames$, and two role descriptions $\rho_1, \rho_2$ (as defined above), $C \sqsubseteq D$ and $\rho_1 \sqsubseteq \rho_2$ are \emph{axioms}, $\atomC{a}{C}$ is a \emph{concept assertion} and $\atomP{a}{b}{p}$ is a \emph{role assertion}.
  We write $C \equiv D$ as an abbreviation for two axioms $C \sqsubseteq D$ and $D \sqsubseteq C$, and likewise for $\rho_1 \equiv \rho_2$.
\end{definition}

\reportonly{
\noindent An $\DLogics$ \emph{knowledge base} $\mathcal{K}$ is a pair $(\TBox,\ABox)$ where $\TBox$ is a finite set of axioms and $\ABox$ is a finite set of assertions.
In a slight abuse of notation, given an ABox $\ABox$, we write $\ABox$ to refer to the knowledge base $(\emptyset, \ABox)$, and given a TBox $\TBox$ we write $\TBox$ to refer to $(\TBox, \emptyset)$.
An \emph{interpretation} $\Int$ is a pair $(\Delta^\Int, \cdot^\Int)$ consisting of a set $\Delta^\Int$, called the \emph{domain}, and
a function $\cdot^\Int$ such that we have
for each individual name $a \in \IndividualNames$, an element $a^\Int \in \Delta^\Int$;
for each concept name $A \in \ConceptNames$, a subset $A^\Int \subseteq \Delta^\Int$; and
for each role name $p \in \RoleNames$, a relation $p^\Int \subseteq \Delta^\Int \times \Delta^\Int$.

The function $\cdot^\Int$ is extended to concept descriptions as follows: $\bot^\Int = \emptyset$, $\top^\Int = \Delta^\Int$, $\{a\}^\Int = \{a^\Int\}$, $(C \sqcap D)^\Int = C^\Int \cap D^\Int$, $(C \sqcup D)^\Int = C^\Int \cup D^\Int$, $(\neg C)^\Int = \top^\Int \setminus C^\Int$,
\begin{equation*}
  \begin{aligned}[t]
    (\exists p.C)^\Int
    &=
      \{ d \in \Delta^\Int \mid (d,e) \in p^\Int\ \text{with}\ e \in C^\Int  \},\\ 
    (\exists p^-.C)^\Int
    &=
      \{ d \in \Delta^\Int \mid (e,d) \in p^\Int\ \text{with}\ e \in C^\Int\},\\
    (\forall p.C)^\Int
    &=
      \{ d \in \Delta^\Int \mid \text{for all } e \in \Delta^\Int,
      \\ & \qquad \qquad \qquad \text{if } (d,e) \in p^\Int \text{ then } e \in C^\Int \},\\
    (\forall p^-.C)^\Int
    &=
      \{ d \in \Delta^\Int \mid \text{for all } e \in \Delta^\Int,
      \\ & \qquad \qquad \qquad \text{if } (e,d) \in p^\Int \text{ then } e \in C^\Int \}.
  \end{aligned}
\end{equation*}

An interpretation $\Int$ is a \emph{model} of a knowledge base $\KB = (\TBox,\ABox)$ if and only if $C^\Int \subseteq D^\Int$ for every axiom $C \sqsubseteq D$ in $\TBox$, $p^\Int \subseteq r^\Int$ for every axiom $p \sqsubseteq r$ in $\TBox$, $a^\Int \in C^\Int$ for every assertion $\atomC{a}{C}$ in $\ABox$, and $(a^\Int, b^\Int) \in p^\Int$ for every assertion $\atomP{a}{b}{p}$ in $\ABox$. 
Given two knowledge bases $\KB_1$ and $\KB_2$, $\KB_1$ \emph{entails} $\KB_2$, denoted $\KB_1 \models \KB_2$, if and only if every model $\Int$ of $\KB_1$ is also a model of $\KB_2$.
}

\subsection{Simple RDF Graphs}
\label{sec:simple-rdf-graphs}

According to the RDF specification~\cite{rdf}, an RDF graph is a finite set of triples whose elements belong to three pairwise disjoint sets: IRIs, blank nodes, and literals. 
For convenience, we assume the fragment of RDF graphs, called \emph{Simple RDF graphs}, that only considers triples whose elements are IRIs. 
Furthermore, we assume that IRIs are partitioned in the four sets $\ConceptNames$, $\IndividualNames$, $\RoleNames$, and $\{\rdftype\}$, and that an RDF triple has either the form $(a,p,b)$ or $(a,\rdftype,A)$ where $a,b \in \IndividualNames$, $A \in \ConceptNames$, and $p \in \RoleNames$.
We interpret each triple $(a,p,b)$ as an assertion $\atomP{a}{b}{p}$, and each triple $(a,\rdftype,A)$ as an assertion $\atomC{a}{A}$.
With these assumptions we define Simple RDF graphs, and introduce running examples in \Cref{fig:running}.

\begin{definition}[Simple RDF Graph Syntax]
  A \emph{Simple RDF graph} (or just \emph{graph}) is an $\DLogics$ ABox $G$ where the concept description of each concept assertion in $G$ is a concept name $A \in \ConceptNames$.
\end{definition}

\citet{DBLP:conf/lpnmr/BogaertsJB22} highlight that RDF graphs have two different semantics, depending on the inference task we want to perform: 
If the task is \emph{deduction}, the semantics of a graph is given by an ABox, and following the no-unique-name, no-domain-closure and open-world assumptions.
If the task is \emph{validation}, the semantics is given by a model.
Instead of relying on a model-theoretic semantics for validation, our approach benefits from a proof-theoretic semantics.
As \citet{DBLP:conf/db-workshops/Reiter82} suggests, the model-theoretic semantics of databases can be defined in proof theoretic terms:
A database can be seen as a set of formulas instead of a model, where queries are formulae to be proven, and satisfaction of constraints is defined in terms of consistency.
We can therefore extend the deduction semantics of Simple RDF graphs with axioms that encode these assumptions, which are based on the proof-theoretic semantics for relational databases by \citet{DBLP:conf/db-workshops/Reiter82}.

\Cref{prop:equivalent-formalisms} below implies the equivalence of the \citet{DBLP:conf/lpnmr/BogaertsJB22} model-theoretic SHACL semantics (\Cref{def:canonical-interpretation}) and our proof theoretic SHACL semantics (\Cref{def:validation-rdf-semantics}).

\begin{definition}[Graph Interpretation \citep{DBLP:conf/lpnmr/BogaertsJB22}]
  \label{def:canonical-interpretation}
  The \emph{canonical interpretation} of a Simple RDF graph $G$ is the interpretation $\Int_G$ such that $\Delta^{\Int_G} = \IndividualNames$; for each $a \in \IndividualNames$, $a^{\Int_G} = a$; for every concept name $A \in \ConceptNames$, $A^{\Int_G} = \{ a \mid \atomC{a}{A} \in G \}$; and for every role name $r \in \RoleNames$, $r^{\Int_G} = \{(a,b) \mid \atomP{a}{b}{r} \in G\}$. A graph $G$ is \emph{model-valid} according to a set $\Sigma$ of $\DLogics$ axioms if and only if $\Int_G$ is a model of $\Sigma$.
\end{definition}

\begin{definition}[Simple RDF Graph Validation Semantics]%
  \label{def:validation-rdf-semantics}
  The axioms of a Simple RDF graph $G$, denoted $\TBox_G$, are the TBox consisting of the following $\DLogics$ axioms:
  \begin{enumerate}
  \item \emph{Domain Closure Assumption (DCA)}:
    $\top \equiv \bigsqcup_{a \in \IndividualNames}\{a\}$.
  \item \emph{Unique Name Assumption (UNA)}:
   $\{a\} \sqcap \{b\} \equiv \bot$, for each pair of distinct individual names $a, b \in \IndividualNames$.
  \item \emph{Closed-World Assumption (CWA)}: 
    \begin{itemize}
      \item $A \equiv \bigsqcup_{\atomC{a}{A}\in G}\{a\}$, for each concept name $A \in \ConceptNames$, 
      \item $\exists p.\{a\} \equiv \bigsqcup_{\atomP{b}{a}{p}\in G} \{b\}$, and
      \item $\exists p^-.\{a\} \equiv \bigsqcup_{\atomP{a}{b}{p}\in G} \{b\}$, for each role name $p \in \RoleNames$ and each individual name $a \in \IndividualNames$.
    \end{itemize}
  \end{enumerate}
  $(\TBox_G, G)$ is the \emph{validation knowledge base} of $G$. 
  A graph $G$ is \emph{proof-valid} according to a set $\Sigma$ of $\DLogics$ axioms if and only if $\Sigma$ is consistent with the validation knowledge base of $G$ (i.e., the knowledge base $(\TBox_G \cup \Sigma, G)$ admits a model).
\end{definition}

\begin{proposition}\label{prop:equivalent-formalisms}
  For a graph $G$ and set of $\DLogics$ axioms $\Sigma$, the following statements are equivalent: (i) $G$ is model-valid according to $\Sigma$, (ii) $G$ is proof-valid according to $\Sigma$, and (iii) $G$ is proof-valid according to $\{\varphi\}$ for every $\varphi \in \Sigma$.
\end{proposition}

\begin{figure}[t]
\centering
\begin{subfigure}[b]{0.4\columnwidth}
  \centering
  \begin{tikzpicture}[
    every label/.style={inner sep=0.0em}
]
    \node[new node,label={180:$\eAi$}] (a) {$\ea$};
    \node[new node,label={0:$\eBi,\eEi$}] (b) [right = 0.8cm of a]  {$\eb$};
    \path[->]
        (a) edge [bend left=20] node [above] {$\epi$} (b)
        (b) edge [bend left=20] node [below] {$\epi,\eri$} (a);
\end{tikzpicture}
  \caption{$\eGone$}
  \label{fig:rone}
\end{subfigure}%
\begin{subfigure}[b]{0.6\columnwidth}
  \centering
  \begin{tikzpicture}[
    every label/.style={inner sep=0.0em}
]
    \node[new node,label={[label distance=0.2em]90:$\eAi$}] (a) {$\ea$};
    \node[new node,label={[label distance=0.2em]90:$\eBi,\eEi$}] (b) [right = 0.8cm of a]  {$\eb$};
    \node[new node,label={[label distance=0.2em]90:$\eEi$}] (e) [left = 0.8cm of a]  {$\ee$};
    \path[->]
        (a) edge [] node [below] {$\epi$} (b)
        (a) edge [] node [below] {$\epi$} (e);
\end{tikzpicture}
  \caption{$\eGtwo$}
  \label{fig:rtwo}
\end{subfigure}%
\caption{Two example graphs, where we visualize $\rdftype$ edges as floating labels next to nodes (e.g., $\eAi \,\, \ea$ for $\atomC{\ea}{\eAi}$).}
\label{fig:running}
\end{figure}
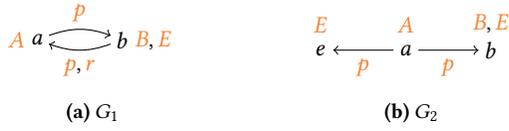

\subsection{Simple SHACL Shapes}

Following the idea that a SHACL schema is a description logic TBox~\cite{DBLP:conf/lpnmr/BogaertsJB22}, a SHACL shape is an axiom of the form $\targetquery \sqsubseteq \phi$ where $\targetquery$ and $\phi$ are concept descriptions, called the \emph{target query} and the \emph{shape constraint}, respectively.
We restrict $\DLogics$ axioms to an essential subset for the sake of simplification, as defined below.
We lift this restriction in the \wheretolook.

\begin{definition}[Simple SHACL Syntax]%
  \label{def:siple-shacl-syntax}
  A \emph{Simple SHACL shape} (or just a \emph{shape}) is an $\DLogics$ axiom $\targetquery \sqsubseteq \phi$ such that the concept expressions $\targetquery$ and $\phi$ are defined by:
  \begin{equation*}
    \begin{aligned}[t]
    \targetquery \Coloneqq A \mid \exists \rho.\top \qquad\qquad \phi \Coloneqq A \mid \exists \rho.A \mid \forall \rho.A 
    \end{aligned}
  \end{equation*}
  
  A \emph{Simple SHACL schema} $\allshapes$ is an $\DLogics$ TBox that consists of a finite set of Simple SHACL shapes.
\end{definition}

Given that shapes are defined in terms of $\DLogics$ axioms, their semantics is defined in terms of the semantics of $\DLogics$ axioms over the validation knowledge base of a graph.

\begin{definition}[Simple SHACL Semantics]
  A graph $G$ is \emph{valid} for a set $S$ of Simple SHACL shapes, denoted $\shaclvalid{G}{S}$, if and only if $G$ is proof-valid according to $S$.
\end{definition}

\begin{example} 
  \label{ex:sone}
  Consider the set of shapes $\eS = \{\esonen, \estwon, \esthreen\}$ where $\esonen = \esonei$, $\estwon = \estwoi$ and $\esthreen = \esthreei$.
  Shape $\esonei$, for example, targets all individuals that are instances of $\eAi$, and requires that there exists at least one edge $\epi$ to a $\eBi$.
  Both graphs in \Cref{fig:running} are \emph{valid} with respect to $\eS$.
\end{example}

\subsection{Simple Conjunctive \construct Queries}
\label{sec:sparql-queries}

This section defines the SPARQL fragment we consider, called \emph{Simple Conjunctive \construct Queries} (\sccqname, or just \emph{queries}). 
This fragment follows the semantics proposed by Kostylev
et al.~\cite{DBLP:conf/icdt/KostylevRU15} and is restricted to basic
graph patterns generated by adding variables for individual names on Simple RDF graphs.

\begin{definition}[\sccqname Syntax]
  An \emph{atomic pattern} $t$ is defined by the following grammar:
  \[
    t \Coloneqq \atomC{a}{A} \mid \atomC{x}{A} \mid \atomP{a}{b}{p} \mid \atomP{x}{a}{p} \mid \atomP{a}{x}{p} \mid \atomP{x}{y}{p}
  \]
  where $A$ stands for concept names, $a$ and $b$ for individual names, $p$ for role names, and $x$ and $y$ for variables. 
  A finite set of atomic patterns is a \emph{simple graph pattern}.
  Given a simple graph pattern $\sccqpattern$, we write $\var(\sccqpattern)$ and $\obj(\sccqpattern)$ to denote the respective sets of variables and individual names occurring in pattern $\sccqpattern$. 
  Given two simple graph patterns $\sccqpattern$ and $\sccqtemplate$, where $\var(\sccqtemplate) \subseteq \var(\sccqpattern)$, the expression $\sccqformal$ is a \sccqname, where $\sccqtemplate$ and $\sccqpattern$ are called the \emph{template} and the \emph{pattern} of the query, respectively.
\end{definition}

A \emph{valuation} of a simple graph pattern $\sccqpattern$ is a function $\mu:\Variables\cup\IndividualNames \to \IndividualNames$ such that $\mu(a)=a$ for every $a \in \IndividualNames$.
In a slight abuse of notation, given two elements $u,v \in \Variables \cup \IndividualNames$ and a simple graph pattern $P$, we write $\mu(\atomC{u}{A}) = \atomC{\mu(u)}{A}$, $\mu(\atomP{u}{v}{p}) = \atomP{\mu(u)}{\mu(v)}{p}$, and $\mu(\sccqpattern)=\{\mu(t) \mid t \in \sccqpattern\}$. 
Intuitively, a valuation substitutes variables in a pattern by individual names. The semantics of \sccqname is defined below.

\begin{definition}[\sccqname Semantics]
  The result of evaluating a \sccqname $\sccqformal$ over a Simple RDF graph $G$ is the Simple RDF graph, denoted $\sparqleval{\sccqformal}{G}$, defined as follows:
  \[
    \sparqleval{\sccqformal}{G} = \bigcup_{\mu(\sccqpattern) \subseteq G} \mu(\sccqtemplate).
  \]
\end{definition}

Intuitively, the pattern $P$ retrieves valuations $\mu$ such that $\mu(P)$ is a subgraph of $G$, which are used to generate the output graph by replacing variables in the template.

\begin{example}
  \label{ex:qone}
  Let $\eqonen = \sccqformal = $
  \[\eqone\]
  For evaluation over the first example graph, $\sparqleval{\eqonen}{\eGone}$, we need to find valuations $\mu$ where $\mu(P) \subseteq \eGone$.
  This holds for $\mu$ where $\mu(\ew) = \ea$, $\mu(\ex) = \ea$, $\mu(\ey) = \eb$, and $\mu(\ez) = \eb$.
  Hence, the result is the graph $
    \sparqleval{\sccqformal}{\eGone}
    = \mu(H)
    = \{\atomC{\eb}{\eEon},\atomC{\eb}{\eBon},\atomP{\eb}{\eb}{\epon}\}
  $.
  Similarly, evaluation $\sparqleval{\eqonen}{\eGtwo} = \{\atomC{\eb}{\eEon},\atomC{\eb}{\eBon},\atomC{\ee}{\eBon},\atomP{\eb}{\eb}{\epon},\atomP{\eb}{\ee}{\epon}\}$.
\end{example}

\section{Formal Problem Statement}
\label{sec:approach}

We aim to construct shapes characterizing the possible result graphs of a query where the input is constrained by shapes as well.

\begin{definition}[Input and Output Graph]%
  \label{def:query-graphs}
  A graph $\graphin$ is an \emph{input graph} with respect to a finite set of shapes $\shapesin$ if $\shaclvalid{\graphin}{\shapesin}$. 
  A graph $\graphout$ is an \emph{output graph} for a query $q$ and a finite set of shapes $\shapesin$ if there exists an input graph $\graphin$ such that $\graphout = \sparqleval{q}{\graphin}$.
\end{definition}

\begin{definition}[Vocabulary]%
  \label{def:shape-voc}
  A \emph{vocabulary} is the set of concept and role names that occur in a concept description $C$, shape $s$, graph $G$, or template of a query $\sccq$, denoted $\voc(C)$, $\voc(s)$, $\voc(G)$, or $\voc(\sccq)$, respectively.
\end{definition}

\begin{definition}[Relevancy]%
  \label{def:irrelevant-shape}
  Shape $s = \targetquery \sqsubseteq \phi$ is \emph{relevant} for query $\sccq$ if there exists a graph $G_+$ with $\voc(G_+) \subseteq \voc(q)$ such that $\shaclvalid{G_+}{\{s\}}$ and $(\TBox_{G_+},G_+) \not\models \psi \sqsubseteq \bot$, 
  and a graph $G_-$ with $\voc(G_-) \subseteq \voc(q)$ such that not $\shaclvalid{G_-}{\{s\}}$.
\end{definition}

\makeatletter
\renewcommand*{\ALG@name}{Problem}

\makeatother

Problem \problemLUB formalizes the set of shapes that best characterize the possible output graphs of a \sccqname.
The first restriction on the solution ensures only relevant shapes are in the output, 
i.e., shapes that validate some graphs in the vocabulary $\voc(q)$, but not all of them (\Cref{def:irrelevant-shape}).
This excludes, for example, shapes with targets outside the vocabulary (which are thereby vacuously satisfied), or shapes with constraints requiring concept or role names outside the vocabulary, which can never be satisfied.
The second restriction states that $\shapesoutopt$ defines an upper bound for the set of output graphs, while the third requires this upper bound to be minimal.
Later in this paper, we will present a sound, but not complete, algorithm for solving this problem, i.e., an algorithm that satisfies the first two, but not the last condition.
{
\renewcommand{\thealgorithm}{} %
\begin{algorithm}[ht]
  \caption{\problemLUB $:$ \problemLUBTypeOpt}
  \label{p:lub}
  \begin{algorithmic}[1]
      \REQUIRE \problemLUBIn
      \ENSURE A set of shapes $\shapesoutopt$ such that:\\
      1. every $s \in \shapesoutopt$ is relevant for $q$,\\
      2. for every $G$ with $\shaclvalid{G}{\shapesin}$ and $\graphout = \sparqleval{q}{G}$, $\shaclvalid{\graphout}{\shapesoutopt}$,\\
      3. the set of graphs $G$ such that $\shaclvalid{G}{\shapesoutopt}$ is minimal.
  \end{algorithmic}
\end{algorithm}
}

\begin{example}
  \label{ex:formalrunning}
  Consider $\eqonen$ (\Cref{ex:qone}) and $\eS$ (\Cref{ex:sone}).
  The shapes $\eEon \sqsubseteq \exists \epon . \eBon, \eEon \sqsubseteq \eBon \in \eSout$ constrain the results of evaluating $\eqonen$ on \emph{any} graph that is valid with respect to $\eS$, e.g., the example graphs in \Cref{fig:running}.
  Shape $\eEon \sqsubseteq \exists \epon . \eBon \in \eSout$ follows directly from the query template,
  whereas shape $\eEon \sqsubseteq \eBon$ is only contained in $\eSout$ because $\esthreei$ holds on all input graphs and we can thus infer that all bindings for $\ey$ are also bindings for $\ez$.
\end{example}

Simple SHACL shapes are not sufficiently expressive to rule out \emph{all} impossible output graphs of a query.
For example, we know for $\eqonen$ and $\eS$ that each instance of $\eEon$ has a $\epon$ edge to itself.
Simple SHACL shapes cannot express reflexivity, so graphs without reflexive $\epon$ cannot be ruled out.

\section{Computing Candidate \problemLUB}
\label{sec:cand}

We break down Problem~\problemLUB into two subproblems: 
The generation of a finite set of candidate shapes $\shapescandidates$ -- a superset of the solution -- and
the filtering of this set (Problem~\problemSE).

\begin{algorithm}[ht]
  \caption{\problemSE $:$ \problemSEType}
  \label{p:main}
  \begin{algorithmic}[1]
      \REQUIRE A finite set of shapes $\shapesin$, a \sccqname $\sccq = \sccqformal$, and a shape $s$ that is relevant for this query $q$.
      \ENSURE Does $\shaclvalid{\sparqleval{q}{\graphin}}{\{s\}}$ hold for every graph $\graphin$ where $\shaclvalid{\graphin}{\shapesin}$? 
  \end{algorithmic}
\end{algorithm}

\makeatletter
\renewcommand*{\ALG@name}{Algorithm}
\makeatother
\setcounter{algorithm}{0}

\Cref{alg:outshapes} outlines this approach, by referring to Problem $\problemSE$.
In \Cref{sec:algorithm}, we will define a sound, but not complete, algorithm solving this problem (\Cref{alg:main}).
Thus, \Cref{alg:outshapes} is a sound approximation of problem \problemLUB satisfying its first two, but not the third condition (minimality).
In the following we use $\shapesout$ to refer to such an approximation of $\shapesoutopt$.

\begin{algorithm}[h]
  \caption{\problemLUB $:$ \problemLUBType}
  \label{alg:outshapes}
  \begin{algorithmic}[1]
    \REQUIRE \problemLUBIn
    \ENSURE The set of output shapes $\shapesout$.
    \STATE $\shapesout \gets \emptyset$, $\shapescandidates \gets$ the finite set of shapes over $\voc(q)$
    \FORALL{$s \in \shapescandidates$}
    \IF{$\problemSE(\shapesin, q, s) = \textsc{YES}$}
    \STATE $\shapesout \gets \shapesout\; \cup \{s\}$
    \ENDIF
    \ENDFOR
    \STATE \textbf{return} $\shapesout$
  \end{algorithmic}
\end{algorithm}

In order to obtain a finite set of candidates $\shapescandidates$, \Cref{prop:relevant-shapes} allows us to discard shapes that do not describe output graphs and
limit thus the search space of \Cref{alg:outshapes} to the shapes that are built from the vocabulary of the query.

\begin{proposition}\label{prop:relevant-shapes}
  If a shape $s = \psi \sqsubseteq \phi$ is relevant for a \sccqname $q$, then $s$ satisfies one of the following two conditions: 
    (i) $\voc(s) \subseteq \voc(q)$, or
    (ii) $\voc(\psi) \subseteq \voc(q)$ and $\phi$ is either $\forall p.A$ or $\forall p^-.A$ where $p \in \voc(q)$ and $A \notin \voc(q)$.
\end{proposition}

In order to cover all relevant shapes that satisfy condition (i), we can include the finite combinations of elements in the vocabulary of the query.
Condition (ii) requires special care:
Each role name $p \in \voc(q)$ defines a family of shapes of the form $\psi \sqsubseteq \forall p.A$ or $\psi \sqsubseteq \forall p^-.A$, where $A \notin \voc(q)$.
To explore this family, it suffices to consider a representative by including in the set of candidate shape constraints for each role name $p \in \voc(q)$ the two concept descriptions $\forall p.A$ and $\forall p^-.A$, such that $A \not\in \voc(q)$.

The search space is therefore bounded by the vocabulary of the query, which is relatively small.
In the \wheretolook we show that there are $(n + 2m)(n + 4nm + 2m) - n$ candidate shapes if $\voc(q)$ contains $n$ concept names and $m$ role names.

\section{Axiomatizations Over Executions}
\label{sec:extended}

A query $\sccq = \sccqformal$ works on any input graph $\graphin$ defined by $\shapesin$ (\Cref{def:query-graphs}) and returns a result graph $\graphout$ in two steps:
By matching $\sccqpattern$ with $\graphin$, determining valuations $\mu$ where $\mu(\sccqpattern) \subseteq \graphin$,
and then by replacing variables in $\sccqtemplate$ with these valuations producing $\graphout$.
As a result, multiple occurrences of the same concept names do not have the same extensions: As an example, consider the query 
$\sccquery{%
\atomC{\ex}{\eAo},%
\atomC{\ey}{\eAo},%
\atomC{\ey}{\eEo}%
}{%
\atomC{\ex}{\eAm},%
\atomC{\ex}{\eBm},%
\atomC{\ey}{\eEm}%
}$
and the input shape $\eAi \sqsubseteq \eEi$, where we mark different occurrences
of the same names with zero, one, or two dots.
The extension of $\eAm$ includes only individuals matched by the query pattern, which requires $\eBm$ as well. Thus, $\eAm$ may be subsumed by and may be unequal to $\eAi$. 
Similarly, $\eAo$ now also includes all bindings of variable $\ey$, unlike $\eAi$ or $\eAm$.

We now want to axiomatize how all possible $\graphin$ are connected with their corresponding $\graphout$.
Virtually putting these axiomatizations together creates an \emph{extended graph} that holds axioms from these two steps allowing us to prove statements about $\graphout$.
We distinguish inputs and step outcomes by a syntactic trick that rewrites input symbols $A,p$ into \emph{fresh} symbols $\dot{A},\dot{p}$ after the first step, and into $\ddot{A},\ddot{p}$ after the second step.
We also write, e.g., $\dot{G}$, meaning substitution of all symbols $A,p$ in graph $G$ with $\dot{A}, \dot{p}$.

These rewritten symbols allow us to encode assertions that are valid for only specific states of query execution.
Variable bindings, on the other hand, hold throughout:
We codify a variable binding $\mu(x)=a$ as a concept assertion $\atomC{a}{\vconcept{x}}$, where $\vconcept{x}$ is a fresh concept name.
Note, that we assume that all concept names and role names with dots, as well as concept names for variable concepts, exist as fresh names in $\ConceptNames$ and $\RoleNames$.
\Cref{ex:ext-graph} illustrates the construction of such an extended graph, which is defined in \Cref{def:extended-graph}.

\begin{definition}[Extended Graph]\label{def:extended-graph}
  Given an input graph $\graphin$ and a query $\sccqformal$, the following graphs are defined with correspondences to the query execution steps:
  \begin{enumerate}
  \item The \emph{intermediate graph} $\graphmed \coloneqq \bigcup_{\mu(P) \subseteq \graphin} \mu(P)$.
  \item The \emph{variable concept graph} $\graphvar$ containing an assertion $\atomC{a}{\vconcept{x}}$ if and only if there exists a valuation $\mu$ such that $\mu(P) \subseteq \graphin$ and $\mu(x)=a$.
  \item The \emph{output graph} $\graphout \coloneqq \sparqleval{q}{\graphin}$.
  \item The \emph{extended graph} $\graphext \coloneqq \graphin \cup \dot{G}_{\mathrm{med}} \cup \graphvar \cup \ddot{G}_{\mathrm{out}}$.
  \end{enumerate}
\end{definition}

\begin{example}\label{ex:ext-graph}
  Consider $\eqonen$ (\Cref{ex:qone}), $\eS$ (\Cref{ex:sone}), and the graph $\eGone$ (\Cref{fig:running}) as one possible input graph for $\eqonen$.
  The respective extended graph and its components are given in \Cref{fig:gext}.
  Note, that these graphs satisfy different axioms (in different namespaces), e.g., $\exists \epm^-.\top \sqsubseteq \eEm$ is valid in $\dot{G}_{\mathrm{med}}$ but $\exists \epi^-.\top \sqsubseteq \eEi$ is not valid in $\graphin$.
  A range of axioms are valid for $\graphext$, such as $\eEm \sqsubseteq \eEi$ or $\eyv \sqsubseteq \ezv$.
  Indeed, these axioms are valid on \emph{every} extended graph of $\eqonen$, as long as $\shaclvalid{\graphin}{\eS}$, e.g., $\graphin = \eGone$ or $\graphin = \eGtwo$ (\Cref{fig:running}).

  \begin{figure*}[t]
    \centering
    \begin{subfigure}[t]{.30\textwidth}
      \centering
      \begin{tikzpicture}[
    every label/.style={inner sep=0.0em}
]
    \node[new node,label={180:$\eAi,\ewv,\exv$}] (a) {$\ea$};
    \node[new node,label={[align=left]0:$\eBi,\eBm,\eBo,\eyv,$ \\%
        $\eEi,\eEm,\eEo,\ezv$}] (b) [right = 1.1cm of a]  {$\eb$};
    \path[->]
        (a) edge [bend left=60] node [above] {$\epi,\epm$} (b)
        (b) edge [bend left=60] node [below] {$\epi,\eri$} (a)
        (b) edge [loop left=20] node [left] {$\epo$} (b);
\end{tikzpicture}
    \end{subfigure}%
    \hfill%
    \begin{subfigure}[b]{.18\textwidth}
      \centering
      \begin{tikzpicture}[
    every label/.style={inner sep=0.0em}
]
    \node[new node,label={180:$\eAi$}] (a) {$\ea$};
    \node[new node,label={0:$\eBi,\eEi$}] (b) [right = 0.8cm of a]  {$\eb$};
    \path[->]
        (a) edge [bend left=20] node [above] {$\epi$} (b)
        (b) edge [bend left=20] node [below] {$\epi,\eri$} (a);
\end{tikzpicture}
      \caption{$\graphin$}
      \label{fig:gin}
    \end{subfigure}%
    \begin{subfigure}[b]{.18\textwidth}
      \centering
      \begin{tikzpicture}[
    every label/.style={inner sep=0.0em}
]
    \node[new node] (a) {$\ea$};
    \node[new node,label={0:$\eBm,\eEm$}] (b) [right = 0.8cm of a]  {$\eb$};
    \path[->]
        (a) edge [bend left=20] node [above] {$\epm$} (b)
        (b) edge [white,bend left=13] node [below] {$\ep,\er$} (a);
\end{tikzpicture}
      \caption{$\dot{G}_{\mathrm{med}}$}
      \label{fig:gmed}
    \end{subfigure}%
    \begin{subfigure}[b]{.13\textwidth}
      \centering
      \begin{tikzpicture}[
    every label/.style={inner sep=0.0em}
]
    \node[new node,label={0:$\ewv,\exv$}] (a) {$\ea$};
    \node[new node] (c) [below = 0.4cm of a] {};
    \node[new node,label={0:$\eyv,\ezv$}] (b) [below = 0.2cm of a]  {$\eb$};
\end{tikzpicture}
      \caption{$\graphvar$}
      \label{fig:gvar}
    \end{subfigure}%
    \begin{subfigure}[b]{.14\textwidth}
      \centering
      \begin{tikzpicture}[
    every label/.style={inner sep=0.0em}
]
    \node[new node,label={[label distance=0.2em]0:$\eBo,\eEo$}] (b) {$\eb$};
    \node[new node,white] (a) [below = 0.0cm of b]  {};
    \path[->]
        (b) edge [loop left=20] node [left] {$\epo$} (b);
\end{tikzpicture}
      \caption{$\ddot{G}_{\mathrm{out}}$}
      \label{fig:gout}
    \end{subfigure}
    \caption{On the left the graph $\graphext$ as the union of $\namein{\graphin}$ (\subref{fig:gin}), $\textcolor{colorgmed}{\dot{G}_{\mathrm{med}}}$ (\subref{fig:gmed}), $\namevar{\graphvar}$ (\subref{fig:gvar}), and $\textcolor{colorgout}{\ddot{G}_{\mathrm{out}}}$ (\subref{fig:gout}).}
    \label{fig:gext}
  \end{figure*}
\end{example}

Assertions added per step are sound, but not sufficient to fully characterize what happens at each query execution step.
Therefore, axioms we can find to characterize the relationships between $\graphin$ and $\graphout$ will be sound but incomplete.
In the following sections, we will introduce additional axioms per step to extend possible inferences and thus determine a tighter description by output shapes.

\Cref{prop:reduction-extended-graph} shows that axioms valid on any of the graphs  $\graphin$, $\graphmed$, $\graphvar$ and $\graphout$ are valid on the extended graph when applying syntactic rewriting, and vice versa.

\begin{proposition}\label{prop:reduction-extended-graph}
  Given a graph $\graphin$ and a query $\sccq$,
  let the graphs $\graphmed$, $\graphout$, and $\graphext$ be defined according to Definition~\ref{def:extended-graph}.
  For every axiom $\varphi$ that does not include names with dots (e.g., $\dot{A}$, $\ddot{A}$, $\dot{p}$, $\ddot{p}$),   the following equivalences hold:
  \begin{enumerate}
  \item
    $\shaclvalid{\graphin}{\{\varphi\}}$ if and only if $\shaclvalid{\graphext}{\{\varphi\}}$.
  \item
    $\shaclvalid{\graphmed}{\{\varphi\}}$ if and only if $\shaclvalid{\graphext}{\{\dot{\varphi}\}}$.
  \item
    $\shaclvalid{\graphout}{\{\varphi\}}$ if and only if $\shaclvalid{\graphext}{\{\ddot{\varphi}\}}$.
  \end{enumerate}
\end{proposition}

\section{Checking Whether \problemSE}
\label{sec:algorithm}

\Cref{alg:main} (\problemSE) checks for a given shape $s$ if the rewritten shape $\ddot{s}$ is entailed by a set of axioms $\Sigma$ valid for every extended graph $\graphext$ and derived from $\sccq$ and $\shapesin$ ($\textsc{YES}$). 
If this entailment can not be proven, the algorithm returns no answer ($\textsc{NA}$).
In the remainder of this section, we construct $\Sigma$: 
We start by inferring the assumptions of the validation knowledge base of $\graphext$ based on the atoms of the input query ($\vkb$).
Next, we identify subsumptions between query variables in different components of the input query by establishing a mapping between them ($\map$).
Finally, we include subsumptions between role names by considering the query variables constraining them ($\propsub$).
In the \wheretolook we prove that Problem \problemSE is NP-hard.

\Cref{cor:reduction-algo1} establishes the formal foundation for \problemSE based on \Cref{prop:reduction-extended-graph}.

\begin{corollary}\label{cor:reduction-algo1}
  Let $\sccq$ be a \sccqname, $\Sigma$ a set of $\DLogics$ axioms such that $\shaclvalid{\graphext}{\Sigma}$ for every extended graph $\graphext$ of $\sccq$, and $s$ a shape including no names with dots.
  If $\Sigma \models \ddot{s}$, then $\shaclvalid{\graphout}{\{s\}}$ for every output graph $\graphout$ of $\sccq$ .
\end{corollary}

\begin{algorithm}[ht]
  \caption{\problemSE $:$ \problemSEType}
  \label{alg:main}
  \begin{algorithmic}[1]
      \REQUIRE A finite set of shapes $\shapesin$, a \sccqname $\sccq = \sccqformal$, and a shape $s$ that is relevant for this query $q$.
      \ENSURE Does $\shaclvalid{\sparqleval{q}{\graphin}}{\{s\}}$ hold for every graph $\graphin$ where $\shaclvalid{\graphin}{\shapesin}$? 
      \STATE $\ssin \gets \shapesin$
      \STATE $\vkb \gets \UNA(\sccq) \cup \CWA(\sccq)$\hfill\emph{(\Cref{ss:una})}\hspace{1em}\phantom{.}
      \STATE $\map \gets \MapAxiomsSin(P)$\hfill\emph{(\Cref{ss:components} and \ref{ss:extend})}\hspace{1em}\phantom{.}
      \STATE $\propsub \gets \PropertyAxioms(q)$\hfill \emph{(\Cref{ss:properties})}\hspace{1em}\phantom{.}
      \STATE $\Sigma \gets \ssin \cup \vkb \cup \map \cup \propsub$
      \RETURN \textbf{if} $\Sigma \models \ddot{s}$ \textbf{then} $\textsc{yes}$ \textbf{else} $\textsc{na}$
  \end{algorithmic}
\end{algorithm}

\subsection{Axiomatizations from the Validation KB}\label{ss:una}

We first utilize the assumptions of the validation knowledge base (see \Cref{def:validation-rdf-semantics}) to infer axioms from a query $\sccq$
that are valid on \emph{any} extended graph of $\sccq$.
Since we do not know all individual names in the extended graphs, we limit the UNA-encoding to individual names that appear in the query (\Cref{def:una}), 
which are in any non-empty extended graph per definition (see \Cref{def:extended-graph}).

\begin{definition}[UNA-encoding]\label{def:una}
    The \emph{UNA-encoding} of a query $q$, denoted $\UNA(q)$, is the set of $\DLogics$ axioms of the form $\{a\} \sqcap \{b\} \equiv \bot$ 
    for every pair of distinct individual names $a, b$ in $\sccq$.
\end{definition}

\begin{proposition}
  \label{prop:una}
  For every extended graph $\graphext$ of a \sccqname $q$, it holds that $\shaclvalid{\graphext}{\UNA(q)}$.
\end{proposition}

We do not infer any axioms based on the DCA because a SCCQ does not determine the set of individual names $\IndividualNames$.
Concerning the CWA, a query imposes restrictions on concept names that appear 
in the query pattern (e.g., $\dot{A}$), 
the query template (e.g., $\ddot{A}$), 
variables (e.g. $\vconcept{x}$), 
and individual names (e.g., $\{a\}$).
All other concept names are irrelevant (see \Cref{prop:relevant-shapes}).

We define the following utility concepts $\Concept{u}$ (\Cref{def:ct}) for referring to the nominal concept \emph{or} variable concept for an individual name \emph{or} variable~$u$, 
and $\vcg(q)$, referring to the variable connectivity graph of a query $q$.

\begin{definition}
  \label{def:ct}
  For each individual name or variable $u$, $\Concept{u}$ is $\{a\}$ if $u$ is an individual name $a$, 
  or $\Concept{u}$ is $\vconcept{x}$ if $u$ is a variable $x$.
\end{definition}

\begin{definition}[Variable Connectivity Graph]\label{def:vcg}
  The \emph{variable connectivity graph} of query pattern $\sccqpattern$, denoted $\vcg(P)$, 
  is the graph whose nodes are the atoms in $\sccqpattern$, 
  and which has an undirected edge $\{t_1,t_2\}$ if and only if atoms $t_1$ and $t_2$ share a variable. 
\end{definition}

A \sccqname imposes restrictions on concept names in extended graphs, by definition of $\graphmed$, $\graphout$, and $\graphvar$.
For example, each atom $\atomC{x}{A} \in P$ implies $\vconcept{x} \sqsubseteq \dot{A}$, since concept $\vconcept{x}$ is defined from all individual names referred to by $x$, which according to the evaluation semantics of \sccqname result from filtering $A$.
More generally, all atoms $\atomC{x}{A}$ and $\atomP{x}{y}{p}$ in~$\sccqpattern$ (\Cref{ex:ext-graph}) restrict the instances of variable concept $\vconcept{x}$.
These observations can be combined over all atoms in a query, leading to \Cref{def:cwa}.

\begin{definition}[CWA-encoding]\label{def:cwa}
  The \emph{CWA-encoding} for a \sccqname $q = (\sccqformal)$, denoted $\CWA(q)$, is the minimal set of $\DLogics$ axioms including:

    \noindent 1. $\,$ For each concept name $A$ in $\sccqpattern$, $\dot{A} \equiv A \sqcap \bigsqcup_{\atomC{u}{A} \in P} \Concept{u}$.

    \noindent 2. $\,$ For each concept name $A$ in $\sccqtemplate$, $\ddot{A} \equiv \bigsqcup_{\atomC{u}{A} \in H} \Concept{u}$.

    \noindent 3. $\,$ For each variable $x$ in $\var(q)$ the axiom
        \begin{align*}
          \vconcept{x} \sqsubseteq \textstyle 
          & {\bigsqcap_{\atomC{x}{A} \in P}} A 
            \sqcap {\bigsqcap_{\atomP{x}{u}{p} \in P}} \exists p.\Concept{u}
            \sqcap {\bigsqcap_{\atomP{u}{x}{p} \in P}} \exists p^-.\Concept{u},
        \end{align*}

        and if $\vcg(P)$ is acyclic w.r.t $x$, then also the axiom

        \begin{align*}
        \vconcept{x} \sqsupseteq \textstyle 
          & {\bigsqcap_{\atomC{x}{A} \in P}} A 
          \sqcap {\bigsqcap_{\atomP{x}{u}{p} \in P}} \exists p.\Concept{u}
          \sqcap {\bigsqcap_{\atomP{u}{x}{p} \in P}} \exists p^-.\Concept{u}.
        \end{align*}
    
      \noindent 4. $\,$ For each role name $p$ in pattern $\sccqpattern$ the axioms
        \begin{align*}
          \exists \dot{p}.\Concept{v} &\equiv \textstyle \bigsqcup_{\atomP{u}{v}{p} \in P} \Concept{u}, &
          \exists \dot{p}.\top &\equiv {\bigsqcup_{\atomP{u}{v}{p} \in P}} \Concept{u} \sqcap \exists \dot{p}.\Concept{v},\\
          \exists \dot{p}^-.\Concept{u} &\equiv \textstyle \bigsqcup_{\atomP{u}{v}{p} \in P} \Concept{v}, &
          \exists \dot{p}^-.\top &\equiv {\bigsqcup_{\atomP{u}{v}{p} \in P}} \Concept{v} \sqcap \exists \dot{p}^-.\Concept{u}.
        \end{align*}

      \noindent 5. $\,$ For each role name $p$ in template $\sccqtemplate$ the axioms
        \begin{align*}
          \exists \ddot{p}.\Concept{v} &\equiv \textstyle \bigsqcup_{\atomP{u}{v}{p} \in H} \Concept{u}, &
          \exists \ddot{p}.\top &\equiv {\bigsqcup_{\atomP{u}{v}{p} \in H}} \Concept{u} \sqcap \exists \ddot{p}.\Concept{v},\\
          \exists \ddot{p}^-.\Concept{u} &\equiv \textstyle \bigsqcup_{\atomP{u}{v}{p} \in H} \Concept{v}, &
          \exists \ddot{p}^-.\top &\equiv {\bigsqcup_{\atomP{u}{v}{p} \in P}} \Concept{v} \sqcap \exists \ddot{p}^-.\Concept{u}.
        \end{align*}

\end{definition}

Observe, that unlike in the definition for concepts $\dot{A}$ (\Cref{def:cwa}, 1.), the definition for concepts $\ddot{A}$ (\Cref{def:cwa}, 2.) does not include $A$, since elements of $\dot{A}$ are the result of filtering $A$, whereas $\ddot{A}$ is newly constructed for the query template $\sccqtemplate$.
We first demonstrate the general meaning of these axioms in \Cref{ex:cwa}.

\begin{example}\label{ex:cwa}
  Consider again the query $\eqonen = \eqone$ (\Cref{ex:qone}).
  Then, $\CWA(\eqonen)$ consists of the following axioms:

  \begin{enumerate}
  \item
    $\{ \eBm \equiv \eBi \sqcap \eyv,$
    $   \eEm \equiv \eEi \sqcap \ezv \}$,
    because, e.g., concept $\eBm$ in the extended graph is defined by filtering $\eBi$ with variable $\eyv$, based on the query pattern $\atomC{\ey}{\eBi}$ in $\eqonen$.

  \item
    $\{ \eBo \equiv \ezv,$
    $   \eEo \equiv \eyv \}$,
    because, e.g., concept $\eBo$ in the extended graph is defined by $\ezv$, since it only occurs in the single construct pattern $\atomC{\ez}{\eBon}$.
    If there were multiple occurences, it would be defined by the union of all variables, instead.

  \item
    $\{ \ewv \sqsubseteq \exists \epi . \eyv,$
    $   \exv \sqsubseteq \exists \epi . \ezv,$
    $   \eyv \sqsubseteq \exists \epi . \ewv \sqcap \eBi,$
    $   \ezv \sqsubseteq \exists \epi . \exv \sqcap \eEi \}$,
    because variable concepts are defined by constraints to the variable in the query pattern.
    For example, $\eyv$ is constrained by patterns $\atomP{\ew}{\ey}{\epi}$ and $\atomC{\ey}{\eBi}$ in $\eqonen$, 
    and thus bound by $\exists \epi . \ewv \sqcap \eBi$.
    This is a crucial step, since concept and role names in the extended graph are defined in terms of these variable concepts.
    The inverse cases are included, because $\vcg(\sccqpattern)$ is acyclic:
    $\{ \ewv \sqsupseteq \exists \epi . \eyv,$
    $   \exv \sqsupseteq \exists \epi . \ezv,$
    $   \eyv \sqsupseteq \exists \epi . \ewv \sqcap \eBi,$
    $   \ezv \sqsupseteq \exists \epi . \exv \sqcap \eEi \}$ (cf. \Cref{ex:cwafail}).

  \item
    $\{ \exists \epm . \eyv \equiv \ewv,$
    $   \exists \epm . \ezv \equiv \exv,$ 
    $   \exists \epm . \top \equiv (\ewv \sqcap \exists \epm . \eyv) \sqcup (\exv \sqcap \exists \epm . \ezv) \}$,
    because, e.g., role name $\epm$ in the extended graph is defined by the variables concepts that it occurs with.
    Similarly, the following axioms for inverse role names are included:
    $\{ \exists \epmi . \ewv \equiv \eyv,$
    $   \exists \epmi . \exv \equiv \ezv,$
    $   \exists \epmi . \top \equiv (\eyv \sqcap \exists \epmi . \ewv) \sqcup (\ezv \sqcap \exists \epmi . \exv) \}$.

  \item
    $\{ \exists \epo . \ezv \equiv \eyv,$
    $   \exists \epo . \top \equiv \eyv \sqcap \exists \epo . \ezv \}$ and
    $\{ \exists \epoi . \eyv \equiv \ezv,$ 
    $   \exists \epoi . \top \equiv \ezv \sqcap \exists \epoi . \eyv \}$,
    with analogous reasoning as the previous case.
  \end{enumerate}
\end{example}

Note the additional condition in the second part of \Cref{def:cwa} (3.) where we require $\vcg(q)$ to be acyclic.
In the following example (\Cref{ex:cwafail}), we will motivate why this condition is required and then define \Cref{lemma:vcg} with respect to this case.

\begin{example}\label{ex:cwafail}
  Consider the pattern $\sccqpattern = \{\atomP{x}{y}{\eri}, \atomP{y}{z}{\eri}, \atomP{x}{z}{\epi}\}$ of a query $\sccq = \sccqformal$, and the graph $G$ in \Cref{fig:fail}.
  \begin{figure}[t]
  \centering
    \begin{tikzpicture}[
    every label/.style={inner sep=0.0em}
]
    \node[new node] (a) {$a_1$};
    \node[new node] (b) [right = 0.8cm of a]  {$a_2$};
    \node[new node] (c) [right = 0.8cm of b]  {$a_3$};
    \node[new node] (d) [right = 0.8cm of c]  {$a_4$};
    \path[->]
        (a) edge [bend left=20] node [above] {$\eri$} (b)
        (b) edge [bend left=20] node [below] {$\eri$} (a)
        (b) edge [bend left=20] node [above] {$\eri$} (c)
        (c) edge [bend left=20] node [below] {$\eri$} (b)
        (a) edge [bend left=50] node [above] {$\epi$} (c)
        (c) edge [bend left=50] node [below] {$\epi$} (a)
        (c) edge [bend left=20] node [above] {$\eri$} (d)
        (d) edge [bend left=20] node [below] {$\eri$} (c);
\end{tikzpicture}
    \caption{Input graph $G$ for \Cref{ex:cwafail}.}
    \label{fig:fail}
  \end{figure}
  Note, that $\vcg(\sccqpattern)$ is cyclic, since $(\atomP{x}{y}{\eri}, \atomP{y}{z}{\eri})$,  $(\atomP{y}{z}{\eri}, \atomP{x}{z}{\epi})$ as well as $(\atomP{x}{z}{\epi}, \atomP{x}{y}{\eri})$ each share variables.
  
  Evaluating $q$ on $G$ results in mappings $\mu_1 = \{x \mapsto a_1, y \mapsto a_2, z \mapsto a_3\}$ and 
  $\mu_2 = \{x \mapsto a_3, y \mapsto a_2, z \mapsto a_1\}$.
  Thus, the variable concepts are defined as $\exv = \{a_1, a_3\}$, $\eyv = \{a_2\}$ and $\ezv = \{a_3, a_1\}$.
  Note, that $y \mapsto a_4$ is not in any result mapping for query $q$ on graph $G$ (and $a_4 \not\in \eyv$).
  However, $\{a_4\} \sqsubseteq \exists \eri^- . \exv \sqcap \exists \eri . \ezv$.
  Therefore, $\eyv \not\sqsupseteq \exists \eri^- . \exv \sqcap \exists \eri . \ezv$, so we can not include this axiom.
\end{example}

Intuitively, an acyclic graph $\vcg(P)$ allows for separating the pattern $\sccqpattern$ (given, as an example, variable $x$ and concept name $A$) into patterns $P_l$, $\{x:A\}$, and $P_r$, where $P_l$ shares at most variable $x$ with $P_r$.
In these cases, the implicit dependencies between bindings for variables that cause issues as demonstrated for $x$ and $z$ in \Cref{ex:cwafail} do not occur.

\begin{lemma}\label{lemma:vcg}
  Let $\sccq = \sccqformal$ be a query such that $\vcg(\sccqpattern)$ is acyclic.
  Let $G$ be a graph, and let $x$ be a variable corring in $\sccqpattern$. 
  Then
  \[
        \vconcept{x} \sqsupseteq \textstyle 
          {\bigsqcap_{\atomC{x}{A} \in P}} A 
          \sqcap {\bigsqcap_{\atomP{x}{u}{p} \in P}} \exists p.\Concept{u}
          \sqcap {\bigsqcap_{\atomP{u}{x}{p} \in P}} \exists p^-.\Concept{u}.
  \]
\end{lemma}

Given \Cref{lemma:vcg}, the following proposition holds.

\begin{proposition}\label{prop:cwa}
  For every extended graph $\graphext$ of a \sccqname $\sccq$, it holds that $\shaclvalid{\graphext}{\CWA(q)}$, if either $q$ does not include
  any individual names, or the output graph is guaranteed to be non-empty.
\end{proposition}

Note the additional condition in \Cref{prop:cwa}: If the output graph is empty and the query includes individual names, then $\graphext$ may not be valid with respect to $\CWA(q)$, since
the constructed axioms may include individual names that are not guaranteed to exist.
This could be remedied by not allowing individual names in queries; however, since Simple SHACL shapes do not allow individual names, these axioms do not impact soundness of the method.

\subsection{Axiomatizations for Query Subpatterns}
\label{ss:components}

We refer to a pattern $P' \subseteq \sccqpattern$ as a \emph{component} of the pattern $\sccqpattern$, if $\vcg(P')$ 
(see also~\Cref{def:vcg}) is a connected subgraph of $\vcg(\sccqpattern)$ 
and there exists no $P''$ such that $P' \subset P''$ and $\vcg(P'')$ is a connected subgraph of $\vcg(\sccqpattern)$.

\begin{example}\label{ex:homomorphism}
  Query $\eqonen$ (\Cref{ex:qone}) has components $\{\atomP{\ew}{\ey}{\epi}, \atomC{\ey}{\eBi}\}$ and $\{\atomP{\ex}{\ez}{\epi}, \atomC{\ez}{\eEi}\}$.
  The CWA encoding (\Cref{ex:cwa}) does not entail $\eyv \sqsubseteq \ezv$, even though this axiom is both valid in all extended graphs, 
  and required for inferring, e.g., the result shape $\eEon \sqsubseteq \eBon$.
\end{example}

\Cref{ex:homomorphism} shows that the $\CWA$ encoding alone is not sufficient for inferring all subsumptions between variable concepts.
If we could find a homomorphism between two components of the query pattern, we would know that the valuations of one component are a subset of the valuations of the other component (modulo variable names), and thus infer subsumptions between variable concepts.

\begin{definition}[Component Map]\label{def:cmap}
  For components $P_1$ and $P_2$ of $P$, every function $h:\var(P_1) \to \IndividualNames \cup \Variables$ such that $\himage{P_1} \subseteq P_2$ 
  is called a \emph{component map} on $\sccqpattern$, where we write $\himage{P_1}$ to mean substitution of each variable $x$ in $P_1$ by $h(x)$.
\end{definition}
\begin{definition}[Component Map Axioms]\label{def:cmapa}
  The set of axioms inferred from a component map $h$ on $P$, denoted $\MapAxioms_h(P)$, contains axiom $\Concept{h(x)} \sqsubseteq \vconcept{x}$ 
  for every variable $x$ in the domain of $h$. 
  The union of all sets $\MapAxioms_h(\sccqpattern)$ of a graph pattern $P$ is called $\MapAxioms(\sccqpattern)$.
\end{definition}
\begin{example}
  Consider two components $P_1 = \{\atomP{x}{y}{p}\}$ and $P_2 = \{\atomP{z}{z}{p},\atomC{z}{A}\}$. 
  Then we can define the mapping $h(x) = z$ and $h(y) = z$, such that $\himage{P_1} \subseteq P_2$. 
  Therefore, we can construct the axioms $\vconcept{z} \sqsubseteq \vconcept{x}$ and $\vconcept{z} \sqsubseteq \vconcept{y}$ valid on $\graphext$.
\end{example}
\begin{proposition}\label{prop:map}
  For every extended graph $\graphext$ of a \sccqname $q = \sccqformal$, it holds that $\shaclvalid{\graphext}{\MapAxioms(\sccqpattern)}$.
\end{proposition}

\subsection{Extending Query Patterns via Constraints}
\label{ss:extend}

The basic mapping $\MapAxioms(\sccqpattern)$ is not sufficient for inferring certain crucial variable concept subsumptions, as \Cref{ex:extends} shows.

\begin{example}\label{ex:extends}
  Consider components $P_1 = \{\atomP{\ex}{\ez}{\epi}, \atomC{\ez}{\eEi}\}$ and $P_2 = \{\atomP{\ew}{\ey}{\epi}, \atomC{\ey}{\eBi}\}$ of query $\eqonen$ (\Cref{ex:qone}).
  Here, we can not find a mapping $h$ satisfying \Cref{def:cmap}.
  However, we know based on $\eS$ that $\esthreei$ (\Cref{ex:sone}).
  We can utilize this knowledge to extend component $P_2$, adding the pattern $\atomC{\ey}{\eEi}$ which does not alter the queries results.
  Now we can find the mapping $h(\ex) = \ew, h(\ez) = \ey$ such that $\himage{P_1} \subseteq P_2$.
\end{example}

Intuitively, by extending a component as illustrated in \Cref{ex:extends}, we reveal a subsumption relationship that was implicit in the input shapes.
For the same reason, the extended component is not more restrictive than the original one.
We now show how this approach can be generalized.

\begin{definition}[Target Variables]
  A variable $x$ is \emph{target variable} for a shape $\targetquery \sqsubseteq \phi$ in an atomic pattern $t$ if and only if either
  \begin{enumerate}
    \item $t = \atomC{x}{A}$ and $\targetquery = A$,
    \item $t = \atomP{x}{y}{p}$ and $\targetquery = \exists p.\top$, or
    \item $t = \atomP{y}{x}{p}$ and $\targetquery = \exists p^-.\top$.
\end{enumerate}
\end{definition}

\newcommand{\MaxExt}{\operatorname{MaxExt}}
\newcommand{\Ext}{\operatorname{Ext}}

\begin{definition}[Extension]
  The \emph{extension} $\Ext(x, \phi)$ of a variable $x$ with respect to a shape constraint $\phi$ and component $P_i$ is the set of atoms defined below, where $x_0$ is a fresh variable.
  \begin{align*}
    \Ext(x, A) &= \{\atomC{x}{A}\},\\
    \Ext(x, \exists p^-.A) &= \{\atomP{x_0}{x}{p}, \atomC{x_0}{A}\},\\
    \Ext(x, \forall p.A) &= \{\atomC{y}{A} \mid \atomP{x}{y}{p} \in P_i \},\\
    \Ext(x, \exists p.A) &= \{\atomP{x}{x_0}{p}, \atomC{x_0}{A}\},\ \text{and}\\
    \Ext(x, \forall p^-.A) &= \{\atomC{y}{A} \mid \atomP{y}{x}{p} \in P_i \}.
  \end{align*}
\end{definition}

Since new atoms are added to the pattern, they can be targets of input shapes, too.
The recursive extension is bound by the maximum degree and diameter of the connectivity graph $\vcg(P)$ of the query pattern $P$ (\Cref{def:vcg}).

\begin{definition}[Bound extension]\label{def:bound-extension}
  Let $P_i$ be a component of a query pattern $P$, $x$ be a variable in $\var(P_i)$, $\shapesin$ be a finite set of shapes, and $P_i^x$ be a pattern that results from adding iteratively atoms $\Ext(u, s)$ to $P_i$, where $s \in \shapesin$, $u$ is a target variable for $s$, and either $u = x$ or $u \notin \var(P_i)$. Then, $P_i^x$ is a \emph{bound $x$-extension} of $P_i$ using $\shapesin$ if and only if the followings conditions are satisfied:
  \begin{enumerate}
  \item
    the maximum degree of $\vcg(P_i^x)$ is not bigger than the maximum degree of $\vcg(P)$,
  \item
    the diameter of graph $\vcg(P_i^x \setminus P_i)$ is not longer than the maximum diameter of the components of $\vcg(P)$.
  \end{enumerate}
\end{definition}

\begin{definition}[Maximum Extension]\label{def:max-extension}
  Given a component $P_i$ of pattern $P$, a variable $x \in \var(P_i)$, and a finite set of Simple SHACL shapes $\shapesin$, $\MaxExt_x(P_i, \shapesin)$ is the maximum bound $x$-extension for $P_i$ using $\shapesin$.
  The \emph{maximum extension} for $P_i$ using $\shapesin$, denoted $\MaxExt(P_i, \shapesin)$, is the pattern $\bigcup_{x \in \var(P_i)}\MaxExt_x(P_i, \shapesin)$.
\end{definition}

Intuitively, \Cref{def:bound-extension} and \Cref{def:max-extension} ensure that an extended component is finite, but still allows for all possible mappings with another component:
Since we are only interested in finding axioms involving names in $\graphext$, we must use at least one such name in the mapping.
Since the other mapping component is a subset of $P$, the mapping can then, in the worst case, only extend with respect to the maximum degree and diameter of $P$.

The maximum extension thus allows for finding all axioms of interest via component maps $h$ from $P_1$ to $\MaxExt(P_2, \shapesin)$, where $P_1$ and $P_2$ are components of $P$.

\begin{definition}[Extended Component Map Axioms]
  The set of \emph{extended component map axioms} of a pattern $P$, and a set of shapes $\shapesin$, denoted $\MapAxiomsSin(P)$ 
  is the set that includes an axiom $\Concept{u} \sqsubseteq \vconcept{x}$ if and only if there is a pair of components $P_1$ and $P_2$ of $P$,
  and a component map $h$ from $P_1$ to  $\MaxExt(P_2, \shapesin)$ such that $h(x)=u$ and $u$ is a variable or an individual name occurring in $P_2$.
\end{definition}

\begin{proposition}\label{prop:mapsin}
  For every extended graph $\graphext$ of a \sccqname $\sccq = \sccqformal$ and set of input shapes $\shapesin$, 
  it holds that $\shaclvalid{\graphext}{\MapAxiomsSin(P)}$.
\end{proposition}

\subsection{Axiomatizations for Role Hierachies}
\label{ss:properties} 

Not only variable concepts form hierarchies that are not entailed by the axioms included this far.
We finally infer axioms representing additional role hierarchies, that are determined from the query (\Cref{def:sppa}).

\begin{definition}[Role Hierarchy Axioms]\label{def:sppa}
  The \emph{role hierarchy axioms} of a query $q = (H \gets P)$ are the set of axioms, denoted $\PropertyAxioms(q)$, that include:
  \begin{enumerate}
  \item 
    for each role name $p \in \sccqpattern$, the axiom $\dot{p} \sqsubseteq p$,
  \item 
    for each role name $p \in \sccqpattern$, the axiom $p \sqsubseteq \dot{p}$, 
    if all atoms with role name $p$ occurring in $P$ have the form $\atomP{x}{y}{p}$ where variables $x$ and $y$ occur in no other atom in $\sccqpattern$ and $x \neq y$,
  \item
    for each pair of role names $p,r$ with $\atomP{x}{y}{p} \in \sccqpattern$ and either $\atomP{x}{y}{r} \in \sccqtemplate$ or $\atomP{y}{x}{r} \in \sccqtemplate$
    \begin{enumerate}
      \item the axiom $\dot{p} \sqsubseteq \ddot{r}$ (if $\atomP{x}{y}{r} \in \sccqtemplate$) 
            or the axiom $\dot{p} \sqsubseteq \ddot{r}^-$ (if $\atomP{y}{x}{r} \in \sccqtemplate$),
            if $P$ does not contain any other atoms with role name $p$, and
      \item the axiom $\ddot{r} \sqsubseteq \dot{p}$ (if $\atomP{x}{y}{r} \in \sccqtemplate$) 
            or the axiom $\ddot{r}^- \sqsubseteq \dot{p}$ (if $\atomP{y}{x}{r} \in \sccqtemplate$),
            if $H$ does not contain any other atoms with role name $r$.
    \end{enumerate}
  \end{enumerate}
\end{definition}

Trivially, for any role name $p \in \sccqpattern$, the axiom $\dot{p} \sqsubseteq p$ holds, since, by definition, $\graphmed \subseteq \graphin$.
The inverse axiom $p \sqsubseteq \dot{p}$ holds, if the role name is unconstrained in pattern $\sccqpattern$.
Role hierarchy axioms between $p \in \sccqpattern$ and $r \in \sccqtemplate$, that is axioms $\ddot{r} \sqsubseteq \dot{p}$ and $\dot{p} \sqsubseteq \ddot{r}$, hold,
if there are no further restrictions on $r$ and $p$, respectively.

\begin{example}\label{ex:role}
    Consider the input shape $\eAi \sqsubseteq \exists p.\eAi$ and the query $\eqtwon = \eqtwo$.
    The axioms presented prior to \Cref{def:sppa} do not entail the shape $\eAo \sqsubseteq \exists \epo.\eAo$, 
    even though $\eAon \sqsubseteq \exists \epon.\eAon$ should apply to the output graph: 
    After all, we simply copy all instances of $\eAi$ and the entirety of $\epi$.
    If we include, however, axioms $\epi \sqsubseteq \epm$ and $\epm \sqsubseteq \epi$,
    the new set of axioms does indeed entail $\eAo \sqsubseteq \exists \epo.\eAo$, as expected.
    We can include these axioms -- in this case -- since we simply copy $p$ in its entirety, or, formally,
    the variables in the only atomic pattern including the role name $p$ are not further constrained.
\end{example}

\begin{proposition}\label{prop:properties}
  For every extended graph $\graphext$ of a \sccqname $\sccq$, it holds that $\shaclvalid{\graphext}{\PropertyAxioms(q)}$.
\end{proposition}

\section{Related Work}\label{sec:related}

The problem of automatically inferring SHACL (or ShEx~\cite{DBLP:conf/i-semantics/PrudhommeauxGS14}) shapes from various inputs has been studied before.
Most commonly it has been considered in the context of constructing shapes from concrete instance data, based on summaries of statistical information over graphs~\cite{DBLP:conf/semweb/SpahiuMP18, DBLP:journals/kbs/Fernandez-Alvarez22, DBLP:journals/pvldb/RabbaniLH23}, or more involved machine learning techniques~\cite{DBLP:conf/sac/Mihindukulasooriya18, DBLP:journals/semweb/OmranTMH23, DBLP:conf/icdt/GrozLSW22}.
Some approaches combine such methods with tools for manual exploration and adaptation of inferred schemata~\cite{DBLP:conf/semweb/BonevaDFG19}.
Our approach, on the other hand, allows the construction of valid shapes from only input shapes and a given query, without the need to consider (or indeed provide) any concrete instance data.

Our work is based in the correspondence of SHACL and description logics, inspired by~\citet{DBLP:conf/lpnmr/BogaertsJB22}.
This correspondence has been investigated before.
Astrea~\cite{DBLP:conf/esws/CimminoFG20} produces SHACL shapes from OWL ontologies by providing a mapping relating patterns of ontology constructs (i.e., language constructs including a specific usage context) with equivalent patterns of SHACL constructs validating them.
Similarly,~\citet{DBLP:conf/semweb/PanditOL18a} explore the usage of ontology design patterns for the generation of SHACL shapes.

Inference of constraints, as well as SHACL shapes, from other data formalisms has been studied as well.
\citet{DBLP:conf/aaai/CalvaneseFPSS14} and \citet{DBLP:conf/www/SequedaAM12} consider inference of RDFS and OWL, respectively, from direct mappings~\cite{arenas2012direct} between relational data and RDF.
Similarly, \citet{DBLP:conf/semweb/ThapaG21} consider inferences of SHACL shapes from direct mappings, while RML2SHACL~\cite{DBLP:conf/kcap/DelvaSOALD21} allows the translation of RML rules to SHACL shapes.
These approaches differ from our approach, in that the input is restricted to a direct mapping or RML mapping from relational data, whereas in our case, the input is defined by an arbitrary query pattern imposing additional constraints, as well as constraints explicitly given as input shapes.

More generally, the inference of constraints over views of relational databases has been studied in the past~\cite{DBLP:journals/tods/KlugP82, DBLP:journals/pvldb/FanMHLW08, DBLP:conf/sigmod/Stonebraker75, DBLP:journals/tods/JacobsAK82}.
While these approaches are similar to our work, they face severe problems: In~\cite{DBLP:journals/tods/JacobsAK82}, general first-order formulas are considered as constaints.
Even though these are more expressive than our constraints, the presented approach is not feasible in practice. As a result, other approaches restrict constraints, most commonly to functional or join dependencies, e.g.,~\cite{DBLP:journals/tods/KlugP82, DBLP:journals/pvldb/FanMHLW08}. These approaches can be considered complementary to our approach, since our approach lacks cardinalities, which are required to express functional dependencies in SHACL; on the other hand, functional dependencies can not express crucial typing constraints for knowledge graphs that are supported by our approach.

Finally, \citet{DBLP:conf/semweb/ThapaG22} consider mappings of relational data to RDF.
In particular, the work focuses on including SQL integrity constraints (keys, uniqueness and not-null constraints) in the translation to SHACL constraints, allowing for a limited number of property constraints mapped from integrity constraints.

\section{Concluding Remarks}
\label{sec:conclusion}

We have presented an algorithm for constructing a set of shapes characterizing the possible output graphs of CONSTRUCT queries, where the input graphs of these queries can be constrained by a set of shapes as well.
The shapes are expressed in a subset of SHACL, whereas the queries are expressed in a subset of SPARQL.
This enables the inference of shapes over result graphs of data processing pipelines (i.e., compositions of CONSTRUCT queries), which can be used both for validation purposes when working with these result graphs, and informatively, aiding developers directly.

The algorithm decides for the finite set of candidate shapes, whether they are entailed by a set of description-logic axioms valid on the union of graphs involved in the query operation.
We prove soundness of this algorithm, and provide an implementation.

\paragraph{Limitations}
(1) The output shapes computed by our approach are sound, but incomplete.
Consider the problem $\sccq_3 =\{\atomP{x}{y}{p}, \atomC{z}{A}\} \gets \{\atomP{x}{y}{p}, \atomC{z}{A}, \atomP{z}{w}{p}\}$ and $\shapesin = \{A \sqsubseteq \exists p.A\}$, where the input shape would apply to the output, but we can not infer it.
This and similar problems could perhaps be remedied by extending the inference of role hierachy axioms to also consider input shapes.

(2) Our approach is limited to a subset of SHACL and SPARQL.
In the \wheretolook, we show how to extend the approach to arbitrary $\DLogics$ constraints.
Intuitively, this extension is possible because the propositions presented in this paper are not restricted to Simple SHACL (consider, in particular, \Cref{prop:reduction-extended-graph}).

\paragraph{Future Work}
In order to extend the approach to queries involving generic patterns (e.g., $\atomP{x}{y}{z}$ or $\atomC{x}{z}$), an expansion to non-generic queries may be possible, since all relevant role and concept names are known from template and input shapes.

While the application of our approach to entire data processing pipelines is straightforward, there are interesting empirical questions regarding the properties of results shapes, e.g., depending on the nature of input shapes or number of processing steps left as future work.

\begin{acks}
  This work was partially funded by the Deutsche Forschungsgemeinschaft (DFG) under COFFEE -- STA 572\_15-2, and the DFG Germany's Excellence Strategy -- EXC 2120/1 -- 390831618.
\end{acks}

\bibliographystyle{ACM-Reference-Format}
\bibliography{paper}


\begin{thebibliography}{34}


\ifx \showCODEN    \undefined \def \showCODEN     #1{\unskip}     \fi
\ifx \showDOI      \undefined \def \showDOI       #1{#1}\fi
\ifx \showISBNx    \undefined \def \showISBNx     #1{\unskip}     \fi
\ifx \showISBNxiii \undefined \def \showISBNxiii  #1{\unskip}     \fi
\ifx \showISSN     \undefined \def \showISSN      #1{\unskip}     \fi
\ifx \showLCCN     \undefined \def \showLCCN      #1{\unskip}     \fi
\ifx \shownote     \undefined \def \shownote      #1{#1}          \fi
\ifx \showarticletitle \undefined \def \showarticletitle #1{#1}   \fi
\ifx \showURL      \undefined \def \showURL       {\relax}        \fi
\providecommand\bibfield[2]{#2}
\providecommand\bibinfo[2]{#2}
\providecommand\natexlab[1]{#1}
\providecommand\showeprint[2][]{arXiv:#2}

\bibitem[Angles et~al\mbox{.}(2018)]%
        {DBLP:conf/sigmod/AnglesABBFGLPPS18}
\bibfield{author}{\bibinfo{person}{Renzo Angles}, \bibinfo{person}{Marcelo Arenas}, \bibinfo{person}{Pablo Barcel{\'{o}}}, \bibinfo{person}{Peter~A. Boncz}, \bibinfo{person}{George H.~L. Fletcher}, \bibinfo{person}{Claudio Gutierrez}, \bibinfo{person}{Tobias Lindaaker}, \bibinfo{person}{Marcus Paradies}, \bibinfo{person}{Stefan Plantikow}, \bibinfo{person}{Juan~F. Sequeda}, \bibinfo{person}{Oskar van Rest}, {and} \bibinfo{person}{Hannes Voigt}.} \bibinfo{year}{2018}\natexlab{}.
\newblock \showarticletitle{{G-CORE:} {A} Core for Future Graph Query Languages}. In \bibinfo{booktitle}{\emph{Proc. of {SIGMOD}}}. \bibinfo{publisher}{{ACM}}, \bibinfo{pages}{1421--1432}.
\newblock
\urldef\tempurl%
\url{https://doi.org/10.1145/3183713.3190654}
\showDOI{\tempurl}


\bibitem[Arenas et~al\mbox{.}(2012)]%
        {arenas2012direct}
\bibfield{author}{\bibinfo{person}{Marcelo Arenas}, \bibinfo{person}{Alexandre Bertails}, \bibinfo{person}{Eric Prud’hommeaux}, \bibinfo{person}{Juan Sequeda}, {et~al\mbox{.}}} \bibinfo{year}{2012}\natexlab{}.
\newblock \bibinfo{title}{A Direct Mapping of Relational Data to {RDF}}.
\newblock
\newblock
\urldef\tempurl%
\url{https://www.w3.org/TR/rdb-direct-mapping/}
\showURL{%
Retrieved 12.02.2024 from \tempurl}


\bibitem[at~University~of Oxford(2008)]%
        {hermit}
\bibfield{author}{\bibinfo{person}{KRR~Group at~University~of Oxford}.} \bibinfo{year}{2008}\natexlab{}.
\newblock \bibinfo{title}{HermiT OWL Reasoner}.
\newblock
\newblock
\urldef\tempurl%
\url{http://www.hermit-reasoner.com/}
\showURL{%
Retrieved 12.10.2023 from \tempurl}


\bibitem[Baader et~al\mbox{.}(2003)]%
        {DBLP:conf/dlog/2003handbook}
\bibfield{editor}{\bibinfo{person}{Franz Baader}, \bibinfo{person}{Diego Calvanese}, \bibinfo{person}{Deborah~L. McGuinness}, \bibinfo{person}{Daniele Nardi}, {and} \bibinfo{person}{Peter~F. Patel{-}Schneider}} (Eds.). \bibinfo{year}{2003}\natexlab{}.
\newblock \bibinfo{booktitle}{\emph{The Description Logic Handbook: Theory, Implementation, and Applications}}.
\newblock \bibinfo{publisher}{Cambridge University Press}.
\newblock
\showISBNx{0-521-78176-0}


\bibitem[Bielefeldt et~al\mbox{.}(2018)]%
        {BGK2018}
\bibfield{author}{\bibinfo{person}{Adrian Bielefeldt}, \bibinfo{person}{Julius Gonsior}, {and} \bibinfo{person}{Markus Kr{\"{o}}tzsch}.} \bibinfo{year}{2018}\natexlab{}.
\newblock \showarticletitle{Practical Linked Data Access via {SPARQL:} The Case of Wikidata}. In \bibinfo{booktitle}{\emph{Workshop on Linked Data on the Web co-located with The Web Conference 2018, LDOW@WWW 2018}} \emph{(\bibinfo{series}{{CEUR} Workshop Proceedings}, Vol.~\bibinfo{volume}{2073})}. \bibinfo{publisher}{CEUR-WS.org}.
\newblock
\urldef\tempurl%
\url{https://ceur-ws.org/Vol-2073/article-03.pdf}
\showURL{%
\tempurl}


\bibitem[Bogaerts et~al\mbox{.}(2022)]%
        {DBLP:conf/lpnmr/BogaertsJB22}
\bibfield{author}{\bibinfo{person}{Bart Bogaerts}, \bibinfo{person}{Maxime Jakubowski}, {and} \bibinfo{person}{Jan~Van den Bussche}.} \bibinfo{year}{2022}\natexlab{}.
\newblock \showarticletitle{{SHACL:} {A} Description Logic in Disguise}. In \bibinfo{booktitle}{\emph{Proc. of Logic Programming and Nonmonotonic Reasoning}} \emph{(\bibinfo{series}{LNCS}, Vol.~\bibinfo{volume}{13416})}. \bibinfo{publisher}{Springer}, \bibinfo{pages}{75--88}.
\newblock
\urldef\tempurl%
\url{https://doi.org/10.1007/978-3-031-15707-3_7}
\showDOI{\tempurl}


\bibitem[Boneva et~al\mbox{.}(2019)]%
        {DBLP:conf/semweb/BonevaDFG19}
\bibfield{author}{\bibinfo{person}{Iovka Boneva}, \bibinfo{person}{J{\'{e}}r{\'{e}}mie Dusart}, \bibinfo{person}{Daniel Fern{\'{a}}ndez{-}{\'{A}}lvarez}, {and} \bibinfo{person}{Jos{\'{e}} Emilio~Labra Gayo}.} \bibinfo{year}{2019}\natexlab{}.
\newblock \showarticletitle{Shape Designer for ShEx and {SHACL} constraints}. In \bibinfo{booktitle}{\emph{Proc. of the {ISWC} 2019 Satellite Tracks co-located with {ISWC} 2019}} \emph{(\bibinfo{series}{{CEUR} Workshop Proceedings}, Vol.~\bibinfo{volume}{2456})}. \bibinfo{publisher}{CEUR-WS.org}, \bibinfo{pages}{269--272}.
\newblock
\urldef\tempurl%
\url{https://ceur-ws.org/Vol-2456/paper70.pdf}
\showURL{%
\tempurl}


\bibitem[Bonifati et~al\mbox{.}(2020)]%
        {DBLP:journals/vldb/BonifatiMT20}
\bibfield{author}{\bibinfo{person}{Angela Bonifati}, \bibinfo{person}{Wim Martens}, {and} \bibinfo{person}{Thomas Timm}.} \bibinfo{year}{2020}\natexlab{}.
\newblock \showarticletitle{An analytical study of large {SPARQL} query logs}.
\newblock \bibinfo{journal}{\emph{{VLDB} J.}} \bibinfo{volume}{29}, \bibinfo{number}{2-3} (\bibinfo{year}{2020}), \bibinfo{pages}{655--679}.
\newblock
\urldef\tempurl%
\url{https://doi.org/10.1007/s00778-019-00558-9}
\showDOI{\tempurl}


\bibitem[Calvanese et~al\mbox{.}(2014)]%
        {DBLP:conf/aaai/CalvaneseFPSS14}
\bibfield{author}{\bibinfo{person}{Diego Calvanese}, \bibinfo{person}{Wolfgang Fischl}, \bibinfo{person}{Reinhard Pichler}, \bibinfo{person}{Emanuel Sallinger}, {and} \bibinfo{person}{Mantas Simkus}.} \bibinfo{year}{2014}\natexlab{}.
\newblock \showarticletitle{Capturing Relational Schemas and Functional Dependencies in {RDFS}}. In \bibinfo{booktitle}{\emph{Proc. of the {AAAI} Conference on Artificial Intelligence}}. \bibinfo{publisher}{{AAAI} Press}, \bibinfo{pages}{1003--1011}.
\newblock
\urldef\tempurl%
\url{https://doi.org/10.1609/AAAI.V28I1.8867}
\showDOI{\tempurl}


\bibitem[Cimmino et~al\mbox{.}(2020)]%
        {DBLP:conf/esws/CimminoFG20}
\bibfield{author}{\bibinfo{person}{Andrea Cimmino}, \bibinfo{person}{Alba Fern{\'{a}}ndez{-}Izquierdo}, {and} \bibinfo{person}{Ra{\'{u}}l Garc{\'{\i}}a{-}Castro}.} \bibinfo{year}{2020}\natexlab{}.
\newblock \showarticletitle{Astrea: Automatic Generation of {SHACL} Shapes from Ontologies}. In \bibinfo{booktitle}{\emph{Proc. of {ESWC}}} \emph{(\bibinfo{series}{LNCS}, Vol.~\bibinfo{volume}{12123})}. \bibinfo{publisher}{Springer}, \bibinfo{pages}{497--513}.
\newblock
\urldef\tempurl%
\url{https://doi.org/10.1007/978-3-030-49461-2_29}
\showDOI{\tempurl}


\bibitem[Cyganiak et~al\mbox{.}(2014)]%
        {rdf}
\bibfield{author}{\bibinfo{person}{Richard Cyganiak}, \bibinfo{person}{David Wood}, \bibinfo{person}{Markus Lanthaler}, \bibinfo{person}{Graham Klyne}, \bibinfo{person}{Jeremy~J. Carroll}, {and} \bibinfo{person}{Brian McBride}.} \bibinfo{year}{2014}\natexlab{}.
\newblock \bibinfo{title}{{RDF} Concepts and Abstract Syntax}.
\newblock
\newblock
\urldef\tempurl%
\url{https://www.w3.org/TR/rdf11-concepts/}
\showURL{%
Retrieved 12.02.2024 from \tempurl}


\bibitem[Delva et~al\mbox{.}(2021)]%
        {DBLP:conf/kcap/DelvaSOALD21}
\bibfield{author}{\bibinfo{person}{Thomas Delva}, \bibinfo{person}{Birte~De Smedt}, \bibinfo{person}{Sitt~Min Oo}, \bibinfo{person}{Dylan~Van Assche}, \bibinfo{person}{Sven Lieber}, {and} \bibinfo{person}{Anastasia Dimou}.} \bibinfo{year}{2021}\natexlab{}.
\newblock \showarticletitle{{RML2SHACL:} {RDF} Generation Taking Shape}. In \bibinfo{booktitle}{\emph{Proc. of Knowledge Capture Conference}}. \bibinfo{publisher}{{ACM}}, \bibinfo{pages}{153--160}.
\newblock
\urldef\tempurl%
\url{https://doi.org/10.1145/3460210.3493562}
\showDOI{\tempurl}


\bibitem[Fan et~al\mbox{.}(2008)]%
        {DBLP:journals/pvldb/FanMHLW08}
\bibfield{author}{\bibinfo{person}{Wenfei Fan}, \bibinfo{person}{Shuai Ma}, \bibinfo{person}{Yanli Hu}, \bibinfo{person}{Jie Liu}, {and} \bibinfo{person}{Yinghui Wu}.} \bibinfo{year}{2008}\natexlab{}.
\newblock \showarticletitle{Propagating functional dependencies with conditions}.
\newblock \bibinfo{journal}{\emph{Proc. {VLDB} Endow.}} \bibinfo{volume}{1}, \bibinfo{number}{1} (\bibinfo{year}{2008}), \bibinfo{pages}{391--407}.
\newblock
\urldef\tempurl%
\url{https://doi.org/10.14778/1453856.1453901}
\showDOI{\tempurl}


\bibitem[Fern{\'{a}}ndez{-}{\'{A}}lvarez et~al\mbox{.}(2022)]%
        {DBLP:journals/kbs/Fernandez-Alvarez22}
\bibfield{author}{\bibinfo{person}{Daniel Fern{\'{a}}ndez{-}{\'{A}}lvarez}, \bibinfo{person}{Jos{\'{e}} Emilio~Labra Gayo}, {and} \bibinfo{person}{Daniel Gayo{-}Avello}.} \bibinfo{year}{2022}\natexlab{}.
\newblock \showarticletitle{Automatic extraction of shapes using sheXer}.
\newblock \bibinfo{journal}{\emph{Knowledge-Based Systems}}  \bibinfo{volume}{238} (\bibinfo{year}{2022}), \bibinfo{pages}{107975}.
\newblock
\urldef\tempurl%
\url{https://doi.org/10.1016/J.KNOSYS.2021.107975}
\showDOI{\tempurl}


\bibitem[Groz et~al\mbox{.}(2022)]%
        {DBLP:conf/icdt/GrozLSW22}
\bibfield{author}{\bibinfo{person}{Beno{\^{\i}}t Groz}, \bibinfo{person}{Aur{\'{e}}lien Lemay}, \bibinfo{person}{Slawek Staworko}, {and} \bibinfo{person}{Piotr Wieczorek}.} \bibinfo{year}{2022}\natexlab{}.
\newblock \showarticletitle{Inference of Shape Graphs for Graph Databases}. In \bibinfo{booktitle}{\emph{{ICDT}}} \emph{(\bibinfo{series}{LIPIcs}, Vol.~\bibinfo{volume}{220})}. \bibinfo{publisher}{Schloss Dagstuhl - Leibniz-Zentrum f{\"{u}}r Informatik}, \bibinfo{pages}{14:1--14:20}.
\newblock
\urldef\tempurl%
\url{https://doi.org/10.4230/LIPICS.ICDT.2022.14}
\showDOI{\tempurl}


\bibitem[Gutierrez et~al\mbox{.}(2011)]%
        {DBLP:journals/jcss/GutierrezHMP11}
\bibfield{author}{\bibinfo{person}{Claudio Gutierrez}, \bibinfo{person}{Carlos~A. Hurtado}, \bibinfo{person}{Alberto~O. Mendelzon}, {and} \bibinfo{person}{Jorge P{\'{e}}rez}.} \bibinfo{year}{2011}\natexlab{}.
\newblock \showarticletitle{Foundations of Semantic Web databases}.
\newblock \bibinfo{journal}{\emph{J. Comput. Syst. Sci.}} \bibinfo{volume}{77}, \bibinfo{number}{3} (\bibinfo{year}{2011}), \bibinfo{pages}{520--541}.
\newblock


\bibitem[Jacobs et~al\mbox{.}(1982)]%
        {DBLP:journals/tods/JacobsAK82}
\bibfield{author}{\bibinfo{person}{Barry~E. Jacobs}, \bibinfo{person}{Alan~R. Aronson}, {and} \bibinfo{person}{Anthony~C. Klug}.} \bibinfo{year}{1982}\natexlab{}.
\newblock \showarticletitle{On Interpretations of Relational Languages and Solutions to the Implied Constraint Problem}.
\newblock \bibinfo{journal}{\emph{{ACM} Trans. Database Syst.}} \bibinfo{volume}{7}, \bibinfo{number}{2} (\bibinfo{year}{1982}), \bibinfo{pages}{291--315}.
\newblock
\urldef\tempurl%
\url{https://doi.org/10.1145/319702.319730}
\showDOI{\tempurl}


\bibitem[Klug and Price(1982)]%
        {DBLP:journals/tods/KlugP82}
\bibfield{author}{\bibinfo{person}{Anthony~C. Klug} {and} \bibinfo{person}{Rod Price}.} \bibinfo{year}{1982}\natexlab{}.
\newblock \showarticletitle{Determining View Dependencies Using Tableaux}.
\newblock \bibinfo{journal}{\emph{{ACM} Trans. Database Syst.}} \bibinfo{volume}{7}, \bibinfo{number}{3} (\bibinfo{year}{1982}), \bibinfo{pages}{361--380}.
\newblock
\urldef\tempurl%
\url{https://doi.org/10.1145/319732.319738}
\showDOI{\tempurl}


\bibitem[Knublauch and Kontokostas(2017)]%
        {bibshacl}
\bibfield{author}{\bibinfo{person}{Holger Knublauch} {and} \bibinfo{person}{Dimitris Kontokostas}.} \bibinfo{year}{2017}\natexlab{}.
\newblock \bibinfo{title}{{Shapes Constraint Language (SHACL)}}.
\newblock
\newblock
\urldef\tempurl%
\url{https://www.w3.org/TR/shacl/}
\showURL{%
Retrieved 12.02.2024 from \tempurl}


\bibitem[Kostylev et~al\mbox{.}(2015)]%
        {DBLP:conf/icdt/KostylevRU15}
\bibfield{author}{\bibinfo{person}{Egor~V. Kostylev}, \bibinfo{person}{Juan~L. Reutter}, {and} \bibinfo{person}{Mart{\'{\i}}n Ugarte}.} \bibinfo{year}{2015}\natexlab{}.
\newblock \showarticletitle{{CONSTRUCT} Queries in {SPARQL}}. In \bibinfo{booktitle}{\emph{Proc. of International Conference on Database Theory, {ICDT}}} \emph{(\bibinfo{series}{LIPIcs}, Vol.~\bibinfo{volume}{31})}. \bibinfo{publisher}{Schloss Dagstuhl - Leibniz-Zentrum f{\"{u}}r Informatik}, \bibinfo{pages}{212--229}.
\newblock
\urldef\tempurl%
\url{https://doi.org/10.4230/LIPIcs.ICDT.2015.212}
\showDOI{\tempurl}


\bibitem[Leinberger et~al\mbox{.}(2019)]%
        {DBLP:conf/semweb/LeinbergerSSLS19}
\bibfield{author}{\bibinfo{person}{Martin Leinberger}, \bibinfo{person}{Philipp Seifer}, \bibinfo{person}{Claudia Schon}, \bibinfo{person}{Ralf L{\"{a}}mmel}, {and} \bibinfo{person}{Steffen Staab}.} \bibinfo{year}{2019}\natexlab{}.
\newblock \showarticletitle{Type Checking Program Code Using {SHACL}}. In \bibinfo{booktitle}{\emph{Proc. of {ISWC}}} \emph{(\bibinfo{series}{LNCS}, Vol.~\bibinfo{volume}{11778})}. \bibinfo{publisher}{Springer}, \bibinfo{pages}{399--417}.
\newblock
\urldef\tempurl%
\url{https://doi.org/10.1007/978-3-030-30793-6_23}
\showDOI{\tempurl}


\bibitem[Mihindukulasooriya et~al\mbox{.}(2018)]%
        {DBLP:conf/sac/Mihindukulasooriya18}
\bibfield{author}{\bibinfo{person}{Nandana Mihindukulasooriya}, \bibinfo{person}{Mohammad Rifat~Ahmmad Rashid}, \bibinfo{person}{Giuseppe Rizzo}, \bibinfo{person}{Ra{\'{u}}l Garc{\'{\i}}a{-}Castro}, \bibinfo{person}{{\'{O}}scar Corcho}, {and} \bibinfo{person}{Marco Torchiano}.} \bibinfo{year}{2018}\natexlab{}.
\newblock \showarticletitle{{RDF} Shape Induction Using Knowledge Base Profiling}. In \bibinfo{booktitle}{\emph{Prov. of the Symposium on Applied Computing}}. \bibinfo{publisher}{{ACM}}, \bibinfo{pages}{1952--1959}.
\newblock
\urldef\tempurl%
\url{https://doi.org/10.1145/3167132.3167341}
\showDOI{\tempurl}


\bibitem[Omran et~al\mbox{.}(2023)]%
        {DBLP:journals/semweb/OmranTMH23}
\bibfield{author}{\bibinfo{person}{Pouya~Ghiasnezhad Omran}, \bibinfo{person}{Kerry Taylor}, \bibinfo{person}{Sergio Jos{\'{e}}~Rodr{\'{\i}}guez M{\'{e}}ndez}, {and} \bibinfo{person}{Armin Haller}.} \bibinfo{year}{2023}\natexlab{}.
\newblock \showarticletitle{Learning {SHACL} shapes from knowledge graphs}.
\newblock \bibinfo{journal}{\emph{Semantic Web}} \bibinfo{volume}{14}, \bibinfo{number}{1} (\bibinfo{year}{2023}), \bibinfo{pages}{101--121}.
\newblock
\urldef\tempurl%
\url{https://doi.org/10.3233/SW-223063}
\showDOI{\tempurl}


\bibitem[Pandit et~al\mbox{.}(2018)]%
        {DBLP:conf/semweb/PanditOL18a}
\bibfield{author}{\bibinfo{person}{Harshvardhan~J. Pandit}, \bibinfo{person}{Declan O'Sullivan}, {and} \bibinfo{person}{Dave Lewis}.} \bibinfo{year}{2018}\natexlab{}.
\newblock \showarticletitle{Using Ontology Design Patterns To Define {SHACL} Shapes}. In \bibinfo{booktitle}{\emph{Proc. of the Workshop on Ontology Design and Patterns ({WOP} 2018) co-located with {ISWC} 2018}} \emph{(\bibinfo{series}{{CEUR} Workshop Proceedings}, Vol.~\bibinfo{volume}{2195})}. \bibinfo{publisher}{CEUR-WS.org}, \bibinfo{pages}{67--71}.
\newblock
\urldef\tempurl%
\url{https://ceur-ws.org/Vol-2195/research\_paper\_3.pdf}
\showURL{%
\tempurl}


\bibitem[Prud'hommeaux et~al\mbox{.}(2014)]%
        {DBLP:conf/i-semantics/PrudhommeauxGS14}
\bibfield{author}{\bibinfo{person}{Eric Prud'hommeaux}, \bibinfo{person}{Jos{\'{e}} Emilio~Labra Gayo}, {and} \bibinfo{person}{Harold~R. Solbrig}.} \bibinfo{year}{2014}\natexlab{}.
\newblock \showarticletitle{Shape expressions: an {RDF} validation and transformation language}. In \bibinfo{booktitle}{\emph{Proc. of the International Conference on Semantic Systems, SEMANTiCS}}. \bibinfo{publisher}{{ACM}}, \bibinfo{pages}{32--40}.
\newblock
\urldef\tempurl%
\url{https://doi.org/10.1145/2660517.2660523}
\showDOI{\tempurl}


\bibitem[Rabbani et~al\mbox{.}(2023)]%
        {DBLP:journals/pvldb/RabbaniLH23}
\bibfield{author}{\bibinfo{person}{Kashif Rabbani}, \bibinfo{person}{Matteo Lissandrini}, {and} \bibinfo{person}{Katja Hose}.} \bibinfo{year}{2023}\natexlab{}.
\newblock \showarticletitle{Extraction of Validating Shapes from very large Knowledge Graphs}.
\newblock \bibinfo{journal}{\emph{Proc. {VLDB} Endow.}} \bibinfo{volume}{16}, \bibinfo{number}{5} (\bibinfo{year}{2023}), \bibinfo{pages}{1023--1032}.
\newblock
\urldef\tempurl%
\url{https://doi.org/10.14778/3579075.3579078}
\showDOI{\tempurl}


\bibitem[Reiter(1982)]%
        {DBLP:conf/db-workshops/Reiter82}
\bibfield{author}{\bibinfo{person}{Raymond Reiter}.} \bibinfo{year}{1982}\natexlab{}.
\newblock \showarticletitle{Towards a Logical Reconstruction of Relational Database Theory}. In \bibinfo{booktitle}{\emph{On Conceptual Modelling (Intervale)}} \emph{(\bibinfo{series}{Topics in Information Systems})}. \bibinfo{publisher}{Springer}, \bibinfo{pages}{191--233}.
\newblock


\bibitem[Seifer et~al\mbox{.}(2024)]%
        {darus-3977_2024}
\bibfield{author}{\bibinfo{person}{Philipp Seifer}, \bibinfo{person}{Daniel Hernández}, \bibinfo{person}{Ralf Lämmel}, {and} \bibinfo{person}{Steffen Staab}.} \bibinfo{year}{2024}\natexlab{}.
\newblock \bibinfo{title}{{Code for From Shapes to Shapes}}.
\newblock
\newblock
\urldef\tempurl%
\url{https://doi.org/10.18419/darus-3977}
\showDOI{\tempurl}


\bibitem[Seifer et~al\mbox{.}(2021)]%
        {DBLP:conf/semweb/SeiferLS21}
\bibfield{author}{\bibinfo{person}{Philipp Seifer}, \bibinfo{person}{Ralf L{\"{a}}mmel}, {and} \bibinfo{person}{Steffen Staab}.} \bibinfo{year}{2021}\natexlab{}.
\newblock \showarticletitle{ProGS: Property Graph Shapes Language}. In \bibinfo{booktitle}{\emph{Proc. of {ISWC}}} \emph{(\bibinfo{series}{LNCS}, Vol.~\bibinfo{volume}{12922})}. \bibinfo{publisher}{Springer}, \bibinfo{pages}{392--409}.
\newblock
\urldef\tempurl%
\url{https://doi.org/10.1007/978-3-030-88361-4_23}
\showDOI{\tempurl}


\bibitem[Sequeda et~al\mbox{.}(2012)]%
        {DBLP:conf/www/SequedaAM12}
\bibfield{author}{\bibinfo{person}{Juan~F. Sequeda}, \bibinfo{person}{Marcelo Arenas}, {and} \bibinfo{person}{Daniel~P. Miranker}.} \bibinfo{year}{2012}\natexlab{}.
\newblock \showarticletitle{On directly mapping relational databases to {RDF} and {OWL}}. In \bibinfo{booktitle}{\emph{{WWW}}}. \bibinfo{publisher}{{ACM}}, \bibinfo{pages}{649--658}.
\newblock
\urldef\tempurl%
\url{https://doi.org/10.1145/2187836.2187924}
\showDOI{\tempurl}


\bibitem[Spahiu et~al\mbox{.}(2018)]%
        {DBLP:conf/semweb/SpahiuMP18}
\bibfield{author}{\bibinfo{person}{Blerina Spahiu}, \bibinfo{person}{Andrea Maurino}, {and} \bibinfo{person}{Matteo Palmonari}.} \bibinfo{year}{2018}\natexlab{}.
\newblock \showarticletitle{Towards Improving the Quality of Knowledge Graphs with Data-driven Ontology Patterns and {SHACL}}. In \bibinfo{booktitle}{\emph{Proc. of the Workshop on Ontology Design and Patterns ({WOP} 2018) co-located with {ISWC} 2018}} \emph{(\bibinfo{series}{{CEUR} Workshop Proceedings}, Vol.~\bibinfo{volume}{2195})}. \bibinfo{publisher}{CEUR-WS.org}, \bibinfo{pages}{52--66}.
\newblock
\urldef\tempurl%
\url{https://ceur-ws.org/Vol-2195/research\_paper\_2.pdf}
\showURL{%
\tempurl}


\bibitem[Stonebraker(1975)]%
        {DBLP:conf/sigmod/Stonebraker75}
\bibfield{author}{\bibinfo{person}{Michael Stonebraker}.} \bibinfo{year}{1975}\natexlab{}.
\newblock \showarticletitle{Implementation of Integrity Constraints and Views by Query Modification}. In \bibinfo{booktitle}{\emph{Proc. of {SIGMOD}}}. \bibinfo{publisher}{{ACM}}, \bibinfo{pages}{65--78}.
\newblock
\urldef\tempurl%
\url{https://doi.org/10.1145/500080.500091}
\showDOI{\tempurl}


\bibitem[Thapa and Giese(2021)]%
        {DBLP:conf/semweb/ThapaG21}
\bibfield{author}{\bibinfo{person}{Ratan~Bahadur Thapa} {and} \bibinfo{person}{Martin Giese}.} \bibinfo{year}{2021}\natexlab{}.
\newblock \showarticletitle{A Source-to-Target Constraint Rewriting for Direct Mapping}. In \bibinfo{booktitle}{\emph{Proc. of {ISWC}}} \emph{(\bibinfo{series}{{LNCS}}, Vol.~\bibinfo{volume}{12922})}. \bibinfo{publisher}{Springer}, \bibinfo{pages}{21--38}.
\newblock
\urldef\tempurl%
\url{https://doi.org/10.1007/978-3-030-88361-4\_2}
\showDOI{\tempurl}


\bibitem[Thapa and Giese(2022)]%
        {DBLP:conf/semweb/ThapaG22}
\bibfield{author}{\bibinfo{person}{Ratan~Bahadur Thapa} {and} \bibinfo{person}{Martin Giese}.} \bibinfo{year}{2022}\natexlab{}.
\newblock \showarticletitle{Mapping Relational Database Constraints to {SHACL}}. In \bibinfo{booktitle}{\emph{Proc. of ISWC}} \emph{(\bibinfo{series}{LNCS}, Vol.~\bibinfo{volume}{13489})}. \bibinfo{publisher}{Springer}, \bibinfo{pages}{214--230}.
\newblock
\urldef\tempurl%
\url{https://doi.org/10.1007/978-3-031-19433-7_13}
\showDOI{\tempurl}


\end{thebibliography}

\reportorpaper{\
  \appendix
  \section{Structure of the Appendix}
\label{a:structure}

This appendix is structured as follows.
In \Cref{a:implementation} we give details about our implementation, and a feasibility experiment investigating runtime performance.
\Cref{a:examples} gives extended versions of the running examples used throughout the paper.
\Cref{a:proofs} contains the full proofs for all propositions from the main paper.
Finally, \Cref{a:extension} details how the method from the main paper can be extended for more general types of SHACL shapes, 
and \Cref{a:proofsext} gives proofs related to these extensions.

  \reportorpaper{\section{Implementation Runtime Evaluation}}{\section{Appendix}}
\label{a:implementation}

\subsection{Implementation Overview}

We implemented \Cref{alg:outshapes}, relying on a straightforward translation of \Cref{alg:main} to Scala for validation with respect to a single candidate shape, and a generator for candidates based on the syntax of Simple SHACL.
For reasoning tasks, our implementation supports any OWL API\footnotemark[3] reasoner.
In particular, we rely on the HermiT~\cite{hermit} reasoner as a default.

Our implementation features tools for generative exploration regarding query and vocabulary size, a setup for performance evaluation and results (see also the remainder of this section), a test suite, and the examples from this paper in mechanized form.
To this end, the implementation features a command-line application for parsing and inferring result shapes for examples (e.g., those supplied in the project repository), taking SCCQ and (Simple) SHACL shapes as input, both in DL and JSON-LD syntax.
The implementation also features a library for applying the shapes-to-shapes method, supporting both Simple SHACL and $\DLogics$-based shapes as input, and returning inferred axioms in internal datastructures that can then be utilized in a wider range of reasoning tasks.
Additionally, we also include a tutorial giving a more intuitive and hands-on introduction to our method.
The implementation~\cite{darus-3977_2024} is available under a free software license on GitHub\footnotemark[4].

\footnotetext[3]{\url{https://github.com/owlcs/owlapi}}
\footnotetext[4]{\url{https://github.com/softlang/s2s}}

\subsection{Evaluation: Feasibility}

We show feasibility of our method and implementation through the following evaluation of runtime performance.
This evaluation is based on randomly generated problems (i.e., sets of input shapes as well as queries).
Thus, the experiment discussed here shows basic feasibility of our method with synthetic data, though results on real-world data may differ.

As performance depends largely on the reasoning task (see also below), results can differ based on the reasoner implementation used. 
Multiple reasoners are available with our implementation, and their respective optimization strategies may differ.
Some reasoner implementations or optimization strategies are not deterministic; 
Thus, as a simple optimization, our implementation can abort runs with a set timeout and retry computing results, in order to avoid unlucky models for non-deterministic reasoner optimization strategies (this is reported in the results).

\paragraph{Experimental Setup}

We refer to the project setup (in particular, the specification in \texttt{build.sbt}) with respect to versions of the respective software, et. al.
Beyond that, we run experiments with Microsoft JDK build \texttt{openjdk 17.0.7 2023-04-18 LTS} on Windows 10 Pro (Version 10.0.19045), on commodity hardware (Intel i5-6600K @ 3.5GHz, 16GB RAM).

We define the following three sample configurations and give the number of atomic patterns per query (for template and pattern each) as well as the number of input shapes:

\begin{itemize}
    \item \texttt{SMALL} 1-2 templates and patterns, 1-2 shapes.
    \item \texttt{MEDIUM} 5-7 templates and patterns, 5-7 shapes.
    \item \texttt{LARGE} 11-13 templates and patterns, 11-13 shapes
\end{itemize}

As a basis for these scenarios, we refer to the following real-world query datasets (logs), where most queries (more than $90\%$) have fewer than 6 or 7 patterns~\cite{BGK2018,DBLP:journals/vldb/BonifatiMT20}, relating to our \texttt{MEDIUM} configuration.
More than half include only on pattern, relating to our \texttt{SMALL} configuration.
Our \texttt{LARGE} configuration covers outliers of very large queries (less than $1\%$ of real-world queries).

For all samples, we draw fresh variables (per pattern) with a probability of $0.5$ and fresh concepts or role names with a probability of $0.8$, and sample property versus concept atomic patterns with a ratio of $0.3$.
We provide the full details on all parameters used for sampling with the implementation source code.

Note, that we generate shapes from the \emph{vocabulary of the query}.
Thus, the number of input shapes given here is not comparable to the size of usual sets of SHACL shapes in real-world datasets. 
That is, the sets of 1.5/6/12 shapes constrain the relatively small vocabulary of an input query rather tightly.
We do not know of any empirical data on the average number of shapes in the query (pattern) vocabulary, i.e., that apply to a particular query, thus we estimate the numbers as given above.

We run and measure $5.000$ samples each for the three given configurations, using a fixed seed for the random generator, and measure execution time for a single run of the algorithm per sample, after first running $100$ additional samples as warmup.
This experiment uses the HermiT\cite{hermit} reasoner.

\paragraph{Results of the Experiment}

A summary of results is given in Table~1.
The full output with input shapes, input queries, fine-grained execution metrics, as well as output shapes is included with the project source code as a \texttt{CSV} file and a summary report.
Both the full output as well as the summary can be generated by executing the \texttt{profile} main method (see the implementation documentation for details on how to re-run the experiment as-is, or modify it across various dimensions).

\begin{table}
    \label{tab:perf-results}
    \caption{Results (average and median execution time in milliseconds without timeouts, the number of timeouts (limit: 10 minutes), as well as percentage of processing time spent on reasoning) for \texttt{SMALL}, \texttt{MEDIUM} and \texttt{LARGE} configurations.}
\begin{center}
\begin{tabular}{ l r r r r } 
    \toprule
    Configuration & Average & Median & T/O & Reasoning\\
    \midrule
    \texttt{SMALL} & 3 & 0 & 0 & 38,42\%\\ 
    \texttt{MEDIUM} & 40 & 20 & 0 & 87,11\%\\ 
    \texttt{LARGE} & 693 & 243 & 20 & 97,66\%\\ 
    \bottomrule
\end{tabular}
\end{center}
\end{table}

\paragraph{Interpretation}

We show in this experiment the basic feasibility of our method, with average and median execution times for even very large samples of less than one second. 
For the largest samples, few ($0.4\%$) samples time out, with a set timeout of 10 minutes.
We hypothesize that this is due to the reasoner sometimes choosing an unlucky model, where reasoning takes a very long time; 
though this could also perhaps be caused by bugs in the reasoner implementation.
Indeed, the majority of time is spent on reasoning for larger configuration (see Table~1) and by detailed inspection of the full log, this holds true for timeouts as well.

  \reportorpaper{\section{Extended Examples}}{\subsection{Extended Examples}}
\label{a:examples}

In this section, we extend upon the running examples.
We first provide additional details on the running example incorporated in the body of the paper, including the example using concrete SHACL and SPARQL syntax.
Next, we extend upon the running example by giving additional example queries.

\paragraph{Implementation}

All examples from the main paper and from this section are also provided in mechanized form with the implementation.
To this end, the implementation contains \texttt{.shacl} (as well \texttt{.json}) and \texttt{.sparql} files, where example shapes are included in formal description logics and JSON-LD syntax, and SCCQ queries in concrete SPARQL syntax.
We refer to the documentation (\texttt{README.md}) for more details on running example instances, and obtaining different kinds of outputs.

\paragraph{\texorpdfstring{Running Example: Full Set of Shapes $\shapesout$}{Running Example: Full Set of Output Shapes}}

The full set of output shapes from \Cref{ex:formalrunning} is given below.
Note, that some shapes (such as tautologies and shapes trivially entailed by other shapes) are omitted.

\begin{example}\label{ex:fullrunning}
    Full output shapes for $\eqonen = \eqone$ (\Cref{ex:qone}) and the set of input shapes $\eS = \{\esone, \estwo, \esthree\}$ (\Cref{ex:sone}), as first introduced in \Cref{ex:formalrunning}.
    \begin{align*}
    &\eSout = \{\\
        &\eBon \sqsubseteq \forall \epoin.\eBon,
        &&\eBon \sqsubseteq \forall \epoin.\eEon,
        &&&\eBon \sqsubseteq \forall \epon.\eBon,\\
        &\eBon \sqsubseteq \exists \epoin.\eBon,
        &&\eBon \sqsubseteq \exists \epoin.\eEon,
        &&&\eEon \sqsubseteq \eBon,\\
        &\eEon \sqsubseteq \forall \epoin.\eBon,
        &&\eEon \sqsubseteq \forall \epoin.\eEon,
        &&&\eEon \sqsubseteq \forall \epon.\eBon,\\
        &\eEon \sqsubseteq \exists \epoin.\eBon,
        &&\eEon \sqsubseteq \exists \epoin.\eEon,
        &&&\eEon \sqsubseteq \exists \epon.\eBon,\\
        &\exists \epoin.\top \sqsubseteq \eBon,
        &&\exists \epoin.\top \sqsubseteq \forall \epoin.\eBon,
        &&&\exists \epoin.\top \sqsubseteq \forall \epoin.\eEon,\\
        &\exists \epoin.\top \sqsubseteq \forall \epon.\eBon,
        &&\exists \epoin.\top \sqsubseteq \exists \epoin.\eBon,
        &&&\exists \epoin.\top \sqsubseteq \exists \epoin.\eEon,\\
        &\exists \epon.\top \sqsubseteq \eBon,
        &&\exists \epon.\top \sqsubseteq \eEon,
        &&&\exists \epon.\top \sqsubseteq \forall \epoin.\eBon,\\
        &\exists \epon.\top \sqsubseteq \forall \epoin.\eEon,
        &&\exists \epon.\top \sqsubseteq \forall \epon.\eBon,
        &&&\exists \epon.\top \sqsubseteq \exists \epoin.\eBon,\\
        &\exists \epon.\top \sqsubseteq \exists \epoin.\eEon,
        &&\exists \epon.\top \sqsubseteq \exists \epon.\eBon\ \}
    \end{align*}
\end{example}

\paragraph{Running Example: Concrete Syntax}

We next give the running example (e.g., \Cref{ex:fullrunning}) in concrete SPARQL and SHACL (Turtle) syntax.
We assume the default prefix \texttt{:} for the example domain (unspecified), and prefix \texttt{sh:} for SHACL (i.e., bound to \url{http://www.w3.org/ns/shacl#}).

\begin{figure}
\begin{lstlisting}[language=SHACL]
:s1 a sh:NodeShape ;
    sh:targetClass :A ;
    sh:property [ 
        sh:path :p ;
        sh:qualifiedMinCount 1 ;
        sh:qualifiedValueShape [ 
            sh:class :B
        ] 
    ] .

:s2 a sh:NodeShape ;
    sh:targetSubjectsOf :r ;
    sh:class :B .

:s3 a sh:NodeShape ;
    sh:targetClass :B ;
    sh:class :E .
\end{lstlisting}
\label{fig:sh:one}
\caption{Shapes $\esonen$ ($\esone$), $\estwon$ ($\estwo$), and $\esthreen$ ($\esthree$) using a concrete SHACL syntax (Turtle).}
\end{figure}

\begin{figure}
\begin{lstlisting}[language=SPARQL]
CONSTRUCT {
  ?y a :E .
  ?y :p ?z .
  ?z a :B
} WHERE {
  ?w :p ?y .
  ?y a :B .
  ?x :p ?z .
  ?z a :E
}
\end{lstlisting}
\label{fig:q:one}
\caption{Example query $\eqonen = \eqone$ in concrete SPARQL syntax.}
\end{figure}

\paragraph{Additional Examples}

We now give additional examples problem instances, that is, queries and sets of input shapes, and (a subset of) the corresponding output shapes.
For the full output, as well as all intermediate components, i.e., the inferred axioms, we refer to the implementation, which renders full internal details via the \texttt{--debug} flag.
All examples are included with the implementation.

\anotherexample{4}{
    \sccquery{
        \atomC{\ex}{\eBon},
        \atomC{\ey}{\eAon},
    }{
        \atomC{\ex}{\eAi},
        \atomC{\ey}{\eBi},
    }
}{
    \eAi \sqsubseteq \eBi
}{
    \eBon \sqsubseteq \eAon
}{
    Query $q_4$ is a simple example, demonstrating how our method maintains subsumption relationships through renaming of concepts, in this simple case swapping the names $A$ and $B$.
    A core mechanism allowing this is the subsumption between the variable concept for query variables $\texttt{?x}$ and $\texttt{?y}$ entailed by our inferred axioms, which holds on all extended graphs for $q_4$ and $S_4$.
}

\anotherexample{5}{
    \sccquery{
        \atomC{\ex}{\eBon},
        \atomC{\ey}{\eAon},
    }{
        \atomC{\ex}{\eAi},
        \atomP{\ex}{\ey}{\epi},
        \atomC{\ey}{\eBi},
    }
}{
    \eBi \sqsubseteq \eAi, \eBi \sqsubseteq \exists \epi . \eBi
}{
    \eAon \sqsubseteq \eBon
}{
    Here, we continue with another example with the same, simple template as in the previous example.
    This simple template serves to demonstrate the consequences of variable concept subsumption between the variables $\texttt{?y}$ and $\texttt{?x}$ directly, as the output shape $\eAon \sqsubseteq \eBon$.
    This subsumption relationship results from the mapping step discussed in \Cref{sec:algorithm}: Since we know for all bindings of $\ey$ in the query, that both the pattern $\atomC{\ey}{\eAi}$ is always satisfied (since $\eBi \sqsubseteq \eAi$) and the same for $\atomP{\ey}{z}{\epi}, \atomC{z}{\eBi}$ (for some fresh variable $z$, since $\eBi \sqsubseteq \exists \epi . \eBi$), we can obtain a mapping resulting in subsumption $\vconcept{\ey} \sqsubseteq \vconcept{\ex}$.
}

\anotherexample{6}{
    \sccquery{
        \atomC{\ex}{\eAon},
        \atomC{\ey}{\eBon},
        \atomP{\ex}{\ey}{\epon}
    }{
        \atomC{\ex}{\eAi},
        \atomC{\ey}{\eBi},
    }
}{}{
    \eAon \sqsubseteq \forall \epoin.\eAon,
    \eAon \sqsubseteq \forall \epon.\eBon,
    \eAon \sqsubseteq \exists \epon.\eBon,
    \eBon \sqsubseteq \forall \epoin.\eAon,
    \eBon \sqsubseteq \forall \epon.\eBon,
    \eBon \sqsubseteq \exists \epoin.\eAon,
    \exists\epoin.\top \sqsubseteq \eBon,
    \exists\epoin.\top \sqsubseteq \forall \epoin.\eAon,
    \exists\epoin.\top \sqsubseteq \forall \epon.\eBon,
    \exists\epoin.\top \sqsubseteq \exists \epoin.\eAon,
    \exists\epon.\top \sqsubseteq \eAon,
    \exists\epon.\top \sqsubseteq \forall \epoin.\eAon,\\
    \exists\epon.\top \sqsubseteq \forall \epon.\eBon,
    \exists\epon.\top \sqsubseteq \exists \epon.\eBon,
}{%
    With this example, we demonstrate inference of shapes from the query (template) itself, without any given input shapes, resulting only from the closure assumptions used in the method.
    The query pattern simply introduces the variables $\texttt{?x}$ and $\texttt{?y}$ without any further context (arbitrary names $\eAi$ and $\eBi$).
    In the template, we introduce the additional role name $\epon$ between these two variables.
}

  \section{Proofs}
\label{a:proofs}

In this section, we present the full proofs for \Cref{prop:equivalent-formalisms} through \Cref{prop:properties}, and introduce \Cref{theorem:np} (and its proof) as well as \Cref{prop:nm} (and its proof).

\subsection{\texorpdfstring{Proof for \Cref{prop:equivalent-formalisms}}{Proof for Proposition 1}}
\label{proofs:equivalent-formalisms}

In order to prove \Cref{prop:equivalent-formalisms} we need to show that every model of the validation knowledge base of a graph is isomorphic to the canonical model of the graph.
To this end, we introduce the following lemma.

\begin{lemma}
  \label{lemma:isomorphic-models}
  Let $G$ be a graph, $\Int_G$ the canonical interpretation of $G$, and $(\TBox_G, G)$ the validation knowledge base of $G$. Then, all models $\Int$ of $(\TBox_G, G)$ are isomorphic to $\Int_G$.
\end{lemma}

\begin{proof}
  The fact that $\Int$ and $\Int_G$ are isomorphic follows from the existence of a function $f:\Delta^{\Int_G} \to \Delta^\Int$ satisfying the following properties:
  \begin{description}
  \item[P.1] Function $f$ is bijective.
  \item[P.2]
    $f(a^{\Int_G}) = a^\Int$ for every $a \in \IndividualNames$.
  \item[P.3]
    $\{f(x) \mid x \in A^{\Int_G}\} = A^\Int$ for every $A \in \ConceptNames$.
  \item[P.4]
    $\{(f(x), f(y)) \mid (x, y) \in r^{\Int_G}\} =  r^\Int$ for every $r \in \RoleNames$.
  \end{description}
  We next prove properties \textbf{P.1} to \textbf{P.4} for function $f$.
  \begin{description}
  \item[Proof for P.1:]
    Let $f:\Delta^{\Int_G} \to \Delta^\Int$ be the function defined as $f(a)=a^\Int$, for every individual name $a\in \IndividualNames$. 
    Function $f$ is well-defined because, by definition of $\Int_G$, $\Delta^{\Int_G} = \IndividualNames$. 
    To show that function $f$ is bijective, it suffices to prove that $f$ is injective and surjective:
    \begin{description}
    \item[Surjective:] 
      The domain closure assumption axioms in the knowledge base $(\TBox_G, G)$ imply that $\Delta^\Int = \bigcup_{a \in \IndividualNames}\{a^\Int\}$.
      Then, for every element $e \in \Delta^\Int$, there exists an individual name $\{a\}$ such that $e \in \{a\}^\Int$. 
      That is, $f(a) = e$. 
      Hence, $f$ is surjective.
    \item[Injective:] 
      The unique-name assumption axioms in the knowledge base $(\TBox_G, G)$ imply that $\Int \models \{b\} \sqcap \{a\} \equiv \bot$ for every pair of distinct individual names $a$ and $b$. 
      That is, $f(a) \neq f(b)$.
      Hence, $f$ is injective.
    \end{description}
  \item[Proof for P.2:]
    Let $a \in \IndividualNames$ be an arbitrary individual name.
    By definition of $\Int_G$, $a^{\Int_G} = a$.
    By definition of $f$, $f(a) = a^\Int$.
    Hence, combining both identities, we obtain the identity $f(a^{\Int_G}) = a^\Int$.
  \item[Proof for P.3:]
    Let $A \in \ConceptNames$ be an arbitrary concept name. 
    By definition of $\Int_G$, $A^{\Int_G} = \{a \mid \atomC{a}{A} \in G\}$.
    The closed-world assumption axioms in the knowledge base $(\TBox_G, G)$ imply that $A^{\Int} = \bigcup_{\atomC{a}{A} \in G}\{a^{\Int}\}$.
    That is, $A^\Int = \{a^\Int \mid \atomC{a}{A} \in G\}$.
    Since $f(a) = a^\Int$, we conclude that $\{f(a) \mid a \in A^{\Int_G}\} = A^\Int$.
  \item[Proof for P.4:]
    Let $r \in \RoleNames$ be an arbitrary role name.
    By definition of $\Int_G$, $r^{\Int_G} = \{(a,b) \mid \atomP{a}{b}{r} \in G\}$.
    The closed-world assumption axioms in the knowledge base $(\TBox_G, G)$ imply that $(\exists r.\{b\})^\Int = \bigcup_{\atomP{a}{b}{r} \in G}\{a^\Int\}$.
    That is, $r^\Int = \{(a^\Int, b^\Int) \mid \atomP{a}{b}{r} \in G\}\}$.
    Since $f(a) = a^\Int$ and $f(b) = b^\Int$, we conclude that $\{(f(a), f(b)) \mid (a,b) \in r^{\Int_G}\} = r^\Int$.
  \end{description}
  Hence, we have proved the lemma.
\end{proof}

\begin{proof}[Proof of \Cref{prop:equivalent-formalisms}]
  This proof follows from \Cref{lemma:isomorphic-models}, which states that $(\TBox_G, G)$ has a unique model up to isomorphism, namely $\Int_G$;
  thus for every set $\Sigma$ of $\DLogics$ axioms, $\Int_G \models \Sigma$ if and only if $\Int_G \models (\TBox_G \cup \Sigma, G)$.
  That is, $\Int_G \models S$ if and only if $(\TBox_G \cup S, G)$ is consistent.
  Hence, statements (i) and (ii) are equivalent.
  Similarly, statements (ii) are (iii) are equivalent because $(\TBox_G, G)$ has a unique model up to isomorphism.
  In general, given two sets of axioms $\Sigma_1$ and $\Sigma_2$, the consistence of $(\Sigma_1, G)$ and $(\Sigma_2, G)$ does not imply the consistency of $(\Sigma_1 \cup \Sigma_2, G)$
  because the sets models of $(\Sigma_1, G)$ and $(\Sigma_2, G)$ can be non-empty and disjoint.
  However, in this case the implication is true because $(\TBox_G, G)$ admits a single model up to isomorphism.
\end{proof}

\subsection{\texorpdfstring{Proof for \Cref{prop:relevant-shapes}}{Proof for Proposition 2}}
\label{proof:relevant-shapes}

In order to prove \Cref{prop:relevant-shapes}, we show by contraposition that if a given shape $s$ does not satisfy the conditions of the proposition, then it is irrelevant.

To this end, we introduce two lemmas, relating the structure of a concept expression $C$ with the vocabulary of $C$ and of a given graph $G$.
First, we consider concept names and existential quantification.

\begin{lemma}\label{prop:concept-empty}
  Let $C$ be a concept description defined as follows:
  \[
    C ::= A \mid \exists p.\top \mid \exists p^-.\top \mid
    \exists p.A \mid \exists p^-.A\;,
  \]
  where $A$ is a concept name, and $p$ is a role name. Let $G$ be a Simple RDF graph, and $(\TBox_G, G)$ be the validation knowledge base of $G$. Then, $\voc(C) \not\subseteq \voc(G)$ implies $(\TBox_G, G) \models C \equiv \bot$.
\end{lemma}

\begin{proof}
  Let $\Int$ be a model of $(\TBox_G, G)$. We will prove this lemma by proving the contraposition: the existence of an individual $a^\Int\in C^\Int$ implies that $\voc(C) \subseteq \voc(G)$. By \Cref{lemma:isomorphic-models}, every model $\Int$ of $(\TBox_G, G)$ is isomorphic to the canonical model of $G$. The proof follows case by case:
  \begin{enumerate}
  \item If $C$ is $A$ then, $\atomC{a}{A} \in G$. Hence, $\voc(C) \subseteq \voc(G)$.
  \item If $C$ is $\exists p.\top$ or $\exists p^-.\top$, then there is an individual name $b$ such that $\atomP{a}{b}{p} \in G$ or $\atomP{b}{a}{p} \in G$. Hence, $\voc(C) \subseteq \voc(G)$.
  \item If $C$ is $\exists p.A$ or $\exists p^-.A$, then there is an individual name $b$ such that $\atomC{b}{A} \in G$, and $\atomP{a}{b}{p} \in G$ or $\atomP{b}{a}{p} \in G$. Hence, $\voc(C) \subseteq \voc(G)$.
  \end{enumerate}
  Hence, we prove the lemma by contraposition.
\end{proof}

Next, we consider universal quantification.

\begin{lemma}\label{prop:forall-empty}
  Let $C$ be a concept description defined as follows:
  \[
    C ::= \forall p.A \mid \forall p^-.A\;,
  \]
  where $A$ is a concept name, and $p$ is a role name. Let $G$ be a Simple RDF graph, and $(\TBox_G, G)$ be the validation knowledge base of $G$. Then the following holds.
  \begin{enumerate}
  \item
    If $p \notin \voc(G)$ then $(\TBox_G, G) \models C \equiv \top$.
  \item
    If $A \notin \voc(G)$ and $C$ is $\forall p.A$ then \\ $(\TBox_G, G) \models C \equiv \neg(\exists p.\top)$.
  \item
    If $A \notin \voc(G)$ and $C$ is $\forall p^-.A$ then \\ $(\TBox_G, G) \models C \equiv \neg(\exists p^-.\top)$.
  \end{enumerate}
\end{lemma}

\begin{proof}
  Let $\Int$ be a model of the validation knowledge base of graph $G$.
  We prove this lemma using the equivalencies $\forall p.A \equiv \neg\exists p.\neg A$ and $\forall p^-.A \equiv \neg\exists p^-.\neg A$.
  \begin{enumerate}
  \item
    If $p \notin \voc(G)$ then $p^\Int$ is empty, since $\Int$ is isomorphic to the canonical interpretation of $G$ (\Cref{lemma:isomorphic-models}).
    Thus, every element in the domain $\Delta^\Int$ belongs to concepts $\neg\exists p.\neg A$ and $\neg\exists p^-.\neg A$. Hence, $\Int \models C \equiv \top$.
  \item
    If $A \notin \voc(G)$ then $A^\Int$ is empty, since $\Int$ is isomorphic to the canonical interpretation of $G$ (\Cref{lemma:isomorphic-models}).
    Then, $\Int \models \neg A \equiv \top$.
    Hence,
    \begin{enumerate}
    \item
      If $C$ is $\forall p.A$, then $\Int \models C \equiv \neg(\exists p.\top)$.
    \item
      If $C$ is $\forall p^-.A$, then $\Int \models C \equiv \neg(\exists p^-.\top)$.
    \end{enumerate}
  \end{enumerate}
  Since $\Int$ is an arbitrary model of $(\TBox_G, G)$, we conclude this proof.
\end{proof}

\begin{proof}[Proof of \Cref{prop:relevant-shapes}]
  Let $\psi \sqsubseteq \phi$ be a Simple SHACL shape, $q$ be a \sccqname, and $G$ be a Simple RDF graph with $\voc(G) \subseteq \voc(q)$, and $(\TBox_G, G)$ be the validation knowledge base of graph $G$. We have the following disjoint cases:
  \begin{enumerate}
  \item
    Case $\voc(\psi) \not\subseteq \voc(q)$.  Then, by \Cref{prop:concept-empty}, $(\TBox_G, G) \models \psi \equiv \bot$ (since $\psi$ is, per definition, restricted to one of the cases covered in the lemma). Hence, shape $\psi \sqsubseteq \phi$ is not relevant (\Cref{def:irrelevant-shape}).
  \item
    Case $\voc(\psi) \subseteq \voc(q)$ and $\phi$ has the form $\forall p.A$ or $\forall p^-.A$. We have the following subcases:
    \begin{enumerate}
    \item
      Case $p \notin \voc(q)$. Then, by \Cref{prop:forall-empty}, $(\TBox_G, G) \models \phi \equiv \top$. Hence, shape $\psi \sqsubseteq \phi$ is not relevant.
    \item
      Case $p \in \voc(q)$ and $A \not\in \voc(G)$. Then, by \Cref{prop:concept-empty}, $(\TBox_G, G) \models \neg A \equiv \top$. We have the following subcases:
      \begin{enumerate}
      \item
        Case $\phi$ is $\forall p.A$, then $(\TBox_G, G) \models \phi \equiv \neg(\exists p.\top)$.
      \item
        Case $\phi$ is $\forall p^-.A$, then $(\TBox_G, G) \models \phi \equiv \neg(\exists p^-.\top)$.
      \end{enumerate}
    \end{enumerate}
  \item 
    Case $\voc(\psi) \subseteq \voc(q)$ and $\voc(\phi) \not\subseteq \voc(q)$ and $\phi$ has not the form $\forall p.A$ or $\forall p^-.A$.
    Then, $\phi$ has one of the forms covered in \Cref{prop:concept-empty}, and by this lemma, $(\TBox_G, G) \models \phi \equiv \bot$.
    Therefore, the shape is not relevant.
  \item
    Case $\voc(\psi) \subseteq \voc(q)$ and $\voc(\phi) \subseteq \voc(q)$. We have the following subcases:
    \begin{enumerate}
    \item
      Shape $\psi \sqsubseteq \phi$ has the form $A \sqsubseteq A$. Then, the shape is not relevant because it is a tautology.
    \item
      Shape $\psi \sqsubseteq \phi$ has not the form $A \sqsubseteq A$.
    \end{enumerate}
  \end{enumerate}
  Hence, we have shown that in all cases, except for those mentioned in \Cref{prop:relevant-shapes} (2.b.i, 2.b.ii, and 4.b), the shapes are not relevant.
  Hence, for all relevant shapes, the properties in \Cref{prop:relevant-shapes} hold.
\end{proof}

\subsection{\texorpdfstring{Proof for \Cref{prop:reduction-extended-graph}}{Proof for Proposition 3}}
\label{proof:reduction-extended-graph}

\begin{proof}
  We prove first the second case of \Cref{prop:reduction-extended-graph}, namely $\shaclvalid{\graphmed}{\{\varphi\}}$ if and only if $\shaclvalid{\graphext}{\{\dot{\varphi}\}}$.
  The proofs for the other two cases work exactly analogously, since all three subgraphs $\graphin$, $\graphmed$ and $\graphout$ form distinct namespaces.

  \newcommand{\IntExt}{\Int_{\mathrm{ext}}}
  \newcommand{\IntMed}{\Int_{\mathrm{med}}}
  \newcommand{\IntIn}{\Int_{\mathrm{in}}}
  \newcommand{\IntOut}{\Int_{\mathrm{out}}}
  \newcommand{\dgraphmed}{\dot{G}_{\mathrm{med}}}
  \newcommand{\dgraphout}{\ddot{G}_{\mathrm{out}}}
  \newcommand{\dgraphin}{G_{\mathrm{in}}}

  Let $\IntExt$ and $\IntMed$ be the canonical models of $\graphmed$ and $\graphext$, respectively.
  To prove this case, it suffices to show that for every axiom $\varphi$ not including any names with dots (e.g., $\dot{A}$, $\ddot{A}$, $\dot{p}$, or $\ddot{p}$), and every concept expression $C$ occurring in $\varphi$, $C^{\IntExt} = C^{\IntMed}$.
  Indeed, if this is true for every arbitrary concept expression $C$ in $\varphi$, then for every such axiom $\varphi$, $\shaclvalid{\graphmed}{\{\varphi\}}$ if and only $\shaclvalid{\graphext}{\{\dot{\varphi}\}}$.
  By construction of $\graphext$, for every concept assertion $\atomC{a}{A} \in \graphmed$, $\atomC{a}{\dot{A}} \in \graphext$ if and only if $\atomC{a}{\dot{A}} \in \dot{G}_{\mathrm{med}}$, and for every role assertion $\atomP{a}{b}{p} \in \graphmed$, $\atomP{a}{b}{\dot{p}} \in \graphext$ if and only if $\atomP{a}{b}{\dot{p}} \in \dot{G}_{\mathrm{med}}$.
  That is, for every concept name $A \in \ConceptNames$ and role name $p \in \RoleNames$, it holds that $A^{\IntMed} = \dot{A}^{\IntExt}$ and $p^{\IntMed} = \dot{p}^{\IntExt}$.
  Hence, for every concept description $C$, $C^{\IntMed} = \dot{C}^{\IntExt}$.

  We give below, for the sake of completeness, the remaining two cases, case 1. and case 3.

  \begin{itemize}
    \item Case 1. 
    By construction of $\graphext$, for every concept assertion $\atomC{a}{A} \in \graphin$, $\atomC{a}{A} \in \graphext$ if and only if $\atomC{a}{A} \in G_{\mathrm{in}}$, and for every role assertion $\atomP{a}{b}{p} \in \graphin$, $\atomP{a}{b}{p} \in \graphext$ if and only if $\atomP{a}{b}{p} \in G_{\mathrm{in}}$.
    That is, for every concept name $A \in \ConceptNames$ and role name $p \in \RoleNames$, it holds that $A^{\IntIn} = A^{\IntExt}$ and $p^{\IntIn} = p^{\IntExt}$.
    Hence, for every concept description $C$, $C^{\IntIn} = C^{\IntExt}$.

    \item Case 3.
    By construction of $\graphext$, for every concept assertion $\atomC{a}{A} \in \graphout$, $\atomC{a}{\ddot{A}} \in \graphext$ if and only if $\atomC{a}{\ddot{A}} \in \ddot{G}_{\mathrm{out}}$, and for every role assertion $\atomP{a}{b}{p} \in \graphout$, $\atomP{a}{b}{\ddot{p}} \in \graphext$ if and only if $\atomP{a}{b}{\ddot{p}} \in \ddot{G}_{\mathrm{out}}$.
    That is, for every concept name $A \in \ConceptNames$ and role name $p \in \RoleNames$, it holds that $A^{\IntOut} = \ddot{A}^{\IntExt}$ and $p^{\IntOut} = \ddot{p}^{\IntExt}$.
    Hence, for every concept description $C$, $C^{\IntOut} = \ddot{C}^{\IntExt}$.
  \end{itemize}
\end{proof}

\subsection[Proof for Corollary 1]{Proof for \Cref{cor:reduction-algo1}}

\begin{proof}
  The proof for \Cref{cor:reduction-algo1} follows immediately from case 3 of \Cref{prop:reduction-extended-graph}.
  Let $\sccq$ be a \sccqname, $\Sigma$ a set of $\DLogics$ axioms such that $\shaclvalid{\graphext}{\Sigma}$ for every extended graph $\graphext$ of $\sccq$, and $s$ a Simple SHACL shape such that $\Sigma \models \ddot{s}$.
  Then also $\shaclvalid{\graphext}{\{\ddot{s}\}}$, from which by case 3 of \Cref{prop:reduction-extended-graph} follows immediately that $\shaclvalid{\graphout}{s}$.
\end{proof}

\subsection{\texorpdfstring{Proof for \Cref{prop:una}}{Proof for Proposition 4}}
\label{proof:una}

\begin{proof}
  By construction, each axiom in $\UNA(q)$ is also an axiom in the validation knowledge base of the graph $\graphext$ (see \Cref{def:validation-rdf-semantics}).
\end{proof}

\subsection{\texorpdfstring{Proof for \Cref{prop:cwa}}{Proof for Proposition 5}}
\label{proof:cwa}

We first prove \Cref{lemma:vcg}.

\begin{proof}
  We prove \Cref{lemma:vcg} by contradiction.
  Let $\sccq = \sccqformal$ be a query such that $\vcg(\sccqpattern)$ is acyclic.
  Let $G$ be a graph, and let $x$ be a variable corring in $\sccqpattern$. 
  Let $C$ be a concept defined as
  \[
      C \equiv \textstyle 
          {\bigsqcap_{\atomC{x}{A} \in P}} A 
          \sqcap {\bigsqcap_{\atomP{x}{u}{p} \in P}} \exists p.\Concept{u}
          \sqcap {\bigsqcap_{\atomP{u}{x}{p} \in P}} \exists p^-.\Concept{u}.
  \]
  and assume that there is an individual name $c$ in $G$ such that $\atomC{c}{C}$
  is valid in $\graphext$, but $\atomC{c}{\vconcept{x}}$ is not valid in $\graphext$.

  Without loss of generality, assume that $\sccqpattern$ includes a single concept assertion including the variable $x$, namely $\atomC{x}{A}$,
  and no role assertions $\atomP{x}{d}{r}$ where $d$ is an individual name.
  Indeed, if there are serveral atoms $\atomC{x}{A_1}, \ldots, \atomC{x}{A_n}$ and $\atomP{x}{d_1}{r_1}, \ldots, \atomP{x}{d_m}{r_m}$ we can define 
  $A \equiv A_1 \sqcap \ldots \sqcap A_m \sqcap \exists r_1.\{d_1\} \sqcap \ldots \sqcap \exists r_m.\{d_m\}$.
  Without loss of generality, assume that $\graphext$ includes the graph 
  \[\{ 
    \atomC{a}{V_y}, \atomP{a}{c}{r},
    \atomC{c}{A}, \atomP{c}{b}{s},
    \atomC{b}{V_z}
  \}\]
  Let $\Omega$ be the set of all mappings $\mu$ such that $\mu(P) \subseteq G$.

  Then, by definition there exist the mappings $\mu_1, \mu_2 \in \Omega$ such that $\mu_1(y) = a$ and $\mu_2(z) = b$,
  but there not exists the mapping $\mu\in\Omega$ such that $\mu(x) = c$.
  Then, $\mu_1(x) \neq c$ and $\mu_2(x) \neq c$.

  Let $P_y$ be the part of pattern $\sccqpattern$ which \emph{connects} with variables $y$ and $x$, but not $z$.
  Let $P_z$ be the part of pattern $\sccqpattern$ which \emph{connects} with variables $z$ and $x$, but not $y$.
  Let 
  \begin{align*}
    \mu_1^y &= \mu_1\big\vert_{\var(P_y)\setminus\{x\}}, & \mu_1^z &= \mu_1\big\vert_{\var(P_z)\setminus\{x\}},\\
    \mu_2^y &= \mu_2\big\vert_{\var(P_y)\setminus\{x\}}, & \mu_2^y &= \mu_2\big\vert_{\var(P_z)\setminus\{x\}}.
  \end{align*}
  Then, $\mu_1 = \mu_1^y \cup \{x \mapsto \mu_1(x)\} \cup \mu_1^z$ and $\mu_2 = \mu_2^y \cup \{x \mapsto \mu_2(x)\} \cup \mu_2^z$.
  Since $\mu_1^y$ and $\mu_2^z$ share no variables, $\mu_3 = \mu_1^y \cup \{x \mapsto c\} \cup \mu_2^z$ is a mapping.

  By the definition of the semantics of \sccqname, $\mu_3 \in \Omega$.
  Then $\atomC{c}{\vconcept{x}}$ is valid in $\graphext$.
  This contradicts the initial assumptions, from which we conclude $\vconcept{x} \sqsupseteq C$.
\end{proof}

\noindent
We now continue with the proof for \Cref{prop:cwa}.

\begin{proof}
  We prove \Cref{prop:cwa} by showing the validity of this proposition for the cases 1 through 5 in Definition~\ref{def:cwa}. 
  We divide the proof in two groups: First, for cases 1, 2, and 3, and then for cases 4 and 5.

  \paragraph{Cases 1, 2, and 3.}
  For cases 1 through 3, we divide the proof in two parts each, one for either inclusion (i.e., $\sqsubseteq$ and $\sqsupseteq$).
  To show an inclusion $A \sqsubseteq B$ in $\graphext$, we will assume that there exists at least one valuation $\mu$ such that $\mu(P) \subseteq \graphin$, and then prove that inclusion $\{a\} \sqsubseteq A$ implies inclusion $\{a\} \sqsubseteq B$ for every individual name occurring in $\graphext$.

  \begin{enumerate}

    \item[$1_{\sqsubseteq}$] $\dot{A} \sqsubseteq A \sqcap \bigsqcup_{\atomC{u}{A}\in P} \Concept{u}$.
    Let $a$ be an arbitrary individual name such that $\{a\} \sqsubseteq \dot{A}$ is valid in $\graphext$.
    Then $\atomC{a}{\dot{A}} \in \graphext$, so 
    $\atomC{a}{A} \in \graphmed$. 
    Then $\atomC{a}{A} \in \graphin$ and there is an atom $\atomC{v}{A} \in P$ where $v$ is either individual name $a$ or a variable $x$.
    If $v$ is $a$, then $\{a\} \sqsubseteq \Concept{v}$ is trivially valid in $\graphext$. 
    Otherwise $v$ is $x$ and by construction, $\atomC{a}{\vconcept{x}} \in \graphext$.
    Then, $\{a\} \sqsubseteq \Concept{v}$ is valid in $\graphext$, so $\{a\} \sqsubseteq \bigsqcup_{\atomC{u}{A}\in P} \Concept{u}$ is also valid in $\graphext$.
    Similarly, since $\graphin \subseteq \graphext$, $\atomC{a}{A} \in \graphext$, so $\{a\} \sqsubseteq A$ is valid in $\graphext$. 
    Therefore, $\{a\} \sqsubseteq A \sqcap \bigsqcup_{\atomC{u}{A}\in P} \Concept{u}$ is valid in $\graphext$.

    \item[$1_{\sqsupseteq}$] $\dot{A} \sqsupseteq A \sqcap \bigsqcup_{\atomC{u}{A}\in P} \Concept{u}$.
    Let $a$ be an arbitrary individual name such that $\{a\} \sqsubseteq A \sqcap \bigsqcup_{\atomC{u}{A} \in P} \Concept{u}$ is valid in $\graphext$. 
    By construction, for every individual name $b$, if $\atomC{b}{A} \in \graphext$ then $\atomC{b}{A} \in \graphin$. 
    Since $\atomC{a}{A} \in \graphext$, $\atomC{a}{A} \in \graphin$ holds. 
    By definition, $\{a\} \sqsubseteq \Concept{u}$ is valid in $\graphext$ for at least one atom $\atomC{u}{A} \in P$. 
    If $u$ is an individual name, then $u$ is $a$, and $\atomC{a}{A} \in P$. 
    If $\graphmed$ is not empty, then $\atomC{a}{A} \in \graphmed$. 
    Otherwise, if $u$ is a variable $x$ then $\Concept{u}$ is the variable concept $\vconcept{x}$, and $\atomC{a}{\vconcept{x}} \in \graphext$. 
    By definition, $a$ is an instance of variable $x$, and thus $\atomC{a}{A} \in \graphmed$. 
    Hence, $\atomC{a}{A} \in \graphmed$ in all possible cases (when $u$ is an individual name or when $u$ is a variable). 
    By construction, $\atomC{a}{\dot{A}} \in \graphext$ and therefore $\{a\} \sqsubseteq \dot{A}$.

    \item[$2_{\sqsubseteq}$] $\ddot{A} \sqsubseteq \bigsqcup_{\atomC{u}{A}\in H} \Concept{u}$.
    Let $a$ be an arbitrary individual name such that $\{a\} \sqsubseteq \ddot{A}$ is valid in $\graphext$. 
    By construction, $\atomC{a}{A} \in \graphout$ so there is an atom $\atomC{v}{A} \in H$ such that $\atomC{a}{A}$ is an instance of pattern $\atomC{v}{A}$, and $v$ is either concept name $a$ or a variable $x$. 
    If $v$ is $a$ then $\{a\} \sqsubseteq \Concept{v}$ is trivially valid in $\graphext$. 
    Otherwise $v$ is $x$ and since $a$ is an instance of $x$, it holds that $\atomC{a}{\vconcept{x}} \in \graphext$, so $\{a\} \in \Concept{v}$ is valid in $\graphext$. 
    Therefore, $\{a\} \sqsubseteq \bigsqcup_{\atomC{u}{A}\in H} \Concept{u}$ is valid in $\graphext$.

    \item[$2_{\sqsupseteq}$] $\ddot{A} \sqsupseteq \bigsqcup_{\atomC{u}{A}\in H} \Concept{u}$.
    Let $a$ be an arbitrary individual name such that $\{a\} \sqsubseteq \bigsqcup_{\atomC{u}{A} \in H} \Concept{u}$ is valid in $\graphext$.
    Then there is at least one atom $\atomC{v}{A} \in H$ such that $\{a\} \sqsubseteq \Concept{v}$ is valid in $\graphext$.
    If $v$ is $a$, then $\atomC{a}{A} \in \graphout$, and thus $\atomC{a}{\ddot{A}} \in \graphext$.
    Otherwise $v$ is a variable $x$, and $\atomC{a}{\vconcept{x}} \in \graphext$.
    By the definition of variable concepts, $\atomC{a}{A} \in \graphout$, so $\atomC{a}{\ddot{A}} \in \graphext$.
    Therefore, $\{a\} \sqsubseteq \ddot{A}$ is valid in $\graphout$.

    \item[$3_{\sqsubseteq}$] \emph{Variable ($\sqsubseteq$).} 
      Let $a$ be an arbitrary individual name such that $\{a\} \sqsubseteq \vconcept{x}$ is valid in $\graphext$.
      We show separately for each operand $k$ in the intersection, that $\{a\} \sqsubseteq k$, assuming that the respective component is defined, below.
      \begin{enumerate}
        \item For $k = {\bigsqcap_{\atomC{x}{A} \in \sccqpattern}} A$:  
          If $\{a\} \sqsubseteq \vconcept{x}$, then by definition $a$ is an instance of variable $x$ in $\sccqpattern$, i.e., $a \in \mu(x)$.
          Then for each concept name $A$ occurring in an atomic pattern of the form $\atomC{x}{A} \in \sccqpattern$ there must be $\atomC{a}{A} \in \graphin$ 
          (since otherwise $a \not\in \mu(x)$), so also $\atomC{a}{A} \in \graphext$ for each such $A$.
          Therefore, $\{a\} \sqsubseteq k$.
        \item For $k = {\bigsqcap_{\atomP{x}{u}{p} \in\sccqpattern}} \exists p.\Concept{u}$:
          If $\{a\} \sqsubseteq \vconcept{x}$, then by definition $a$ is an instance of variable $x$ in $\sccqpattern$, i.e., $a \in \mu(x)$.
          Then for each property name $p$ occurring in an atomic pattern of the form $\atomP{x}{u}{p} \in \sccqpattern$, one of two cases applies:
          If $u$ is an individual name, then there must be $\atomP{a}{u}{p} \in \graphin$, so also $\atomP{a}{u}{p} \in \graphext$ for such $p$.
          If $u$ is a variable name, then there must be $\atomP{a}{b}{p} \in \graphin$, so also $\atomP{a}{b}{p} \in \graphext$, 
          and also $b \in \mu(u)$ (since otherwise $a \not\in \mu(x)$).
          Therefore, $\{a\} \sqsubseteq k$.
        \item For $k = {\bigsqcap_{\atomP{u}{x}{p} \in\sccqpattern}} \exists p^-.\Concept{u}$: Analogous to the previous case.
      \end{enumerate}
      If at least one component $k$ is defined, then it follows that 
      \[ \{a\} \sqsubseteq {\bigsqcap_{x:A \in\sccqpattern}} A \sqcap {\bigsqcap_{(x,u):p \in\sccqpattern}} \exists p.\Concept{u} \sqcap {\bigsqcap_{(u,x):p \in\sccqpattern}} \exists p^-.\Concept{u}.\]
      We know, that at least one component $k$ must be defined, since otherwise concept $\vconcept{x}$ would not be defined,
      as there must exists either $\atomC{x}{A} \in\sccqpattern$ for some concept name $A$, 
      or $\atomP{x}{u}{p} \in\sccqpattern$ (or $\atomP{x}{u}{p} \in\sccqpattern$ respectively) for some property $p$, if $x \in \var(P)$.
      Then, at least one of the components $k$ must be defined as well, and we prove this case.

    \item[$3_{\sqsupseteq}$] \emph{Variable ($\sqsupseteq$).} 
      The inverse case follows directly from \Cref{lemma:vcg}.

  \end{enumerate}

  This concludes the proof of cases 1., 2., and 3.
  We next consider cases 4. and 5.

  \paragraph{Cases 4 and 5.}
  Since the proofs of these two cases are similar, we exemplify them proving the equivalency:
  \[
    \exists \dot{p}.\Concept{u} \equiv \textstyle \bigsqcup_{\atomP{v}{u}{p} \in P}\Concept{v}
  \]
  Let $\Int$ be the canonical model of $\graphext$.
  By definition of the validation knowledge base of a graph, $a^\Int \in (\exists \dot{p}.\Concept{u})^\Int$ if and only if there exists an individual name $b$ such that $\atomP{a}{b}{\dot{p}} \in \graphext$ and $(a^\Int,b^\Int) \in p^I$. 
  By construction, $\atomP{a}{b}{\dot{p}} \in \graphext$ if and only if there exists an atom $\atomP{v}{u}{p} \in P$ where $v$ is the individual name $a$ or a variable $x$, and $u$ is the individual name $b$ or a variable $y$ (and thus $a^I \in \Concept{v}^\Int$). 
  Thus, $a^\Int \in (\exists \dot{p}.\Concept{u})^\Int$ if and only if $a^\Int \in \bigcup_{\atomP{v}{u}{p} \in P} \Concept{v}^\Int$. 
  Hence, $\exists \dot{p}.\Concept{u} \equiv \textstyle \bigsqcup_{\atomP{v}{u}{p} \in P}\Concept{v}$.

  Similarly, we exemplify the proof of the remaining axioms of the following form, using one of these axioms:
  \[
    \exists \dot{p}.\top \equiv {\bigsqcup_{\atomP{u}{v}{p} \in P}} \Concept{u} \sqcap \exists \dot{p}.\Concept{v}
  \]
  By definition of the validation knowledge base of a graph, $a^\Int \in (\exists \dot{p} . \top)^\Int$ if and only if there exists an individual name $b$ such that $\atomP{a}{b}{\dot{p}} \in \graphext$ and $(a^\Int, b^\Int) \in p^\Int$.
  By construction, $\atomP{a}{b}{\dot{p}} \in \graphext$ if and only if there exists an atom $\atomP{u}{v}{p} \in P$ where $u$ is the individual name $a$ or a variable $x$ and $v$ is the individual name $b$ or a variable $y$.
  Then, $a^\Int \in \Concept{u}^\Int$ and $a^\Int \in (\exists \dot{p}.\Concept{v})^\Int$, so $a^\Int \in (\exists \dot{p}.\top)^\Int$ if and only if $a^\Int \in {\bigcup_{\atomP{u}{v}{p} \in P}} \Concept{u}^\Int \cap (\exists \dot{p}.\Concept{v})^\Int$.

  \paragraph{Conclusion}

  Finally, given the proofs for the individual cases listed above, we prove this proposition.
\end{proof}

\subsection{\texorpdfstring{Proof for \Cref{prop:map}}{Proof for Proposition 6}}
\label{proof:map}

We begin with the following utility definition.

\begin{definition}
  \label{def:ins}
  For every individual name or variable $u$ we define the set of individual names $\ins(u)$ as follows:
  \begin{align*}
    \ins(u) &=
      \left\{
      \begin{array}{ll}
        \{a\} & \text{ if $u$ is ind. name $a$},\\
        \{a \mid a \text{ instance of }x\} & \text{ if $u$ is a variable $x$}.
      \end{array}
      \right.
  \end{align*}
\end{definition}

In order to prove Proposition~\ref{prop:map}, we now define the following lemma, stating a relation between $\ins(u)$ (Definition~\ref{def:ins}) and $\Concept{u}$ (Definition~\ref{def:ct}).

\begin{lemma}\label{lemma:ins-concept}
  Let $\sccqformal$ be a \sccqname, $\graphext$ an extended graph for the query, and $u$ and $v$ two individual names or variables occurring in the query.
  Then, $\ins(v) \subseteq \ins(u)$ if and only if the inclusion $\Concept{v} \sqsubseteq \Concept{u}$ is valid for graph $\graphext$.
\end{lemma}

\begin{proof}
  Let $I$ be an interpretation of the validation knowledge base of graph $\graphext$.
  By definition, for every individual $b \in \top^I$ there exists a unique individual name $a$ in the graph such that $a^I=b$.
  Let $u$ and $v$ be two individual or variable names occurring in pattern $P$. 
  Then, $\cdot^I$ defines a bijection between sets $\ins(u)$ and $\Concept{u}^I$, and a bijection between sets $\ins(v)$ and $\Concept{v}^I$. 
  Thus, $\ins(v) \subseteq \ins(u)$ if and only if $\Concept{v}^I \subseteq \Concept{u}^I$. 
  Hence, $\ins(v) \subseteq \ins(u)$ if and only if $\Concept{v} \sqsubseteq \Concept{u}$.
\end{proof}

\begin{proof}[Proof for \Cref{prop:map}]
  Let $P_1$ and $P_2$ be components of graph pattern $P$, the function $h : \var(P_1) \to \var(P_2)$ be a component map, and $x$ and $y$ be two variables in $P_1$ and $P_2$, respectively, such that $h(x) = y$.
  According to Lemma~\ref{lemma:ins-concept}, to prove that $\vconcept{y} \sqsubseteq \vconcept{x}$ is valid in graph $\graphext$ it suffices to prove $a \in \ins(y)$ implies $a \in \ins(x)$ for every individual name $a$.

  Let $a$ be an individual name in $\ins(y)$. 
  Then, there exists a valuation $\mu$ such that $\mu(P) \subseteq \graphin$. 
  Let $h': \var(P) \to \var(P)$ be the function that extends $h$ for the variables in $P$ that are not in the domain of $h$ as follows:
  \[
    h'(z) = \left\{
      \begin{array}{ll}
        z & \text{if } z \notin \dom(h),\\
        h(z) & \text{if } z \in \dom(h).
      \end{array}
    \right.
  \]
  By construction, $h'(P) \subseteq P$. 
  Applying $\mu$ to both sides of the inclusion we get $\mu(h'(P)) \subseteq \mu(P)$. 
  By transitivity, $\mu(h'(P)) \subseteq \graphin$. 
  That is, $\mu'(P) \subseteq \graphin$ where $\mu'$ is the composite valuation $h'\mu$. 
  Since $\mu'(x) = a$, we conclude that $a \in \ins(x)$.
\end{proof}

\subsection{\texorpdfstring{Proof for \Cref{prop:mapsin}}{Proof for Proposition 7}}
\label{proof:mapsin}

\begin{proof}
    To prove Proposition~\ref{prop:mapsin}, it suffices to show that the extension approach is sound, i.e., that both the extended and non-extended components are equivalent with respect to the bindings for all actual query variables, since then the proof for Proposition~\ref{prop:map} applies.

    Consider the variable $x$ as a target of shape $s = \targetquery \sqsubseteq \phi $.
    Then, the following extensions $\operatorname{extp}(x, \phi)$ are permitted, depending on $\phi$:

    \begin{enumerate}
        \item $\phi = A$ and $\{\atomC{x}{A}\}$. 
        For any input graph $\graphin$ it holds that $\shaclvalid{\graphin}{\{\targetquery \sqsubseteq A\}}$.
        Then $\forall a \in \mu(x)\ :\ \atomC{a}{A} \in \graphin$, since $x$ is a target of $s$.
        Therefore, pattern $\atomC{x}{A}$ is satisfied for all $\graphin$.
        \item $\phi = \exists p.A$ and $\{\atomP{x}{x_0}{p}, \atomC{x_0}{A}\}$. 
        For every input graph $\graphin$, $\shaclvalid{\graphin}{\{\targetquery \sqsubseteq \exists p.A\}}$.
        Then for all $a \in \mu(x)$, $\atomP{a}{b}{p},\atomC{b}{A} \in \graphin$, since $x$ is a target of $s$.
        Therefore, patterns $\atomP{x}{x_0}{p}$ and $\atomC{x_0}{A}$ are satisfied for all $\graphin$.
        \item $\phi = \exists p^-.A$ and $\{\atomP{x_0}{x}{p}, \atomC{x_0}{A}\}$. This case is similar to the previous case.
        \item $\phi = \forall p.A$ and $\{\atomC{y}{A} \mid \atomP{x}{y}{p} \in P_{\text{ext}} \}$. 
        For every input graph $\graphin$, $\shaclvalid{\graphin}{\{\targetquery \sqsubseteq \forall p.A\}}$.
        Then, for all $a \in \mu(x)$, $\atomP{a}{b}{p} \in \graphin$ implies $\atomC{b}{A} \in \graphin$, since $x$ is a target of $s$.
        Therefore, for any pattern $\atomP{x}{y}{p} \in P_{\text{ext}}$, $\atomC{y}{A}$ is satisfied for all $\graphin$.
        \item $\phi = \forall p^-.A$ and $\{\atomC{y}{A} \mid \atomP{y}{x}{p} \in P_{\text{ext}} \}$.  This case is similar to the previous case.
    \end{enumerate}
\end{proof}

\subsection{\texorpdfstring{Proof for \Cref{prop:properties}}{Proof for Proposition 8}}
\label{proof:properties}

We separately prove the two components (1) and (2) of \Cref{def:sppa} involved in \Cref{prop:properties}.
To this end, we write $\shaclvalid{\graphext}{\PropertyAxioms_1(\sccq)}$ and $\shaclvalid{\graphext}{\PropertyAxioms_2(\sccq)}$, where $\sccq = \sccqformal$ is a \sccqname.
  
\begin{proof}[Proof for $\shaclvalid{\graphext}{\PropertyAxioms_1(\sccq)}$.]
  For an arbitrary property name $p \in \voc(\sccqpattern)$, the axiom $\dot{p} \sqsubseteq p$ is always true since $\graphmed \subseteq \graphin$, by definition.
  When the pattern only contains $\atomP{x}{y}{p}$ such that $x$ and $y$ do not occur in any other atomic patterns in $\sccqpattern$ (i.e., $x$ and $y$ are otherwise unrestricted), then for any $\atomP{a}{b}{p} \in \graphin$, $\atomP{a}{b}{\dot{p}} \in \graphmed$.
  Therefore $p \sqsubseteq \dot{p}$.
\end{proof}

\begin{proof}[Proof for $\shaclvalid{\graphext}{\PropertyAxioms_2(\sccq)}$]
  Let $p \in \voc(\sccqpattern)$ and $r \in \voc(\sccqtemplate)$, such that $\sccqpattern$ contains the atomic pattern $\atomP{x}{y}{p}$ and $\sccqtemplate$ contains $\atomP{x}{y}{r}$, and neither $\sccqtemplate$ nor $\sccqpattern$ contains any other atomic patterns involving $x$ or $y$, and $p$ or $r$, respectively.
  Then, for any $\atomP{a}{b}{\dot{p}} \in \graphmed$ we construct $\atomP{a}{b}{\ddot{r}} \in \graphout$, therefore, $\dot{p} \sqsubseteq \ddot{r}$.
  In addition, since $r$ does not occur again in $\sccqtemplate$, $\dot{p} \equiv \ddot{r}$, i.e. also $\ddot{r} \sqsubseteq \dot{p}$.
\end{proof}

\subsection{Theorem 1 and Proof}
\label{proof:t1}

\begin{theorem}
  \label{theorem:np}
  Problem \problemSE is NP-hard.
\end{theorem}

\begin{proof}
  We next show that the simple graph entailment problem described by Gutierrez et al.~\cite{DBLP:journals/jcss/GutierrezHMP11} (called SGE in what follows) 
  can be reduced to problem \problemSE.
  Problem SGE is equivalent to deciding if for a pattern $P$ consisting of two components $P_1$ and $P_2$ there is a component map $h$ from $P_1$ to $P_2$.
  Let $\shapesin$ be an empty set, and $q$ be the \sccqname $H \gets P$ where $H$ contains an atom $\atomC{u}{A_u}$ for each variable or individual name in $P$.
  By \Cref{prop:map}, there exists such a mapping $h$ if and only if $\problemSE(\shapesin, q, A_u \sqsubseteq A_x)=\textsc{yes}$,
  for each pair $(x,u)$ where $h(x) = u$ and $x$ is a variable in $P_1$ and $u$ is a variable in $P_2$.
  Since the number of pairs $(x,u)$ is quadratic on the size of $P$, we have shown a reduction from problem SGE to problem \problemSE.
  Since SGE is NP-hard, problem \problemSE is also NP-hard.
\end{proof}

\subsection{Proposition 9 and Proof}
\label{proof:mn}

\begin{proposition}\label{prop:nm}
  If $\voc(q)$ contains $n$ concept names, and $m$ role names, then we need to iterate over $n + 2m$ target queries, and $n + 4nm + 2m$ shape constraints, and return $(n + 2m)(n + 4nm + 2m) - n$ many relevant shapes.
\end{proposition}

\begin{proof}
  We have $n$ possible target queries with a concept name ($A$ for each $A\in\voc(q)$), and $2m$ with a role name ($\exists p.\top$ and $\exists p^-.\top$ for each $p\in\voc(q)$).
  Similarly, we have $n$ possible shape constraints including only a concept name, and $4nm$ possible shape constraints including a concept name and a role name ($\exists p.A$, $\exists p^-.A$, $\forall p.A$, and $\forall p^-.A$ for each $A\in\voc(q)$ and $p\in\voc(q)$).
  We also have $2m$ representatives for families of the form $\forall p.B$ and $\forall p^-.B$ for each $p\in\voc(q)$ and for some proxy concept name $B\not\in\voc(q)$.
  The subtrahend in the number of relevant shapes indicates the number of tautologies of the form $A \sqsubseteq A$ for all $A\in\voc(q)$.
\end{proof}

  \section{Extending the Method}
\label{a:extension}

In this section, we show how our approach can be extended to arbitrary $\DLogics$ axioms, and thus a much more extensive subset of SHACL.
To this end, we first define $\DLogics$ SHACL shapes as follows.

\begin{definition}[$\DLogics$ SHACL Syntax]%
  \label{def:glogics-shacl-syntax}
  A $\DLogics$ \emph{SHACL shape} is an $\DLogics$ axiom $\targetquery \sqsubseteq \phi$ such that the concept expressions $\phi$ is an arbitrary $\DLogics$ concept description, and $\targetquery$ is defined by:
  \begin{equation*}
    \begin{aligned}[t]
    \targetquery & \Coloneqq A \mid \exists p.\top \mid \exists p^-.\top 
    \end{aligned}
  \end{equation*}

  A $\DLogics$ \emph{SHACL schema} $\allshapes$ is an $\DLogics$ TBox that consists of a finite set of $\DLogics$ SHACL shapes.
\end{definition}

\begin{definition}[$\DLogics$ SHACL Semantics]
  A graph $G$ is \emph{valid} for a set $S$ of $\DLogics$ SHACL shapes, denoted $\shaclvalid{G}{S}$, if and only if $G$ is proof-valid according to $S$.
\end{definition}

We omit explicitly redefining the remainder of the main paper in terms of $\DLogics$ SHACL shapes, for the sake of simplicity, since definitions do not substantially change (as we argue in the next subsection).
Instead, we instruct the reader to consider $\shapesin$, $\shapescandidates$ and $\shapesout$ (and other sets of shapes) as a set of $\DLogics$ SHACL shapes for the remainder of this section, and with respect to prior definitions.

We first consider soundness.
In the remainder of this section, we present further notes on extending the method, first considering and extended axiom inferences then additional features beyond $\DLogics$ axioms.
Finally, we remark how to extend the implementation.

\subsection{Soundness}

In this subsection, we argue for the soundness of our presented approach for more general $\DLogics$ SHACL shapes.
Indeed, we show that soundness of the method introduced in the main body of the paper is not affected by more general $\DLogics$ input shapes or shape candidates.
We revisit each proposition from the main body of the paper and consider, whether the proposition or its proof need to be adapted.

\begin{enumerate}
    \item \Cref{prop:equivalent-formalisms} is independent of the subset of SHACL, so the same proof applies.
    \item \Cref{prop:relevant-shapes} must be extended for the extended set of SHACL shapes to demonstrate that the method is useful for $\DLogics$ SHACL shapes (i.e., there is a meaningful finite set of candidates),
      though this does not effect soundness, as a set of candidates could also be determined by some heuristic. See \Cref{a:proofsext} for the extended proof. 
      Indeed, not generating the full set of shapes is likely desireable in most use cases, anyways.
    \item \Cref{prop:reduction-extended-graph} (and \Cref{cor:reduction-algo1}) were already proven for arbitrary $\DLogics$ axioms, and thus apply in the context of $\DLogics$ SHACL shapes as well.
    \item For \Cref{prop:una}, neither its definition, nor the definition of the UNA for a simple RDF graph, depend on the subset of SHACL.
    \item Similarly, for \Cref{prop:cwa} the proof is independent of the subset of SHACL and still applies as well.
    \item The proof for \Cref{prop:map} does not depend on the subset of SHACL shapes as well, and thus the proposition holds.
    \item \Cref{prop:mapsin} does involve the set of input shapes. Here, we need to decide whether we extend the approach. As per the argument in the following subsection, we consider this extension to be future work, and limit expansion to the subset of Simple SHACL shapes. Then, the proof applies and the proposition holds. Note, that this is only a minor restriction, since extending the query with respect to $\DLogics$ SHACL shapes would be limited to shapes expressible as \sccqname anyways, which essentially means that we would need to include intersection of constraints in shapes as two sets of extension patterns, which is a trivial extension to the method we present.
    \item Finally, for \Cref{prop:properties}, while this step of the approach does depend on the input shapes, it only considers the role names in the vocabulary of the input shapes. Neither the method itself, nor the proof of this proposition, depend on the types of constraints expressed as input shapes, but rather are about the query patterns used to restrict these specific role names. Thus, the proposition (and its proof) still apply without modification.
\end{enumerate}

Thus, the method remains sound for $\DLogics$ SHACL shapes.
We only need to show that there is a sensible finite set of candidates for $\DLogics$ SHACL as well (see \Cref{a:proofsext} for the proof), 
in order to demonstrate that the algorithm can obtain a useful and finite set of result shapes.

For the sake of completeness, we mention that \Cref{prop:nm} (count of candidate shapes) does no longer apply for the extended method.
However, \Cref{prop:nm} is not required for soundness; indeed, a similar proposition could be formulated for the count of candidates for the extended method.

Finally, \Cref{theorem:np} (NP-hardness) clearly holds for the extended method, since the original problem can be trivially reduced to the extended problem by restricting the set of shapes to Simple SHACL.

\subsection{Extended Axiom Inference}

In the previous section we show that our approach is sound for $\DLogics$ SHACL.
This extension also makes the approach more powerful in multiple ways.
The set of input shapes can be extended to $\DLogics$ SHACL, thus, more expressive constraints are considered to hold on the input graph.
Similarly, the set of candidates (and thus output shapes) also includes more general shapes.

The remaining question to consider is whether or not more axioms can now be inferred in order to further improve the method:
The axioms inferred for the UNA-encoding and CWA-encoding are independent of the subset of SHACL, but rather depend on the query language (and indeed graph model) used.
Similarly, the axioms inferred as subsumptions between query variables (\emph{mappings}) again depend on the query language, not on the subset of SHACL; however, the extension approach utilizing input shapes does depend on the set of input shapes.
Here, additional rules could be added for extending the approach;
we leave this as future work, as this is not a substantial addition to the method.

\subsection{\texorpdfstring{SHACL Features Beyond $\DLogics$}{SHACL Features Beyond ALCHOI}}

We consider in passing our intuition on whether our approach can be extended for additional SHACL features beyond $\DLogics$.

\begin{itemize}
  \item Node target queries. Node target queries where omitted for the sake of simplicity, since inferring such shapes based on the query template would not be very productive. We believe this is a trivial extension to our method, if a use case were to require such shapes.
  \item Qualified number restrictions. Using an underlying DL with support for qualified number restrictions, we believe that there would not be an issue with supporting them. However, we omit them, since we think that there are only exceedingly rare circumstances, where meaningful qualified number restriction (other than existential and universal quantification) could indeed be inferred for SPARQL CONSTRUCT queries.
  The particular restrictions to consider likely form a finite set informed by the query template.
  \item Non-cyclic shape references. Non-cyclic shape references are syntactic sugar and can be resolved by substitution. Thus, our method essentially supports non-cyclic shape references already.
  \item Cyclic (i.e., recursive) shape references are not supported. For recursive SHACL shapes, sets of results shapes would no longer be independent, and thus, our filtering method not applicable. 
    However, we think that only \emph{validating} a given set of shapes that include recursive shape references over the axioms inferred by our method should be possible.
  \item SHACL features validating literal values. We omit literals for the sake of simplicity; from \sccqname queries, no interesting constraints on literal values could be inferred.
  Literal value constraints that occur in the input shapes could perhaps be maintained through an encoding via some utility concept definition.
\end{itemize}

\subsection{Practicality and Implementation}

In order to efficiently explore the finite (see also the next section), but much larger search space of candidate $\DLogics$ SHACL shapes, we suggest the following approaches.
Importantly, one can notice that full exploration of the candidate space is rarely required.
Indeed, a subset of result shapes entailing all other shapes is most sensible for the majority of use cases, no matter whether the use case is informing users (including redundant shapes would not be necessary, but rather confusing), suggesting shapes for some data set (e.g., data integration use case), or for validation in a programming language context, or any other automatic processing of result shapes, where a minimal set entailing a larger set would generally suffice. 
Indeed, in such automatic cases one may not need to instantiate any shapes, but instead rely only on the set of axioms, which already can be used to check for entailment of individual shapes to the extend required by such a use case.

In order to reduce the set of candidates, one can reduce the syntax by relying on the set semantics of the set of result shapes.
For example, union and intersection of constraints is not required on the top-level of a constraint, since both can be reconstructed from entailment in the result set (e.g., if both $\psi_1 \sqsubseteq \phi_1$ and $\psi_1 \sqsubseteq \phi_2$ are result shapes, trivially, $\psi_1 \sqsubseteq \phi_1 \sqcap \phi_2$ is as well. 
This holds similarly for union and some types of quantification.

Secondly, one can systematically cover the search space with a breath-first search strategy, where immediate candidates are validated, before constructing more complex shapes.
For example, if $\psi_1 \sqsubseteq \phi_1$ is not a result shape, then for target $\psi_1$ we do not need to validate any intersection involving $\phi_1$.

\paragraph{Our Implementation} Our implementation allows for arbitrary $\DLogics$ axioms as input.
However, when obtaining result shapes, only Simple SHACL shapes are enumerated.
Thus, in order to construct a larger set of result shapes, the API must be used to obtain the inferred axioms, which can then be used to search an arbitrary set of shapes candidates, 
constructed by some heuristic suitable to the use case.

  \balance
  \section{Proofs for the Extension}
\label{a:proofsext}

We show here, that a finite set of candidate shapes can be constructed for the extended method.
To this end, we revise \Cref{prop:relevant-shapes} as \Cref{prop:relevant-shapes-ext} for $\DLogics$ shapes.

\begin{proposition}\label{prop:relevant-shapes-ext}
  If a $\DLogics$ SHACL shape $s = \psi \sqsubseteq \phi$ is relevant for a \sccqname $q$, then $\voc(s) \subseteq \voc(q)$.
\end{proposition}

We only consider constraints, i.e., the right hand side $\phi$ of a $\DLogics$ SHACL shape $\targetquery \sqsubseteq \phi$, since for $\targetquery$, we already show that the set of target queries is finite, given a finite vocabulary of a query $\sccq$ (Proof of \Cref{prop:relevant-shapes}, \Cref{proof:relevant-shapes}).

Without loss of generality, we assume that all constraints are in disjunctive normal form, without (syntactical) duplications and with components sorted according to some total order (e.g., by syntactic construct and then alphanumerically by role, concept or individual names). 
Thus, patterns such as $A \sqcap A$ do not occur, and $B \sqcap A$ is considered equal to $A \sqcap B$.
Furthermore, we omit $\forall p.C$, since it is equivalent $\neg\exists p.\neg C$.

We define the following lemmas.
The first one (\Cref{prop:empty-larger}) intuitively means, that for each concept description defined according to the grammar presented in the lemma, if the vocabulary of this description is not a subset of the vocabulary of some graph $G$, then the result is either equivalent to $\top$ or $\bot$, or the concept description can be simplified, such that the resulting concept description is in the vocabulary of $G$, or equivalent to $\top$ or $\bot$.

\begin{lemma}\label{prop:empty-larger}
  Let $G$ be a Simple RDF graph, $(\TBox_G, G)$ the validation knowledge base of $G$, and $C_1$ a concept description defined by the following grammar
  \begin{align} 
    C_1 & ::= C_2 \sqcup C_1 \mid C_2\\
    C_2 & ::= C_3 \sqcap C_2 \mid C_3\\
    C_3 & ::= \neg C_4 \mid C_4 \\
    C_4 & ::= \top \mid \bot \mid A \mid \{a\}
  \end{align}
  where $A$ is a concept name and $a$ an individual name.
  Then, $\voc(C_1)\not\subseteq \voc(G)$ implies one of the following cases:
  \begin{enumerate}
    \item $(\TBox_G, G) \models C_1 \equiv \top$ or $(\TBox_G, G) \models C_1 \equiv \bot$
    \item There exists a concept description $(\TBox_G, G) \models C_1' \equiv C_1$, such that either $\voc(C_1')\subseteq\voc(G)$, or $\voc(C_1')\not\subseteq \voc(G)$ and $(\TBox_G, G) \models C_1 \equiv \top$ or $(\TBox_G, G) \models C_1 \equiv \bot$.
  \end{enumerate}
\end{lemma}

\begin{proof}
  Let $\Int$ be a model of $(\TBox_G, G)$.
  Note, that according to \Cref{lemma:isomorphic-models}, every model $\Int$ of $(\TBox_G, G)$ is isomorphic to the canonical model of $G$.

  We first consider the two trivial cases for $C_4$.

  \begin{enumerate}
    \item If $C_4$ is $\top$, then trivially $(\TBox_G, G) \models \top \equiv \top$.
    \item If $C_4$ is $\bot$, then trivially $(\TBox_G, G) \models \bot \equiv \bot$.
  \end{enumerate}

  We next consider the two remaining cases for $C_4$.

  \begin{enumerate}
    \item If $C_4$ is $A$ and $A \not\in \voc(G)$, then $A^\Int$ is empty.
      Thus, $(\TBox_G, G) \models A \equiv \bot$.
    \item If $C_4$ is $a \not\in \voc(G)$, then $\{a\}^\Int$ is empty. 
      Thus, $(\TBox_G, G) \models \{a\} \equiv \bot$.
    \item If $C_4$ is $\exists p.C_1$, we have the following cases:
  \end{enumerate}

  We next consider the cases for $C_3$.

  \begin{enumerate}
    \item If $C_3$ is $\neg C_4$ and $\voc(C_4)\not\subseteq \voc(G)$, then either $(\TBox_G, G) \models C_4 \equiv \bot$ (and thus $(\TBox_G, G) \models \neg C_4 \equiv \top$), or $(\TBox_G, G) \models C_4 \equiv \top$ (and thus $(\TBox_G, G) \models \neg C_4 \equiv \bot$).
    Thus, $(\TBox_G, G) \models C_3 \equiv \top$ or $(\TBox_G, G) \models C_3 \equiv \bot$, if $\voc(C_3)\not\subseteq \voc(G)$.

    \item If $C_3$ is $C_4$ and $\voc(C_4)\not\subseteq \voc(G)$, then, as previously shown, either $(\TBox_G, G) \models C_4 \equiv \top$ or $(\TBox_G, G) \models C_4 \equiv \bot$ and thus $(\TBox_G, G) \models C_3 \equiv \top$ or $(\TBox_G, G) \models C_3 \equiv \bot$, if $\voc(C_3)\not\subseteq \voc(G)$.

    \item If $C_3$ is $\neg B$ and $\voc(C_3)\not\subseteq \voc(G)$, then 
    $(\TBox_G, G) \models \neg B \equiv \top$, since, by definition, there exists no $\atomC{b}{B} \in G$.
  \end{enumerate}

  We next consider the cases for $C_2$ by induction.

  \begin{enumerate}
    \item If $C_2$ is $C_3$ and $\voc(C_3)\not\subseteq \voc(G)$, then, as previously shown, either $(\TBox_G, G) \models C_3 \equiv \top$ or $(\TBox_G, G) \models C_3 \equiv \bot$ and thus $(\TBox_G, G) \models C_2 \equiv \top$ or $(\TBox_G, G) \models C_2 \equiv \bot$, if $\voc(C_2)\not\subseteq \voc(G)$.

    \item If $C_2$ is $C_3 \sqcap C_2'$ and $\voc(C_3)\not\subseteq \voc(G)$, then, as previously shown, either $(\TBox_G, G) \models C_3 \equiv \top$ or $(\TBox_G, G) \models C_3 \equiv \bot$.

    In the first case, we can reduce the term to $C_2'$ (since $\top \sqcap C \equiv C$), and by induction, either $\voc(C_2')\not\subseteq \voc(G)$ and then $(\TBox_G, G) \models C_2' \equiv \top$ or $(\TBox_G, G) \models C_2' \equiv \bot$, or $\voc(C_2')\subseteq\voc(G)$.

    In the second case, then also $(\TBox_G, G) \models C_2 \equiv \bot$, since ($\bot \sqcap C \equiv \bot$).

    \item If $C_2$ is $C_3 \sqcap C_2'$ and $\voc(C_3) \subseteq \voc(G)$, then, for $C_2'$ if $\voc(C_2') \not\subseteq \voc(G)$ one of the other cases applies recursively.
  \end{enumerate}

  We finally consider the cases for $C_1$ by induction.

  \begin{enumerate}
    \item If $C_1$ is $C_2$ and $\voc(C_2)\not\subseteq \voc(G)$, then, as previously shown, either $(\TBox_G, G) \models C_2 \equiv \top$ or $(\TBox_G, G) \models C_2 \equiv \bot$ and thus $(\TBox_G, G) \models C_1 \equiv \top$ or $(\TBox_G, G) \models C_1 \equiv \bot$, if $\voc(C_1)\not\subseteq \voc(G)$.

    \item If $C_1$ is $C_2 \sqcap C_1'$ and $\voc(C_2)\not\subseteq \voc(G)$, then either, as previously shown, $(\TBox_G, G) \models C_2 \equiv \top$ or $(\TBox_G, G) \models C_2 \equiv \bot$.

    In the first case, then also $(\TBox_G, G) \models C_1 \equiv \top$, since $\top \sqcup C \equiv \top$.

    In the second case, we can reduce the term to $C_1'$ (since $\bot \sqcup C \equiv C$), and by induction, either $\voc(C_1')\not\subseteq \voc(G)$ and then $(\TBox_G, G) \models C_1' \equiv \top$ or $(\TBox_G, G) \models C_1' \equiv \bot$, or $\voc(C_1')\subseteq\voc(G)$.

    \item If $C_1$ is $C_2 \sqcup C_1'$ and $\voc(C_2) \subseteq \voc(G)$, then for $C_1'$ if $\voc(C_1') \not\subseteq \voc(G)$ one of the other cases applies recursively.
  \end{enumerate}

  Hence, we prove the lemma.
\end{proof}

For \Cref{prop:empty-larger-two}, we slightly adapt the allowed concept descriptions by allowing existential quantification for $C_4$, thus, the concept descriptions now cover arbitrary $\DLogics$ concept descriptions in disjunctive normal form.

\begin{lemma}\label{prop:empty-larger-two}
  Let $G$ be a Simple RDF graph, $(\TBox_G, G)$ the validation knowledge base of $G$, and $C_5$ a concept description defined by the following grammar
  \begin{align} 
    C_5 & ::= C_6 \sqcup C_5 \mid C_6\\
    C_6 & ::= C_7 \sqcap C_6 \mid C_7\\
    C_7 & ::= \neg C_8 \mid C_8 \\
    C_8 & ::= \top \mid \bot \mid A \mid \{a\} \mid \exists p.C_5 \mid \exists p^-.C_5
  \end{align}
  where $A$ is a concept name and $a$ an individual name.
  Then, $\voc(C_5)\not\subseteq \voc(G)$ implies one of the following cases:
  \begin{enumerate}
    \item $(\TBox_G, G) \models C_5 \equiv \top$ or $(\TBox_G, G) \models C_5 \equiv \bot$
    \item There exists a concept description $(\TBox_G, G) \models C_5' \equiv C_5$, such that either $\voc(C_5')\subseteq\voc(G)$, or $\voc(C_5')\not\subseteq \voc(G)$ and $(\TBox_G, G) \models C_5 \equiv \top$ or $(\TBox_G, G) \models C_5 \equiv \bot$.
  \end{enumerate}
\end{lemma}

\begin{proof}
  We prove this property by induction on the structure of $C_5$ (\Cref{prop:empty-larger-two}).
  To this end, we consider first as a base cases the case where in $C_8$ existential quantification is restricted to $\exists p.C_1$ (or $\exists p^-.C_1$, respectively).
  According to \Cref{prop:empty-larger}, then the property under investigation holds for $C_1$.
  Furthermore, we assume without loss of generality, that $C_1$ is fully reduced according to \Cref{prop:empty-larger}.
  Thus, if $\voc(C_1)\not\subseteq\voc(G)$ then either $(\TBox_G, G) \models C_1 \equiv \top$ or $(\TBox_G, G) \models C_1 \equiv \bot$, since otherwise $C_1$ would not be fully reduced according to \Cref{prop:empty-larger}.

  We have the following cases if $C_8$ is $\exists p.C_1$ (cases for $\exists p^-.C_1$ work exactly equivalently):

  \begin{enumerate}
    \item If $p \not\in \voc(G)$, then $p^\Int$ is empty, from which we can follow that $(\TBox_G, G) \models \exists p.C_1 \equiv \bot$.
    Note, that this holds independently from $C_1$.
    \item If $p \in \voc(G)$ and $\voc(C_1) \not\subseteq \voc(G)$, then, by definition, either $(\TBox_G, G) \models C_1 \equiv \top$ or $(\TBox_G, G) \models C_1 \equiv \bot$.
    Thus, we can reduce the expression to $\exists p.\top$ or $\exists p.\bot$, respectively.
  \end{enumerate}

  Since the cases for $C_7$, $C_6$ and $C_5$ depend only on the common property between \Cref{prop:empty-larger-two} and \Cref{prop:empty-larger}, the proofs work exactly analogously to the proofs of $C_3$, $C_2$ and $C_1$ (\Cref{prop:empty-larger}), and are omitted for brevity here.

  Then, by induction, starting form the restricted $C_8$ as the base case, the property follows for arbitrary concept descriptions $C_5$.
  Thus, we prove the lemma.
\end{proof}

Finally, we prove \Cref{prop:relevant-shapes-ext}.

\begin{proof}[Proof of \Cref{prop:relevant-shapes-ext}]
  Let $s = \psi \sqsubseteq \phi$ be a $\DLogics$ SHACL shape, $q$ be a \sccqname, and $G$ be a Simple RDF graph with $\voc(G) \subseteq \voc(q)$, and $(\TBox_G, G)$ be the validation knowledge base of graph $G$. 
  We prove first property (i) of \Cref{prop:relevant-shapes-ext}.
  We have the following disjoint cases:

  \begin{enumerate}
  \item
    Case $\voc(\psi) \not\subseteq \voc(q)$.
    Then, by \Cref{prop:concept-empty}, $(\TBox_G, G) \models \psi \equiv \bot$ (since $\psi$ is, per definition, restricted to one of the cases covered in the lemma). 
    Hence, shape $\psi \sqsubseteq \phi$ is not relevant (\Cref{def:irrelevant-shape}).

  \item Case $\voc(\psi) \subseteq \voc(q)$ and $\voc(\phi) \not\subseteq \voc(q)$. 
  Let us assume, without loss of generality, that the concept description $\phi$ is fully reduced according to \Cref{prop:empty-larger-two}.
  (Note, that if due to reduction $\voc(phi) \not\subseteq\voc(G)$ no longer applies, the next case below would be applicable.)

  Then, according to \Cref{prop:empty-larger-two}, either $(\TBox_G, G) \models C_1 \equiv \top$ or $(\TBox_G, G) \models C_1 \equiv \bot$, in which case the shape is not relevant (\Cref{def:irrelevant-shape}).

  \item Case $\voc(\psi) \subseteq \voc(q)$ and $\voc(\phi) \subseteq \voc(q)$. In this case, property (i) is trivially satisfied.
  \end{enumerate}

  Thus we prove \Cref{prop:relevant-shapes-ext}.
\end{proof}

\begin{corollary}\label{cor:finite}
  As corollary of \Cref{prop:relevant-shapes-ext}, the set of $\DLogics$ shapes over $\voc(q)$ of a query $q$ is \emph{not} finite.
\end{corollary}

\begin{proof}
  Follows immediately by inspection of the syntax of a $\DLogics$ concept description over a finite vocabulary.
\end{proof}

However, we can further restrict the set of candidates to obtain a meaningful, finite set of shapes.
To this end, we first define the quantification \emph{nesting depth} as a property of an $\DLogics$ concept description.

\begin{definition}
  The nesting depth $\nest(C)$ is defined as:
  \begin{align}
    \nest(\exists p.C) &:= 1 + \nest(C)\\
    \nest(\forall p.C) &:= 1 + \nest(C)\\
    \nest(C_1 \sqcap C_2) &:= \max(\nest(C_1), \nest(C_2))\\
    \nest(C_1 \sqcup C_2) &:= \max(\nest(C_1), \nest(C_2))\\
    \nest(C) & := 0 \quad \textit{for all other cases}
  \end{align}
\end{definition}

\begin{example}
  The nesting depth $\nest(\exists p . A)$ is $1 + 0 = 1$.
  The nesting depth $\nest(\forall p.A \sqcap \exists p. (B \sqcup \exists p.C))$ is $\max(1 + 0, 1 + (\max(0, 1 + 0))) = 2$.
\end{example}

Then, we restrict the nesting depth of candidate $\DLogics$ SHACL shapes over the vocabulary $\voc(q)$ of a given query $q$ to the diameter of the variable connectivity graph $\vcg(q)$.

\begin{proposition}\label{cor:finite-two}
  Given a query $q$, the set of relevant $\DLogics$ SHACL shapes over $\voc(q)$ and with finite nesting depth is finite.
\end{proposition}

\begin{proof}
  Follows immediately by inspection of the syntax of a $\DLogics$ concept description over a finite vocabulary and \Cref{prop:relevant-shapes-ext}.
\end{proof}

}{
  \appendix
  \balance

}

\end{document}